\documentclass[11pt]{article}
\usepackage[utf8]{inputenc}
\usepackage[a4paper,left = 1in,right = 1in, top = 1.5in, bottom = 1.2in]{geometry}
\usepackage{amsmath}
\usepackage{amsthm}
\usepackage{palatino}
\usepackage{bm}
\usepackage{dsfont}
\usepackage{cancel}
\usepackage{amssymb}
\usepackage{commath}
\usepackage{mathtools}
\usepackage{float}
\usepackage{accents}
\counterwithin{equation}{section}

\usepackage[numbers]{natbib}
\usepackage{caption}
\usepackage{subcaption}
\usepackage{graphicx}
\usepackage{color}
\usepackage{tcolorbox}
\usepackage{tikz}
\usetikzlibrary{matrix}
\newtcolorbox{mybox}[3][]
{
colframe = black,
  colback  = black,
  coltitle = white,  
  coltext = white,
  title    = {#3},
  #1,
}
\usepackage{xcolor}
\usepackage{hyperref}
\hypersetup{colorlinks,linkcolor=blue,citecolor=blue,urlcolor = violet}
\usepackage[capitalise]{cleveref}
\usepackage[symbol]{footmisc}
\renewcommand{\thefootnote}{\fnsymbol{footnote}}
\usepackage{fancyhdr}
\fancypagestyle{firststyle}
{
   \fancyhf{}
   
   \fancyfoot[L]{\vspace{-6mm}  \text{Emails: } 
\href{mailto:bdupire@bloomberg.net}{\footnotesize \tt bdupire@bloomberg.net}, \ \href{mailto:v.tissot-daguette@princeton.edu}{\footnotesize \tt v.tissot-daguette@princeton.edu} \vspace{2cm} }
   
}
\usepackage{multicol}
\def\be{\begin{equation}}
\def\ee{\end{equation}}

\def\bea{\begin{align}}
\def\eea{\end{align}}

\def\bea*{\begin{align*}}
\def\eea*{\end{align*}}
\theoremstyle{plain}
\newtheorem{theorem}{Theorem}[section]

\newtheorem{lemma}[theorem]{Lemma}
\newtheorem{corollary}[theorem]{Corollary}
\newtheorem{proposition}[theorem]{Proposition}

\theoremstyle{definition} 
\newtheorem{definition}[theorem]{Definition}
\newtheorem{example}{Example}
\newtheorem{remark}[theorem]{Remark}
\DeclareMathOperator{\A}{\mathbb{A}}
\DeclareMathOperator{\B}{\mathbb{B}}
\DeclareMathOperator{\C}{\mathbb{C}}
\DeclareMathOperator{\D}{\mathbb{D}}
\DeclareMathOperator{\E}{\mathbb{E}}
\DeclareMathOperator{\F}{\mathbb{F}}

\DeclareMathOperator{\N}{\mathbb{N}}

\DeclareMathOperator{\Q}{\mathbb{Q}}
\DeclareMathOperator{\R}{\mathbb{R}}

\DeclareMathOperator{\calC}{\mathcal{C}}
\DeclareMathOperator{\calD}{\mathcal{D}}

\DeclareMathOperator{\calF}{\mathcal{F}}

\DeclareMathOperator{\calH}{\mathcal{H}}
\DeclareMathOperator{\calI}{\mathcal{I}}
\DeclareMathOperator{\calJ}{\mathcal{J}}
\DeclareMathOperator{\calK}{\mathcal{K}}
\DeclareMathOperator{\calL}{\mathcal{L}}
\DeclareMathOperator{\calM}{\mathcal{M}}

\DeclareMathOperator{\calO}{\mathcal{O}}

\DeclareMathOperator{\calS}{\mathcal{S}}

\DeclareMathOperator{\calV}{\mathcal{V}}

\DeclareMathOperator{\frakF}{\mathfrak{F}}

\DeclareMathOperator{\frakI}{\mathfrak{I}}

\providecommand{\keywords}[1]
{
  \small	
  \textbf{\textit{Keywords ---}} \textit{#1}
}

\date{\vspace{-1em}\normalsize{\today}}

\usepackage{authblk}
\title{\vspace{-25mm} Functional Expansions \vspace{-3mm}}
\date{\today}
\author[$\star$]{Bruno Dupire}
\author[ $\star$,$\diamond$]{Valentin Tissot-Daguette}
\affil[$\star$]
{\footnotesize Quantitative Research, Bloomberg L.P., New York, NY, USA}
\affil[$\diamond$]
{\footnotesize ORFE, Princeton University, Princeton, NJ, USA\vspace{-2mm}}
\begin{document}
\maketitle
\renewcommand{\thefootnote}{\arabic{footnote}}
\thispagestyle{firststyle}

\vspace{-9mm}
\begin{abstract}

Path dependence is omnipresent in many disciplines such as  engineering, system theory and finance. It reflects the influence of the past on the future, often expressed through 
functionals.  However,  non-Markovian  problems 
are often infinite-dimensional, thus challenging from a conceptual and computational perspective. 
In this work, we shed light on   expansions of  functionals. 
First, we treat static expansions  made around paths of fixed length and propose a generalization of the Wiener series$-$the intrinsic value expansion (IVE). 
In the dynamic case, we revisit the functional Taylor expansion (FTE).  The latter connects the functional Itô calculus with the   signature to quantify the effect in a functional when a 
“perturbation" path is concatenated with the source path. 
In particular, the FTE elegantly separates the  functional from future trajectories. 
The notions of real analyticity and radius of convergence  are also  extended to the path space. 
We discuss other dynamic expansions arising from Hilbert projections and the Wiener chaos, and finally  show financial applications of the FTE to the pricing and hedging of exotic contingent claims.

\end{abstract}
\vspace{1mm}
	
\keywords{Functional It\^o Calculus, Taylor Expansion, Path Signature, 
Wiener Chaos, Contingent Claims}
\\ \vspace{-3mm}

\textit{\textbf{MSC (2020) Classification —} 
41A58, 
91G20,  
26E15,  
60L10 
}

\setcounter{tocdepth}{2}
\tableofcontents
\section{Introduction}

Traditional problems in many  disciplines rely on the strong hypothesis that the future solely depends on the present. This is the Markov property, which tremendously simplifies the study of causal relationships  
in physics, biology, or social science. 
In 
 finance, 
it  allows to express the price and hedge of vanilla options in terms of the spot value of the underlying. 
This is often no longer the case when  the payoff becomes path-dependent (also called exotic) or the stock dynamics is more complex. The lack of Markovian representation can have  serious consequences from a computational perspective as the problem is usually infinite-dimensional.  

In some cases, a Markovian framework can be recovered by enlarging the state process accordingly.  
As an illustration, consider an at-the-money  lookback option, i.e. 
$g(X_T) = (\max_{0 \le u\le T} x_u - x_0 )^{+}$  
 with
underlying $X$ and maturity $T$. 
If $X$ is Markov and  $Y := (\max_{0 \le u\le t} x_u)_{t\in [0,T]}$, then $(X,Y)$ is also Markov under mild assumptions. 
 In turn, the price of the lookback option at  $t\in [0,T]$ will be function of $ (t,x_t,y_t)$ only. 
 The same remedy may be applied when 
 the dynamics of the stock depends on past information, e.g. 
 path-dependent volatility \cite{Guyon}. Taking 
 the Hobson-Rogers model 
 \cite{HobsonRogers}, the underlying becomes Markov when the offset  processes$-$capturing historical trends of the stock$-$are added to the state. 
 
 Nevertheless, there are situations where  path dependence cannot be absorbed by finitely many features. 
 We characterize such frameworks as \textit{fully non-Markovian}.  
 Among others, 
 we can mention 
 American options whose reward depends on 
 a  moving average of the underlying  \cite{Bernhart}.  
 Also, the pricing problem in rough volatility models belong to this category, even for the simplest instrument. Either way, the option price becomes a \textit{functional} of the path so far, calling for the development of new mathematical  tools.  


One avenue consists of  employing \textit{series expansions} to  project exotic payoffs onto a finite-dimensional subspace arising, e.g.,  from  
orthogonal  polynomials \cite{Bernhart}, the Wiener chaos  \cite{Lelong,Neufeld}, or  the Karhuhnen-Loève expansion \cite{Tissot}.   
Series expansions can also be applied to approximate stock processes with  fully non-Markovian dynamics; see, e.g., \cite{Chevalier}. 
Similar in spirit, the \textit{path signature}  showed promising results 
in option pricing and financial modeling problems  entailing path dependence \cite{Arribas,CuchieroSig,LyonsNum}. 
These works motivate the present paper  as the use of the signature 
is justified by  a deep yet overlooked result: the \textit{functional Taylor expansion} (FTE). First proposed by \citet{Fliess81,Fliess83,Fliess86}, it is a generalization of the  Taylor expansion where  classical derivatives and monomials are respectively replaced by functional derivatives and the signature. Further  historical notes and details will be given at the beginning of \cref{sec:DynamicExpansions}. 



\textbf{Outline.} We devote this study to the \textit{expansions} of  functionals by retracing   historical works and sharing new  findings. In \cref{sec:static},  we collect 
static 
expansions 
in the sense that the path length remains unchanged. 
Starting in the early 1900s, \citet{Volterra} expressed  a functional defined on a fixed horizon $[0,T]-$called $T-$\textit{functional} and denoted by $g$ in this work$-$as combination of homogeneous polynomials in integral form: If $X,Y$ are  paths of length $T$, then one can write 
\begin{equation} \label{eq:Volterra0}
    g(X+Y) = \sum_{k < K} \int_{0}^T \int_{0}^{t_k}\cdots \int_{0}^{t_2} \psi_k(t_1,...,t_k) y_{t_1} \cdots y_{t_k}dt_1 \cdots dt_k + R_K(X,Y). 
\end{equation} 
The kernels $(\psi_k)$ consists of higher-order Fr\'echet derivatives of $g$ evaluated at $X$; 
see \cite{Volterra,PalmPoggio,Boyd} and  \cref{sec:Volterra}. 
As can be seen, Volterra's approach relies on pointwise values of the added path $Y$.  
In contrast, the Wiener series and chaos expansion \cite{Wiener38,Wiener58} encode a 
path in terms of its infinitesimal increments. This gives, formally, the expansion 
\begin{equation} \label{eq:Wiener0}
    g(X + Y) = \sum_{k < K} \int_{0}^T \int_{0}^{t_k}\cdots \int_{0}^{t_2} \phi_k(t_1,...,t_k) dy_{t_1} \cdots dy_{t_k}+R_K(X,Y). 
\end{equation} 
In the chaos expansion, $X,Y$ are typical Brownian paths and as shown in  \citet{Stroock}, the kernels $(\phi_k)$ are expectations of iterated Malliavin derivatives of $g$. 
It turns out that a pathwise expansion can be obtained 
 by computing the Malliavin derivatives of the stopped path 
and employing Stratonovich integrals instead of Itô ones. 
This leads to our first contribution:  the \textit{intrinsic value expansion}, set forth in \cref{thm:IVE}. To the best of our knowledge, neither the intrinsic value expansion nor a variation of it has been discussed
in the literature. 

Alternatively,  one can study the local behavior of a "running" functional$-$denoted by $f-$when extending the source path.  This is the core 
of \cref{sec:DynamicExpansions}, concerned with \textit{dynamic} 
expansions. Let us give some insight.  
In classical calculus, say on the real line, the Taylor expansion reads, 
\begin{equation} \label{eq:taylor1D}
  f(x+y) = \sum_{k<K} f^{(k)}(x) \frac{y^k}{k!} + R_K(x,y), \quad x,y\in \R, \quad f\in   \calC^{K}(\R). 
\end{equation}
The shock $(y)$ is fragmented into its scaled powers $(\frac{y^k}{k!})$,  
weighted by the sensitivities of the function at the initial point $(x)$. 
In the path space, a similar decomposition can be formulated:  If $X,Y$ are paths  of arbitrary lengths, then 
\begin{equation} \label{eq:FTE0}
    f(X \oplus  Y) = \sum_{|\alpha| < K} \Delta_{\alpha}f(X) \calS_{\alpha}( Y) + R_K(X,Y),
\end{equation} 
where $\oplus$ concatenates $X$ and $Y$. This is the \textit{functional Taylor expansion} (FTE)$-$our main object of interest in this work$-$stated  in  \cref{thm:FTE}. 
Let us  describe the right side of $\eqref{eq:FTE0}$ while keeping the level of technicality to a minimum in this introduction. For that reason, we also postpone the literature review to the beginning of  \cref{sec:FTE}. The first summands $(\Delta_{\alpha}f(X))$ are higher order functional  derivatives of $f$, while  $(\calS_{\alpha}(Y))$ are iterated integrals of  $Y$, also called \textit{signature functionals}. 
The indexes $\{\alpha = \alpha_1\cdots \alpha_k \ : \ k < K \}$ 
 concurrently specify the order of differentiation for $\Delta_{\alpha}f$ (with respect to the time or space variable) and integration for $\calS_{\alpha}$. Unlike the Volterra   and Wiener expansion where the functional  and perturbation path $Y$ are quite entangled (see \eqref{eq:Volterra0}, \eqref{eq:Wiener0}), the FTE provides a perfect separation between them.  

\cref{fig:exp2D} compares the types of expansions 
discussed above. As can be seen, static (respectively classical) expansions are made  
\textit{around} a path (resp. a point), whereas dynamic expansions describe the evolution of a functional \textit{after} the initial path.  
\begin{figure}[H]
\caption{Classification of expansions.}
\vspace{-2mm}
\begin{subfigure}[b]{0.33\textwidth}
    \centering
    \caption{Classical (around a point)}
\includegraphics[height=1.6in,width=1.85in]{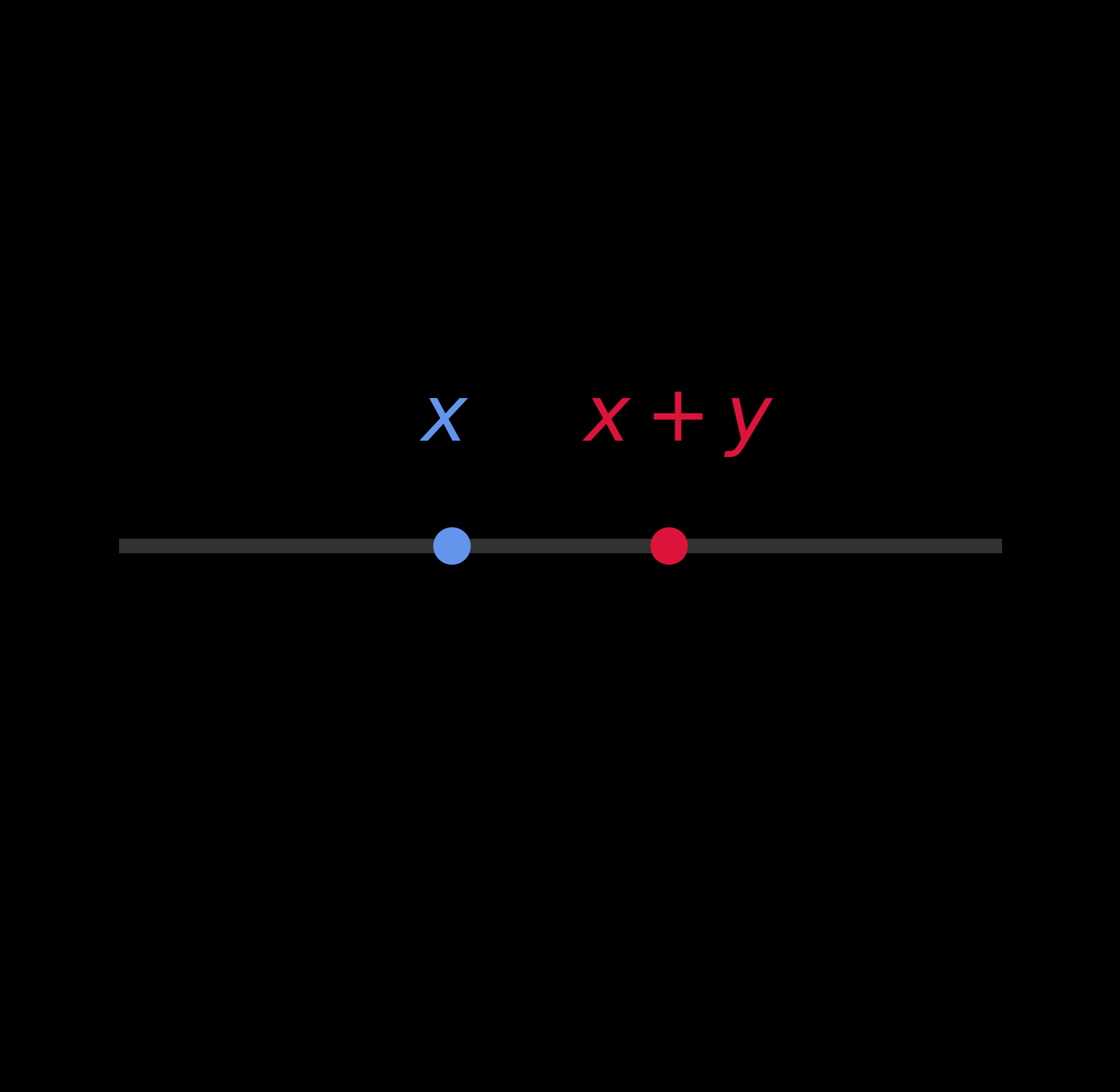}
    \label{fig:classical}
    \end{subfigure}
\begin{subfigure}[b]{0.33\textwidth}
    \centering
    \caption{Static (around a path)}
    \includegraphics[height=1.6in,width=1.85in]{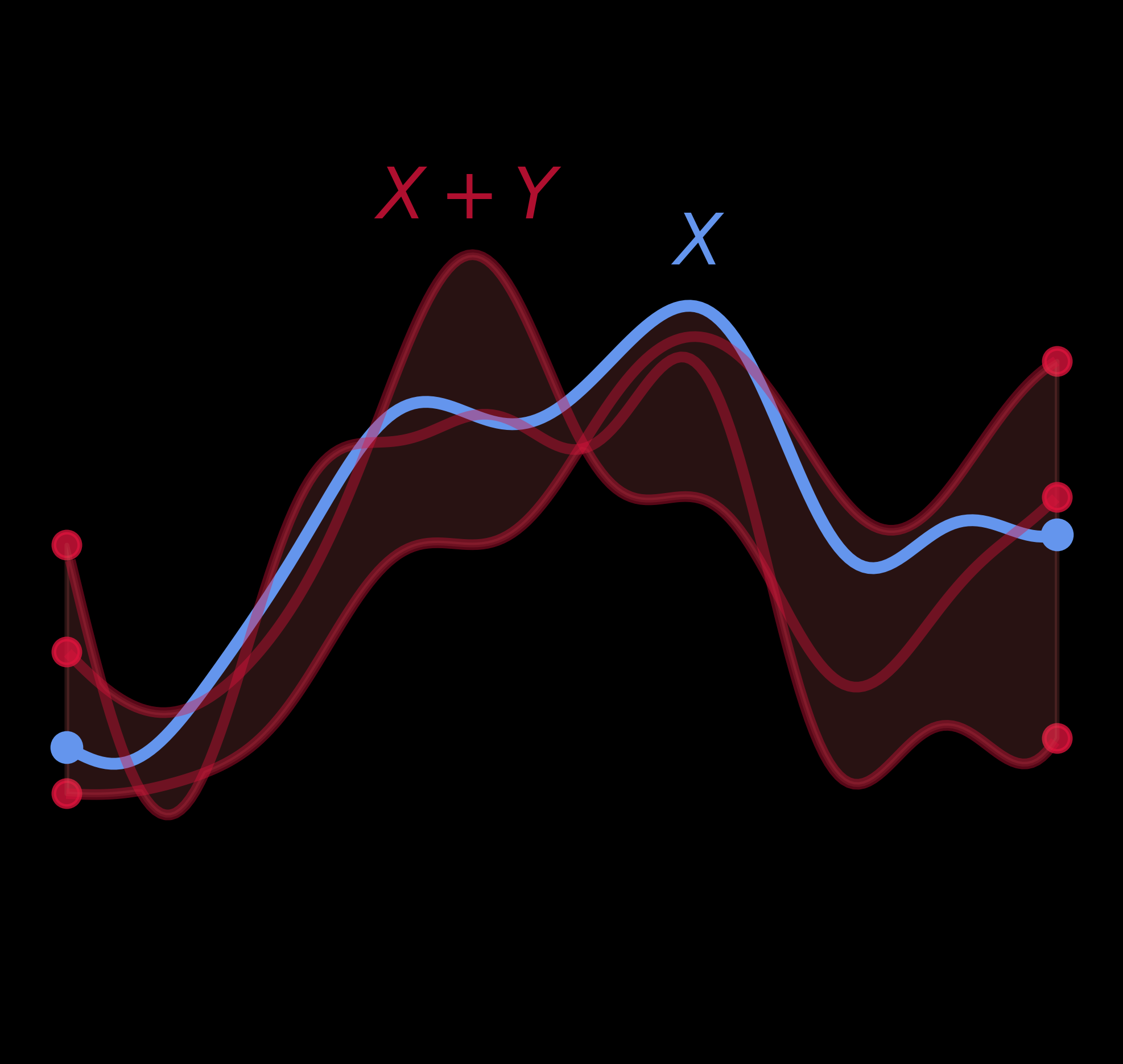}
    \label{fig:static}
\end{subfigure}
\begin{subfigure}[b]{0.33\textwidth}
    \centering
    \caption{Dynamic (after a path)}
    \includegraphics[height=1.6in,width=1.85in]{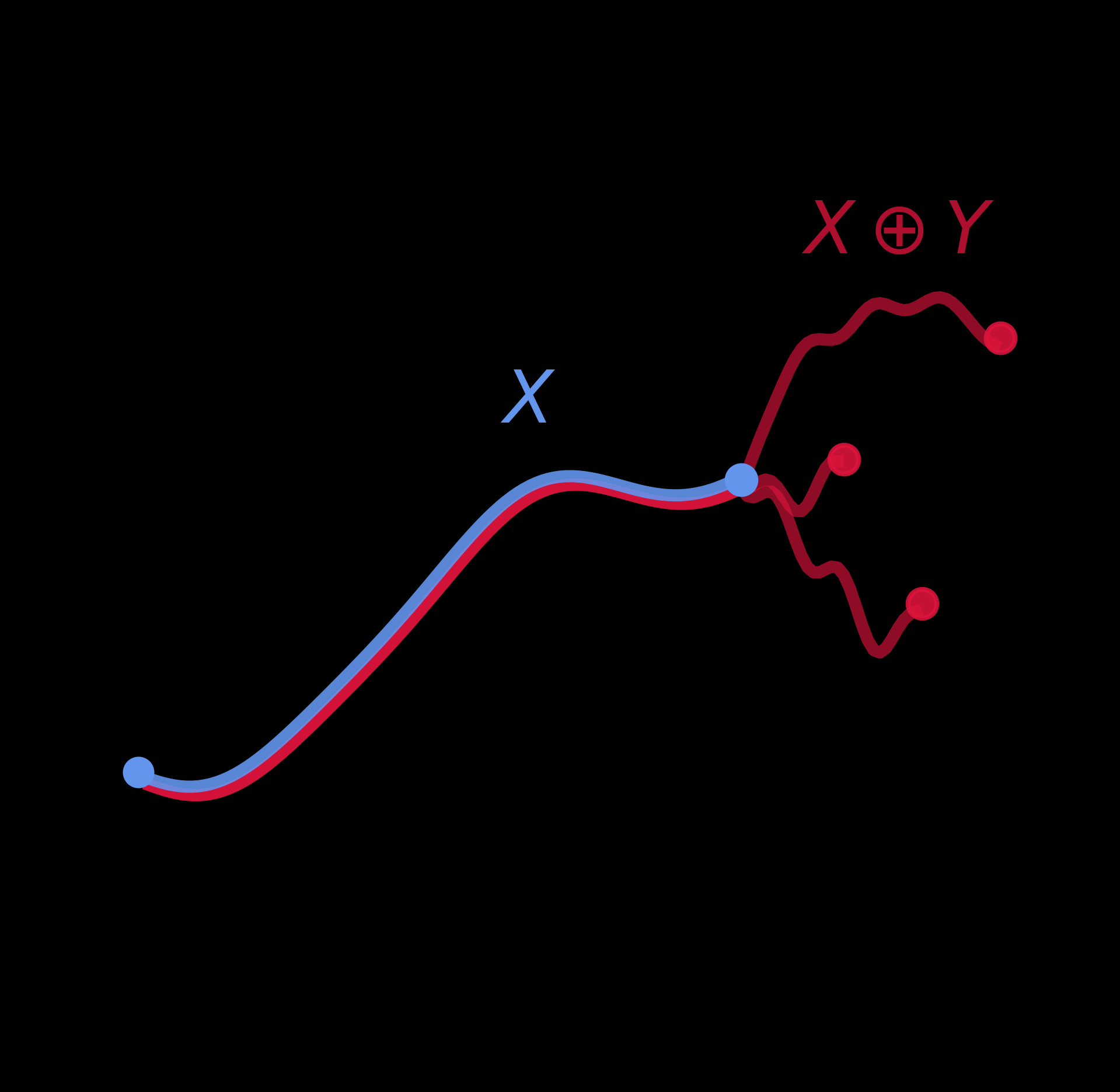}
    \label{fig:dynamic}
\end{subfigure}
\vspace{2mm}
\label{fig:exp2D}
\end{figure}
\vspace{-8mm}

 The paper is organized as follows. In \cref{sec:static}, we examine static expansions   and derive the novel intrinsic value expansion (IVE). 
 \cref{sec:FTE} puts  
 a spotlight on the functional It\^o calculus and  path signature, culminating with the functional Taylor expansion (FTE). 
 We also derive remainder estimates of the FTE and characterize the corresponding radius of convergence.  Connections are finally  established between the FTE and  static expansions. 
 In \cref{sec:otherExp}, we discuss other dynamic expansions arising from Hilbert projections and the Wiener chaos. We present the pricing and dynamic hedging  of exotic options in \cref{sec:FTEApp} 
 as possible applications of the FTE. 
 \cref{sec:conclusion} concludes the study and  
 \cref{app:Proofs} contains postponed proofs of the main results.

\textbf{Notations and Definitions.} We employ  the framework from  the functional Itô calculus \cite{Dupire}. Fix throughout a horizon $T>0$, which can be interpreted as the maturity of a financial derivative. Let $\Lambda_t = \calD([0,t],\R)$ be the Skorokhod space of c\`adl\`ag paths of length $t\in [0,T]$. 
Given $X_t\in \Lambda_t$, $X_s$ denotes the whole trajectory up to time $s\le t$, while $x_s= X_t(s)$ is the value at time $s$. 
Moreover, write $\Lambda := \bigcup_{t\in[0,T]}\Lambda_t$ for the collection of all c\`adl\`ag paths.  
\begin{definition}
A \textit{functional} is a map $f:\Lambda \to \R$, while  a  $T-$\textit{functional} refers to any map $g:\Lambda_T \to \R$.    
\end{definition}
The distance between elements of $\Lambda$ is measured according to 
\begin{equation} \label{eq:dLambda}
    d_{\Lambda}( X_t, Y_s) =t-s + \sup_{0 \le u \le t} |x_u - y_{u\wedge s} | , \quad 0 \le s \le t \le T.
\end{equation}
  We say that a functional  is $\Lambda-$\textit{continuous}  if it is continuous with respect to $d_{\Lambda}$. The $\Lambda_T-$continuity of $T-$functionals is defined similarly, where we remark that $d_{\Lambda}$ coincides with the uniform topology. We now collect important notions: 
  \begin{enumerate}
      \item[\textbf{(i)}] 
We recall the spatial and temporal \textit{functional derivatives} introduced in \cite{Dupire}, namely 

\setlength{\columnsep}{6.5cm}
\begin{multicols}{2}
\begin{align*}
    \Delta_x f(X_t) &= \lim_{h\to 0} \, \frac{f(X^{h}_t)-f(X_t)}{h}, \ &&X^{h}_t(s) \,\, =  x_s + h\,\mathds{1}_{\{s\,=\,t\}}, \quad  s\le t, 
   \\[4.1em]
      \Delta_t f(X_t) &= \lim_{\delta t\downarrow 0} \, \frac{f(X_{t,\delta t})-f(X_t)}{\delta t}, \
    &&X_{t,\delta t}(s) = x_{s \wedge t}, \hspace{1.02cm} s\le t+\delta t,
\end{align*}
\columnbreak

\vspace{11mm}
    \begin{figure}[H]
        \centering \includegraphics[width = 1.1in, height = 0.9in]{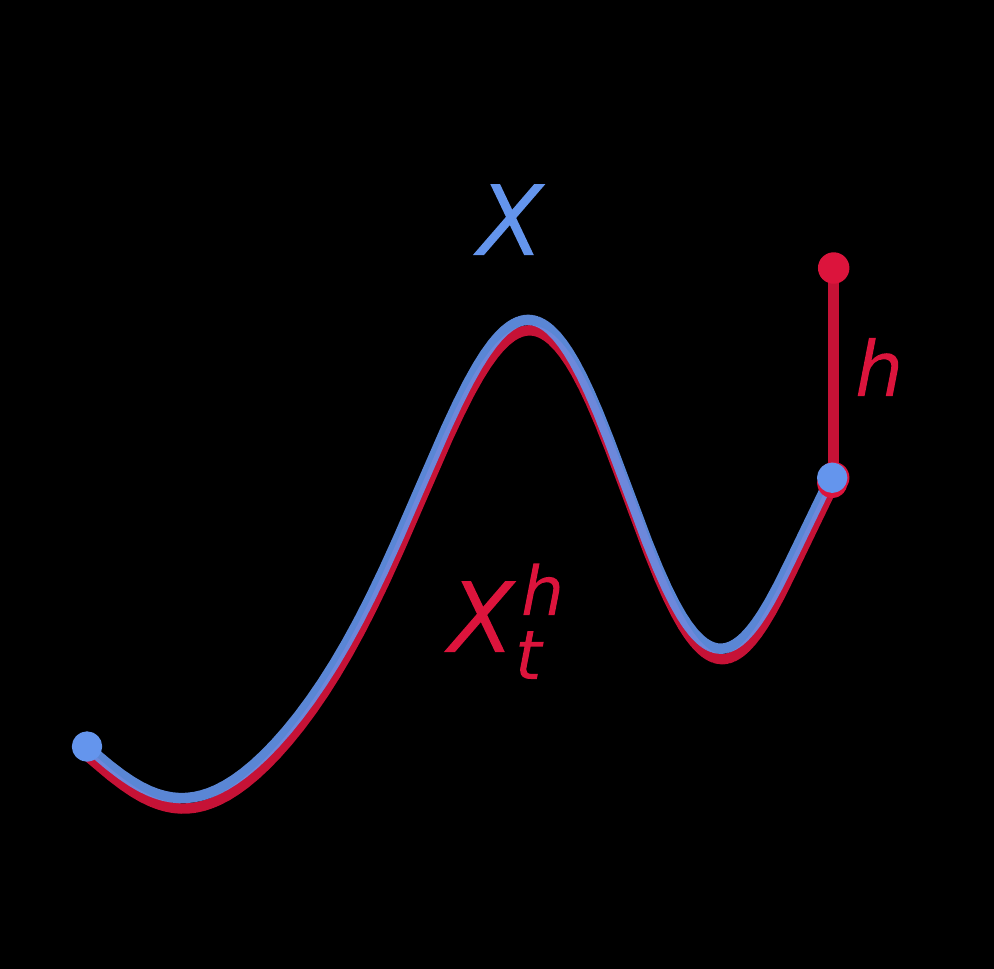}
    \end{figure}

\vspace{-8mm}
    \begin{figure}[H]
        \centering \includegraphics[width = 1.1in, height = 0.9in]{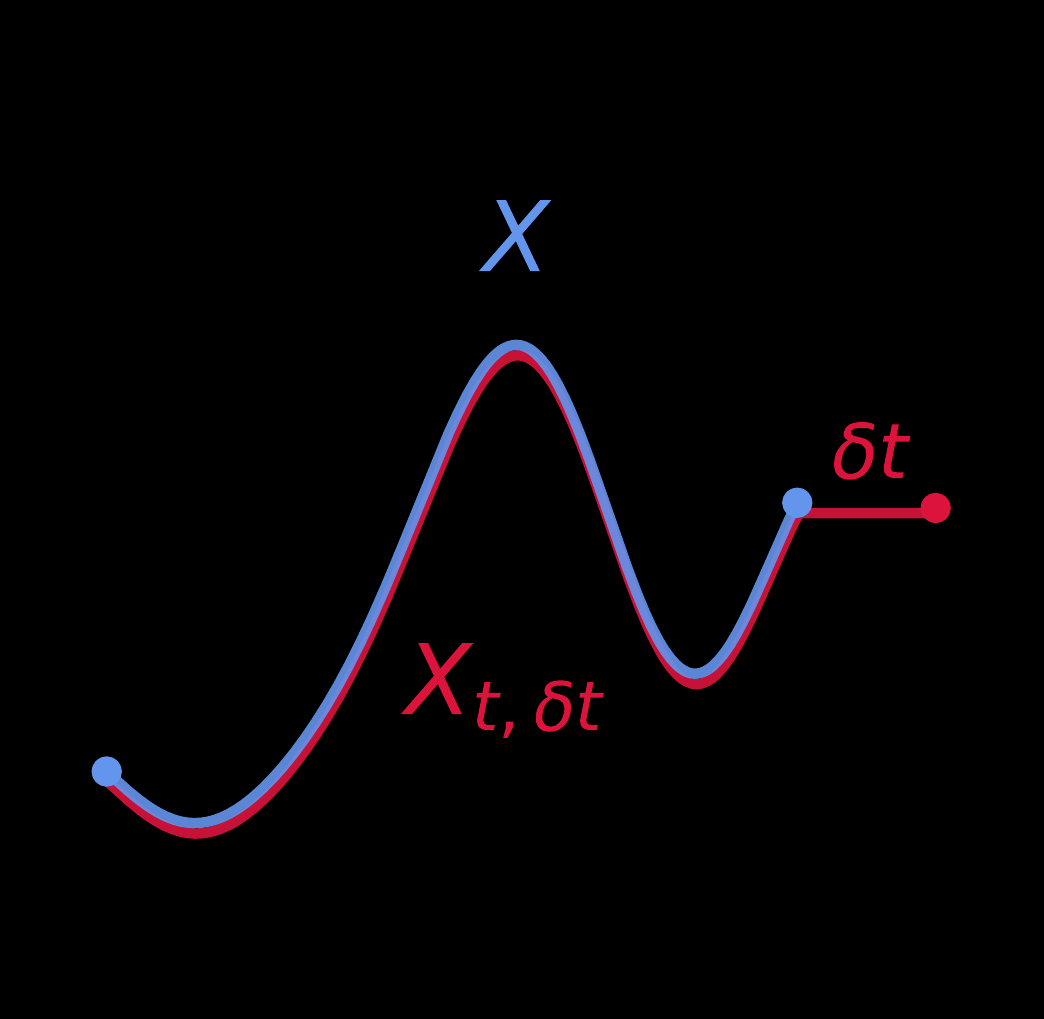}
    \end{figure} 
\end{multicols}
\vspace{-2mm}

whenever these quantities exist. 
We note that $\Delta_x$ and $\Delta_t$ need not commute; see  \cref{ex:IntDev}. 
      
      \item[\textbf{(ii)}] Let $\oplus: \Lambda^2 \to \Lambda$ be the noncommutative binary operation that concatenates paths in a continuous fashion. In other words,   if $X\in \Lambda_{s}$, $ Z \in \Lambda_{u}$ then $Y = X \oplus Z \in \Lambda_{s+u}$ is given by 
$$y_r  = x_{r \wedge s} + z_{(r- s)^+}- z_0, \quad  r \in [0,s+u].$$
 If 
 $u > T-s$, we define instead $X \oplus Z = X \oplus Z_{T-s} \in \Lambda_T$.
      
      \item[\textbf{(iii)}] 
If $\Q$ is a measure  on $\Lambda$ and $f:\Lambda \to \R$, we write $\E^{\Q}[f(Y_t) \ | \ X_s]$ for the   \textit{conditioned expectation} of $f(Y_t)$ given that $Y_t$ coincides with $X_s$ up to $s<t$. In other words, 
 \begin{equation}\label{eq:conditionedExp}
     \E^{\Q}[f(Y_t)\,|\, X_s]=  \int_{\Lambda_{t-s}}f(\underbrace{X_s \oplus Z}_{=\ Y_t}) \Q(dZ),\quad  s \le t.
 \end{equation}
To avoid ambiguity, we favor 
the letters $Y,\ Z$ over
$X$ to represent random or future paths.
      
  \end{enumerate}
  
\section{Static Expansions} 
\label{sec:static}

Functional expansions are traditionally made for paths of fixed length. Without loss of generality, we restrict ourselves to  paths  
in $\Lambda_T$. 
Contrary to $\Lambda$, the restriction  $\Lambda_T$ has the advantage of being a vector space: we can add paths, have a clear sense of directions, etc. 
For simplicity, we often express static expansions around the null path of length $T$.

\subsection{$T-$functionals and Embeddings} \label{sec:TFuncEmbedd}
It proves  helpful to make 
 the distinction between functionals defined on $\Lambda$ and $T-$functionals (those restricted to $\Lambda_T$).    
If $X \in \Lambda_T$ is a stock price  trajectory, then $g:\Lambda_T \to \R$ can be viewed as the payoff of a claim with maturity $T$. Depending on the term sheet of the claim, it is a priori unclear  
how to  quantify the value of the payoff having observed the price path only  up to some intermediate time $t<T$ only. 
We here address this question by 
embedding $T-$functionals into the space of "running" functionals. 
\begin{definition}\label{def:embedding}
Let $\frakF$ be a set of functionals and $\frakF_T$ is the restriction of $\frakF$ to $\Lambda_T$.   
Then an \textit{embedding} is an operator $\iota: \frakF_T \hookrightarrow \frakF$ such that  $f = \iota  g \in \frakF$ 
satisfies  $f|_{\Lambda_T} = g$ for all $g \in \frakF_T$. 
\end{definition}

Notice that many embeddings may exist as the only constraint is that the embedded functional coincides   with a precribed $T-$functional on $\Lambda_T$. 
We here focus our attention on the following class of embeddings which resonates with many meanings in finance. 
Let $\Q$ be a  probability measure on $\Lambda$ and define 
\begin{equation} \label{eq:embedding}
    \iota_{\Q}   g(X_t) := \E^{\Q}[g(Y_T) \,|\, X_t ],
\end{equation}
  with  the conditioned expectation $\E^{\Q}[ \ \cdot \ | \ \cdot \ ]$  given in \eqref{eq:conditionedExp}. In \cref{def:embedding}, we may choose $\frakF_T$, $\frakF$ to be the set of uniformly bounded $(T-)$functionals so that clearly $\iota_{\Q}g \in \frakF$ and is well-defined. Clearly, $\iota_{\Q}g(X_T) = g(X_T)$ so that $\iota_{\Q}$ is indeed an embedding  for every $\Q$.  We stress that $X_t$ need not be in the support of $\Q$ as the latter is only introduced to extrapolate the future.    
In  financial terms and  provided that  $\Q$ is a risk-neutral measure,  $\iota_{\Q}g$ is the \textit{price functional} of the claim $g$  under the 
"model" $\Q$. Taking for instance the measure $\Q_\sigma$ such that the scaled canonical process $Y/ \sigma$ is 
Brownian motion under $\Q_\sigma$, $\sigma >  0$, then
$\iota_{\Q_\sigma}g$ gives the Bachelier price of $g$   with zero interest rates and volatility $\sigma$.  
An important case arises when $\sigma \downarrow 0$, leading to the \textit{intrinsic functional},
 \begin{equation}\label{eq:IV}
     \iota_0 g(X_t) :=  \iota_{\Q_0} g(X_t) = g(X_{t,T-t}).
 \end{equation}
The intrinsic functional is central in \cref{sec:IVE} where we introduce  the
intrinsic value expansion. 

\subsection{Volterra Series} \label{sec:Volterra}

Vito Volterra (1860-1940)  was arguably the first mathematician to introduce and study   functional expansions.  
His findings have had tremendous influence in nonlinear systems, where the Volterra expansion allows to express the solution functional linking the input path to the output (see, e.g., \cite{Volterra,PalmPoggio,Boyd}).  
His idea was to approximate a Fr\'echet differentiable $T-$functional $g:\Lambda_T \to \R$ by a multivariate 
function taking
pointwise values of the input path.

For simplicity, suppose that  $x_0=0$ and consider  the regular partitions $\Pi^N = \{ s_n = \frac{nT}{N},  \ n=0,\ldots,N \}$, $N\ge 1$. 
Then  $g(X_T)$ is estimated by 
$g^N(x_{s_1},...,x_{s_{N}}) = g(X_T^N)$ with the piecewise constant path     
 $X_T^N = \sum_{n=0}^{N-1} x_{s_n}\mathds{1}_{[s_n,s_{n+1})} + x_{T}\mathds{1}_{\{T\}}.$ 
If $g^N$ is real analytic and $\boldsymbol{0} = (0,...,0)\in \R^{N}$, a Taylor expansion of $h \mapsto g^N(h \  x_{s_1},...,h \  x_{s_N}) $ 
around $h =0$ gives 
\begin{align}\label{eq:Volterra1}
   g^N(x_{s_1},\ldots ,x_{s_N}) 
   &= g^N(\boldsymbol{0}) + \sum_{n=1}^N \partial_{n}g^N(\boldsymbol{0}) x_{s_n} + \sum_{1 \le n\le m \le N} \partial_{nm}g^N(\boldsymbol{0}) \,  x_{s_n}  x_{s_m} + \ldots
\end{align}
Put differently, $\eqref{eq:Volterra1}$  is an expansion of $g(X_T^{N})$ around the flat path $X_{0,T}$. 
Now  consider the
Fr\'echet derivatives of $g$  
in the direction of the Dirac mass  $\delta(\cdot - t)$, $t\in [0,T]$, 
i.e.  the family 
$(F_t)_{t\in [0,T]}$ such that
$$D_{Y_T}g(X_T) :=  \lim_{h \to 0}\frac{g(X_T + h Y_T) - g(X_T)}{h } = \int_0^T F_tg(X_T) y_t dt,  \quad Y_T \in \Lambda_T, $$
 when the above is well-defined. 
From the definition of $g^N$, choose $Y_T  = \mathds{1}_{[s_n,s_{n+1})}$ to obtain 
\begin{equation}\label{eq:psi1}
    \partial_{n}g^N(\boldsymbol{0}) = D_{\mathds{1}_{[s_n,s_{n+1})}}g(X_{0,T}) = \int_{s_n}^{s_{n+1}} F_t g(X_T)  dt.
\end{equation}
Similarly, defining  the iterated Fr\'echet derivatives $F_{t_1\cdots t_k} = F_{t_1} \cdots F_{t_k}$  yields  $$\partial_{n_1 \cdots n_k}g^N(\boldsymbol{0}) =  \int_{s_{n_k}}^{s_{n_k+1}}\cdots \int_{s_{n_1}}^{s_{{n_1}+1}} F_{t_1\cdots t_k}g(X_{0,T})  dt_{1}\cdots dt_{k}.$$
Owing to  $\eqref{eq:psi1}$ and the fact that $X^N$ equals $x_{s_n}$ on $[s_n,s_{n+1})$ for all $n\le N$, 
we can rewrite $\eqref{eq:Volterra1}$ as  
\begin{align*}\label{eq:Volterra2}
   g(X^N_T) 
   &=  g(X_{0,T})  + \int_{0}^{T} F_t g(X_{0,T}) x_t^N dt + \int_{0}^{T}\int_0^{t_2} F_{t_1 t_2} g(X_{0,T}) x_{t_1}^N x^N_{t_2} dt_1 dt_2 + \ldots
\end{align*}
Letting $N\to \infty$ gives, at least  formally, the \textit{Volterra series}, 
  \begin{equation}
       g(X_T) =   g(X_{0,T}) + \sum_{k= 1}^{\infty} \int_{\triangle_{k,T}}  F_{t_1 \cdots  t_k} g(X_{0,T}) \, x_{t_1} \cdots x_{t_k}\, dt^{\otimes k}, \label{eq:Volterra}
  \end{equation}
 with the notation $dt^{\otimes k} = dt_1 \cdots dt_k$ and  the simplexes,
\begin{equation}
    \label{eq:simplex}
 \triangle_{k,t} = \{(t_1,\ldots,t_k)\in [0,t]^k\ : \ t_1 \le \cdots \le t_k\}, \quad  k\ge 1, \quad  t\in [0,T]. 
\end{equation}

If $\psi_0 = g(X_{0,T}) $ and  $\psi_k(t_1,\ldots,t_k) =  F_{t_1 \cdots  t_k} g(X_{0,T})$ denotes the \textit{Volterra kernels}, then $\eqref{eq:Volterra}$ can be written more compactly as 
     $g(X_T)  = \psi_0+ \sum_{k= 1}^{\infty}V_k \psi_k(X_T), $
 with the operators 
 $$V_k \psi_k(X_t) = \int_{\triangle_{k,t}}  \psi_k \, x^{\otimes k} dt^{\otimes k} , \quad X_t \in \Lambda_t.$$
 Note that $V_k \psi_k$ is a \textit{homogeneous  functional of degree $k$} in the sense that $V_k\psi_k(\gamma X_t) = \gamma^k V_k \psi_k( X_t) $ with $\gamma X_t := (\gamma x_s)_{s\in [0,t]}$,  $\gamma \ne 0$. For instance, $V_1\psi_1 $ is a linear functional,  $V_2 \psi_2$ is homogeneous quadratic, and so on. 
 By construction, we see that  the kernels 
 can be retrieved from Dirac impulses of $g(X_T)$ at specific dates. \cref{fig:Volterra} shows the approximation of $\psi_1$ as in $\eqref{eq:psi1}$ obtained by bumping a single 
 value  of the discretized path.

\begin{example}
 \label{ex:nullfunctional}
Let $g:\Lambda_T\to \R$ with Volterra expansion $g(X_T) = \psi_0+ \sum_{k= 1}^{\infty}V_k \psi_k(X_T)$. Consider the embedding (see \cref{sec:TFuncEmbedd}) given by $f(X_t) = \psi_0+ \sum_{k= 1}^{\infty}V_k \psi_k(X_t)$. In other words, we 
truncate  the simplexes of each homogeneous functional $(V_k\psi_k)$. It turns out that $f$ coincides with the  $f^0(X_t) := g(X^0_{t,T-t})$ 
where $X^0_{t,T-t}\in\Lambda_T$ consists of the path $X_t$ glued with the null path of length $T-t$; see \cref{fig:VolterraEmbedding}.
Indeed, observe that $V_k \psi_k(X^0_{t,T-t})$ and $V_k \psi_k(X_{t})$ are equal as the integrand of the former 
is zero as soon as $t_k>t$. 
We can thus regard $f^0$ as the natural embedding of $g$ associated to the Volterra series.  
\end{example}

\begin{figure}[H]
\caption{Path deformation  in  Volterra and Wiener series.}
\vspace{-2mm}
\begin{subfigure}[b]{0.49\textwidth}
    \centering
    \caption{Volterra (Fr\'echet  
    derivative)}
    \includegraphics[height=1.9in,width=2.6in]{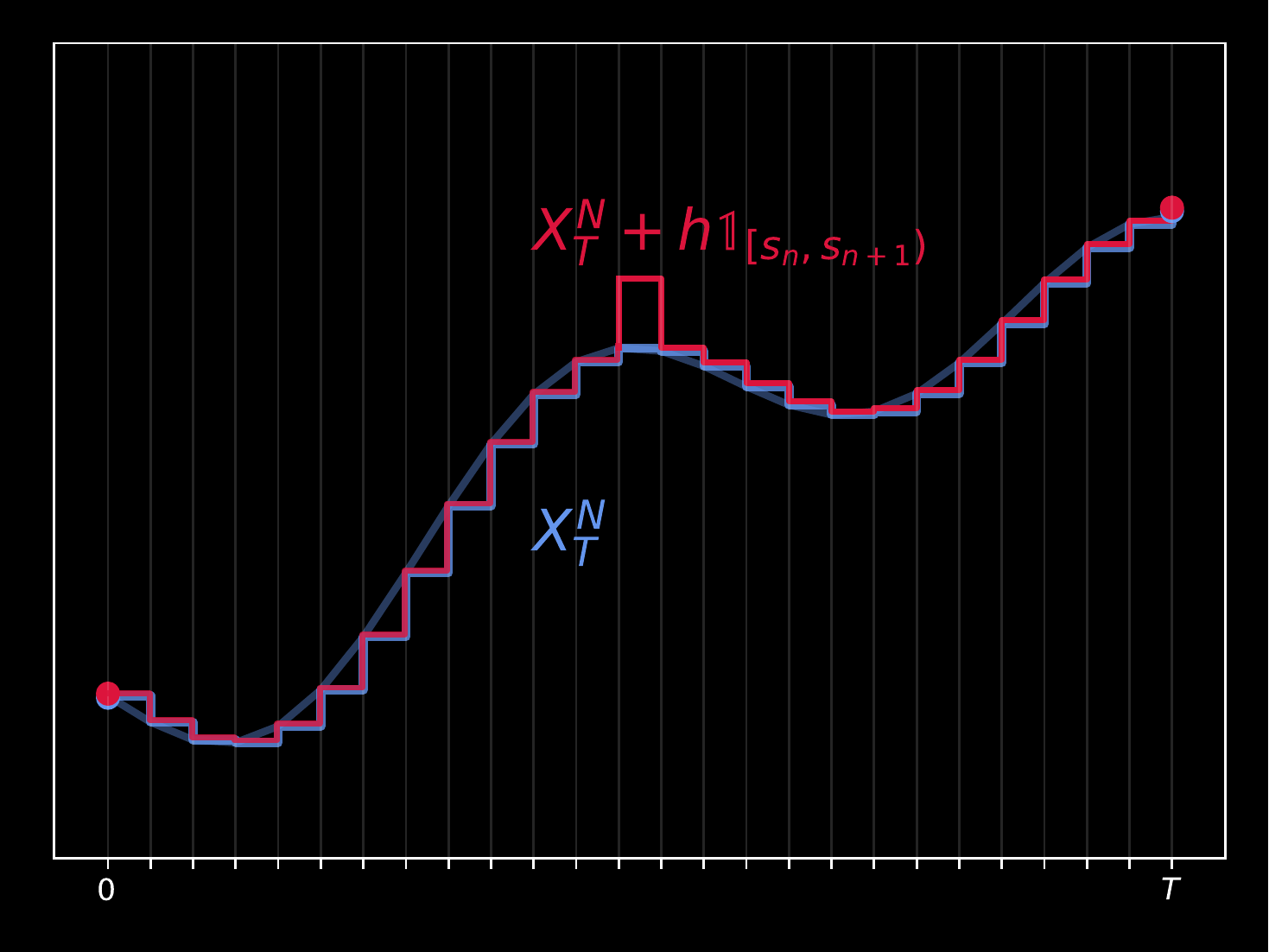}
    \label{fig:Volterra}
\end{subfigure}
\begin{subfigure}[b]{0.49\textwidth}
    \centering
    \caption{Wiener (Malliavin derivative)}
    \includegraphics[height=1.9in,width=2.6in]{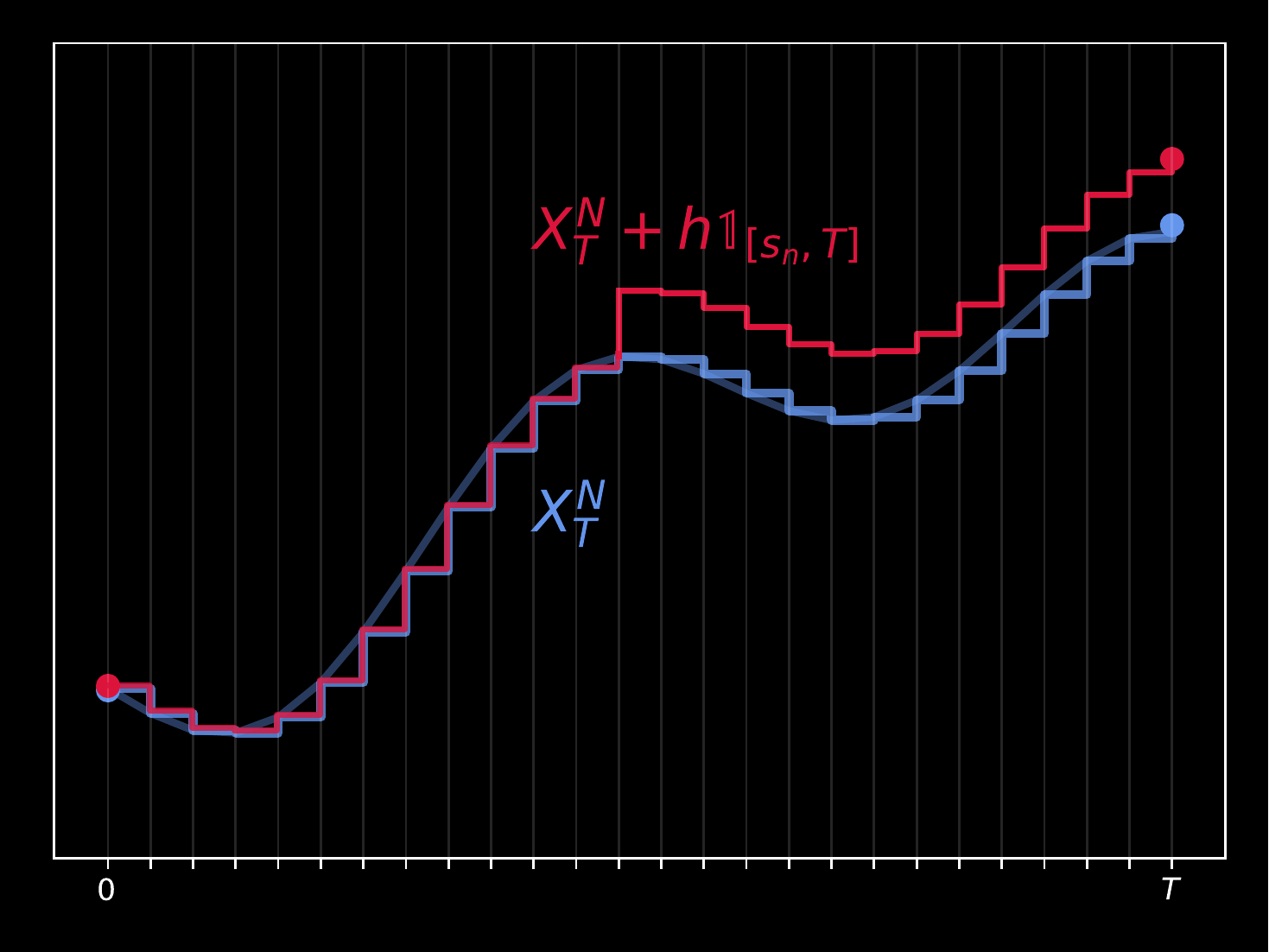}
    \label{fig:WienerChaos}
\end{subfigure}
\label{fig:Kernel}
\end{figure}
\subsection{Wiener Series}
\label{sec:WienerSeries}
As seen above, the Volterra series expresses $T-$functionals$-$seen as the output of a system$-$
in terms of pointwise values of the source path. 
In contrast, the Wiener series and homogeneous chaos expansion \cite{Wiener58,Wiener38}$-$corresponding to the deterministic and stochastic case,  respectively$-$encode the path with its infinitesimal increments.
We now present 
a seemingly unconventional  derivation of the Wiener series by  
 borrowing,  deliberately,  anachronistic elements from the Malliavin calculus.

Let $X\in \Lambda_T$ be a continuous path  of finite variation. 
Assume again that  $x_0 = 0$,   the knowledge of    $(x_{s_1},...,x_{s_N})$ is equivalent to the knowledge of the increments  $\delta x_{s_n} =x_{s_n} - x_{s_{n-1}}$, $n=1,\ldots,N$. We can therefore write $g^N(\delta x_{s_1},...,\delta x_{s_N}) := g(X_T^N)$ with $X_T^N$ as in \cref{sec:Volterra}. 
Similarly, write the Taylor expansion of $h \mapsto g^N(h \  \delta x_{s_1},...,h \  \delta  x_{s_N}) $  
at $h =0$. That is,
\begin{align*}
    g^N(x_{s_1},\ldots ,x_{s_N}) &= \phi^N_0 + \sum_{n=1}^N \phi^N_1(s_n) \delta  x_{s_n} + \sum_{1 \le n\le m \le N} \phi^N_2(s_n,s_m) \, \delta  x_{s_n}  \delta  x_{s_m}  + \ldots 
\end{align*}
with $\phi^N_0 = g^N(\boldsymbol{0})\ (= g(X_{0,T}) \ \forall N)$ and  $\phi^N_k(t_1,\ldots,t_k) = \partial_{n_1 ... n_k}g^N(\boldsymbol{0})$ whenever $t_l \in (s_{n_{l-1}},s_{n_{l}}]$, $l=1,\ldots,k$. 
Defining $\phi_0 = g(X_{0,T})$ and provided that $\phi_k := \lim_{N\to \infty} \phi_k^N $ exist for all $k\ge 1$, a passage to the limit yields, heuristically, the \textit{Wiener series}, 
\begin{align}\label{eq:chaosStieltjes}
 g(X_T) = \phi_0 + \sum_{k= 1}^{\infty}  \int_{\triangle_{k,T}}  \phi_k \, dx^{\otimes k} = \phi_0 + \sum_{k= 1}^{\infty} \int_{\triangle_{k,T}}  \phi_k(t_1,\ldots,t_k) \, dx_{t_1} \cdots dx_{t_k}.
 \end{align}
Let us express the kernels explicitly.  
First, note that $\partial_n g^N(\delta x_{s_1},...,\delta x_{s_N}) = D_{s_n}g(X_T^N) $, where $D_{\cdot}$ is the \textit{Malliavin derivative} 
 \begin{equation}\label{eq:MalliavinDer}
     D_tg(X_T) := D_{\mathds{1}_{[t,T]}}g(X_T) = \lim_{h \downarrow 0}\frac{g(X_T + h \mathds{1}_{[t,T]}) - g(X_T)}{h}. 
 \end{equation}
Indeed, bumping an increment of $X^N_T$
induces a parallel shift in the path 
following the shock; see \cref{fig:WienerChaos}. Letting $N\to \infty$ and evaluating the result at the flat path $X_{0,T}$ leads to $\phi_1(t_1) = D_{t_1}g(X_{0,T})$.  Similarly, writing $D_{t_1\cdots t_k} = D_{t_1} \cdots D_{t_k}$ gives $\partial_{n_1 ... n_k}g^N(\delta x_{s_1},...,\delta x_{s_N}) = D_{s_{n_1} \cdots s_{n_k}}g(X_T^N)$. In conclusion, the kernels in the Wiener series are  iterated Malliavin derivatives, namely
  $\phi_k(t_1,...,t_k) = D_{t_1\cdots t_k}g(X_{0,T})$.  
\begin{remark}
The Volterra and Wiener series  
 share many similarities. 
For a comprehensive comparison between the Volterra and Wiener chaos expansion, see  \citet{PalmPoggio}. 
In particular, the authors  show how  $\psi_k, \ \phi_k$ relate to one another; when $\phi_k$ is smooth, then 
$\psi_k = (-1)^k \partial_{t_1...t_k} \phi_k$ using integration by parts. 
This can be seen from \cref{fig:Kernel} in the case $k=1$.
Alternatively, observe that the Malliavin derivative is 
given by a Dirac impulse of the  "white noise" process $\dot{X} = \frac{d}{dt}X$.  
Consequently, the Volterra kernels of $\dot{X}$ are precisely the Wiener kernels of $X$. 
\end{remark}
  So far, the paths are of finite variation so that the iterated integrals in $\eqref{eq:chaosStieltjes}$ are in the Riemann-Stieltjes sense.  
  Before proceeding with the stochastic case (\cref{sec:chaos}) 
  where the integrals are in the It\^o sense, 
  we  propose 
  a generalization of the Wiener series to paths of finite quadratic variation along a given sequence of partitions.

\subsection{Intrinsic Value Expansion}\label{sec:IVE}

Consider a smooth $T-$functional  $g$ with  associated
intrinsic functional
 $\iota_{0}g(X_t) = g(X_{t,T-t})$  (see  \cref{sec:TFuncEmbedd}). Let us compute the functional derivatives of $\iota_{0}g$. First, the temporal derivative vanishes because the intrinsic value function is already defined through a flat extension. Second, it is easily checked that the spatial derivative  coincides with the Malliavin derivative of $g$ evaluated at the stopped path $X_{t,T-t}$. To sum up, 
 \begin{align}
 \Delta_t \iota_{0}g &\equiv 0.\\[0.5em]
     \Delta_x \iota_{0}g(X_t) &= D_t g(X_{t,T-t}).\label{eq:IVx}
 \end{align}
 
These observations are  illustrated in  \cref{fig:IVDerivatives} and incidentally  allow to represent $g$ in a compact way.  To keep the derivation short, we borrow results from  \cref{sec:FTE}. 

\begin{proposition}\label{prop:IVRepr}
Let $g$ be a twice continuously Malliavin differentiable $T-$functional. If $X_T$ is a continuous with finite quadratic variation along a given  sequence of partitions (see \cref{def:QV}), then 
\begin{equation}\label{eq:IVStrat}
    g(X_T) = g(X_{0,T}) + \int_0^T D_t g(X_{t,T-t}) \circ dx_t,
\end{equation}
where the integral is in the sense of Stratonovich. 
\end{proposition}
\begin{proof}
From the regularity of $g$ in the statement,  we can apply  the pathwise functional  Stratonovich formula (see \cref{thm:FSF}) to obtain 
\begin{align*}
    g(X_T)  &= \iota_{0}g(X_0) +\int_{0}^T \Delta_t \iota_{0}g(X_t)  dt +  \int_{0}^T \Delta_x \iota_{0}g(X_t) \circ dx_t.
\end{align*}
The result follows from \eqref{eq:IVx}, $ \iota_{0}g(X_0)=g(X_{0,T})$, and $\Delta_t \iota_{0}g \equiv 0$.   
\end{proof}
\begin{figure}[H]
\caption{Functional derivatives of the intrinsic value functional.}
\vspace{-2mm}
\begin{subfigure}[b]{0.49\textwidth}
    \centering
    \caption{$\Delta_t \iota_0 g(X_t) \equiv 0$}
    \includegraphics[height=1.6in,width=2.4in]{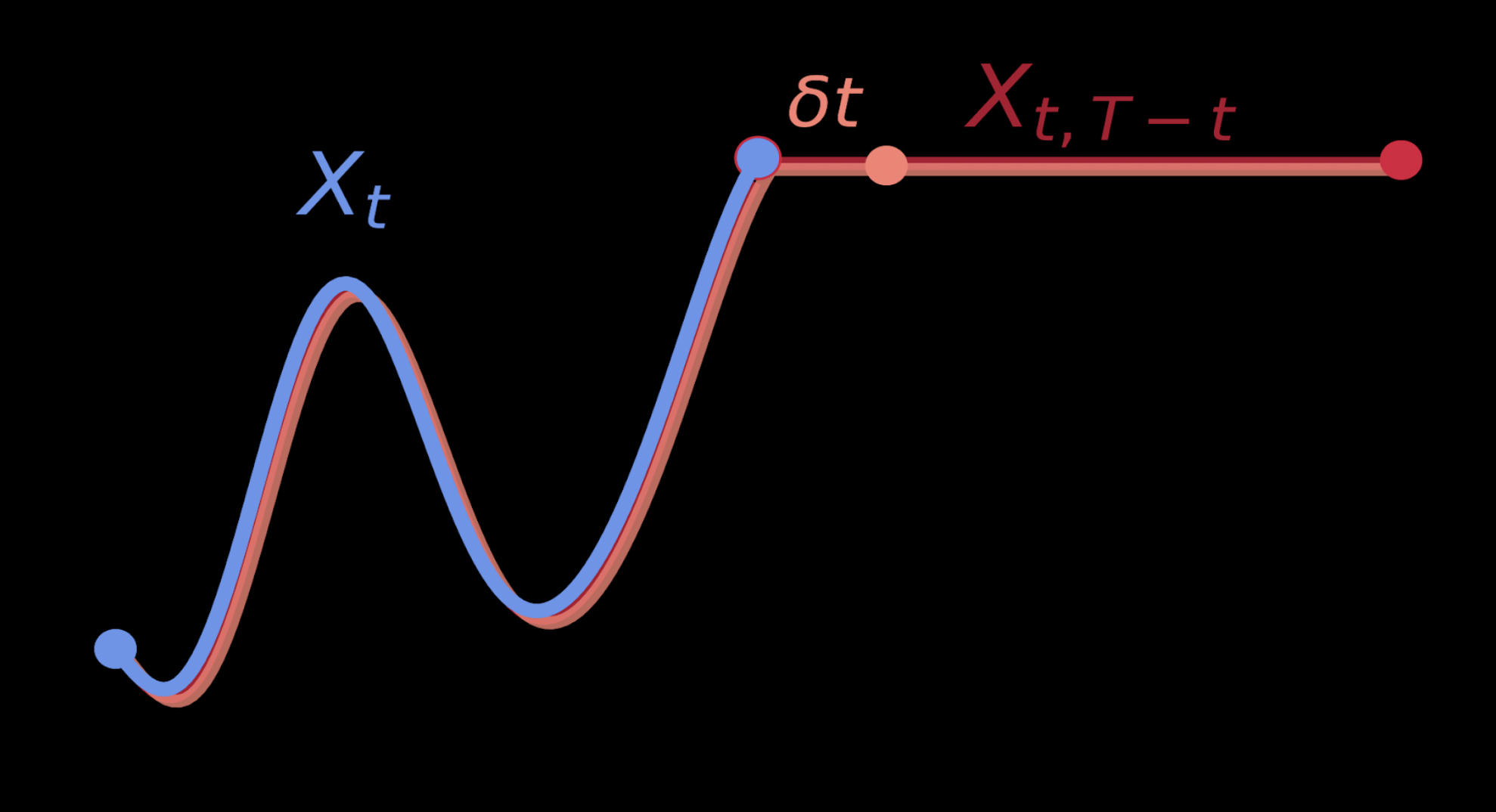}
    \label{fig:IVt}
\end{subfigure}
\begin{subfigure}[b]{0.49\textwidth}
    \centering
    \caption{$\Delta_x \iota_0 g(X_t) = D_t g(X_{t,T-t})$} \includegraphics[height=1.6in,width=2.4in]{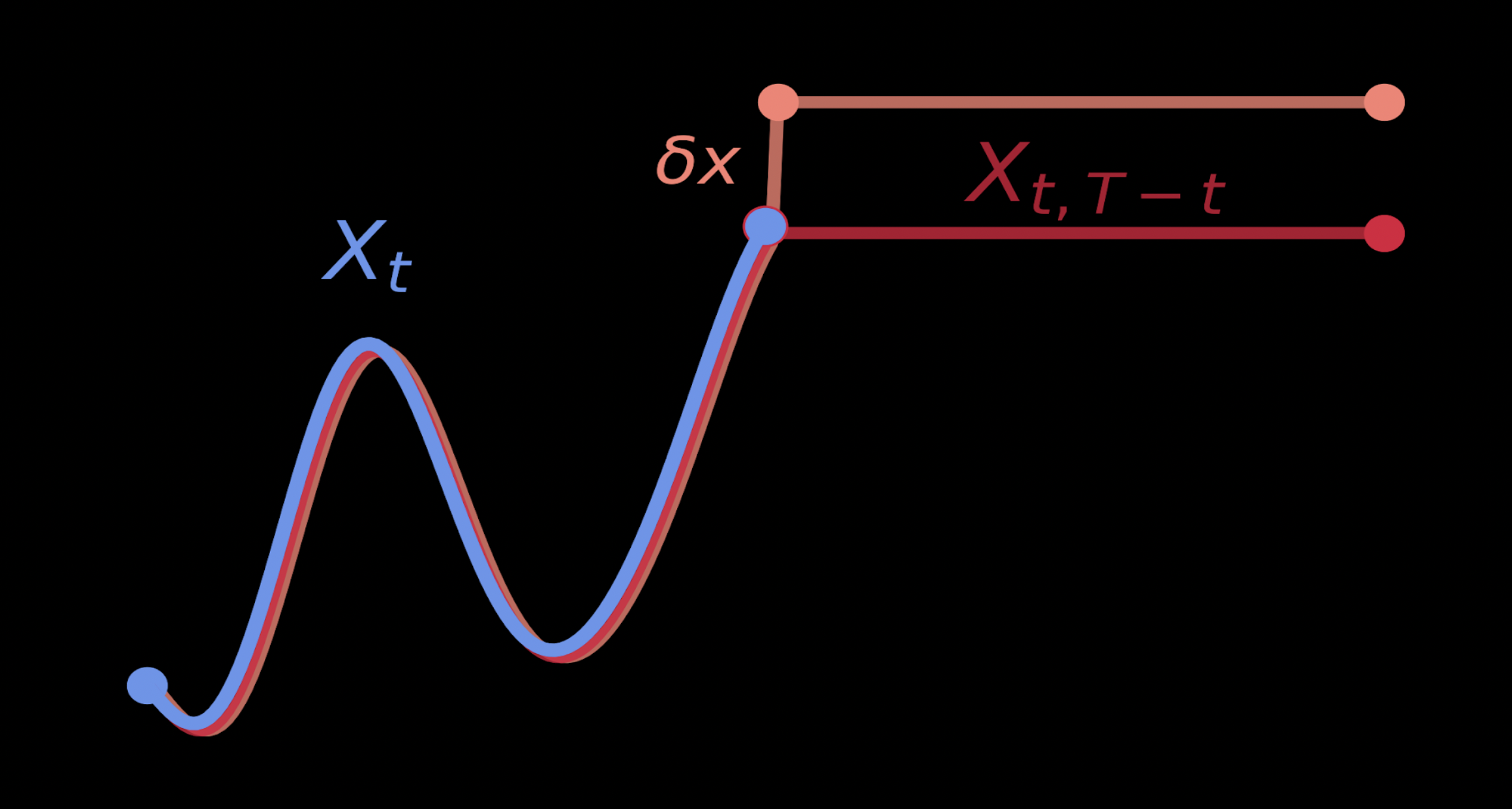}
    \label{fig:IVx}
\end{subfigure}
\label{fig:IVDerivatives}
\end{figure}
We can now iteratively apply \cref{prop:IVRepr} to derive  the \textit{intrinsic value expansion} (IVE) of $g$. 
\begin{theorem}\label{thm:IVE} 
\textnormal{(\textbf{Intrinsic Value Expansion (IVE)})} Let $K\ge 1$ and $g:\Lambda_T\to \R$ be $(K+1)-$times Malliavin differentiable and $X_T$ as in \cref{prop:IVRepr}.   Then 
  \begin{align} \label{eq:IVE}
 	g(X_T) &= g(X_{0,T}) + \sum_{k= 1}^{K-1} \int_{\triangle_{k,T}}  D_{t_1 \cdots t_k}g(X_{0,T})  \, \circ dx^{\otimes k} + R_K(X_T), \\[0.5em]
  R_K(X_T) &= \int_{\triangle_{K,T}}  D_{t_1 \cdots t_K}g(X_{t_1,T-t_1})  \, \circ dx^{\otimes K}.
 \end{align}
When expanding around a fixed auxiliary path $Y_T \in \Lambda_T$, the IVE becomes (assuming $x_0=0$)
  \begin{align} \label{eq:IVEXY}
 	g(X_T + Y_T) &= g(Y_T) + \sum_{k= 1}^{K-1} \int_{\triangle_{k,T}}  D_{t_1 \cdots t_k}g(Y_T)  \, \circ dx^{\otimes k} + R_K(X_T,Y_T).\\[0.5em] 
   R_K(X_T,Y_T) &= \int_{\triangle_{K,T}}  D_{t_1 \cdots t_K}g(Y_T + X_{t_1,T-t_1})  \, \circ dx^{\otimes K}.
 \end{align}
\end{theorem}
\begin{proof} The case $K=1$ corresponds to \cref{prop:IVRepr}. If $K\ge 2$ and  fixed $t\in [0,T]$, set  $g^1(X_t) = D_t g(X_{t,T-t})$ on the new horizon $[0,t]$ and define  $\iota_0 g^1(X_s) = g^1(X_{s,t-s})$. Again, the temporal  derivative is zero and $\Delta_x \iota_0 g^1(X_{s,t-s})=  D_{st} g(X_{s,T-s})$. \cref{prop:IVRepr} thus gives  
\begin{equation*}
	g^1(X_t) = D_t g(X_{0,T}) + \int_0^t  D_{st} g(X_{s,T-s}) \circ dx_t, \quad  D_{st} g(X_{s,T-s}) = \Delta_x g^1(X_{s,t-s}).
\end{equation*}
Iterating  and writing $\circ \ dx^{\otimes k} = \circ \ dx_{t_1} \cdots \circ dx_{t_k}$ yields the claim. To show \eqref{eq:IVEXY}, fix $Y_T\in\Lambda_T$ and apply \eqref{eq:IVE} to the smooth $T-$functional $X_T \mapsto g(X_T + Y_T)$. 
\end{proof}

Given an infinitely Malliavin differentiable $T-$functional $g$, then its \textit{intrinsic value series} is obtained by letting $K\to \infty$ in \cref{thm:IVE}, namely  
\begin{align} \label{eq:IVS}
 	g(X_T) = g(X_{0,T}) + \sum_{k= 1}^{\infty} \int_{\triangle_{k,T}}  D_{t_1 \cdots t_k}g(X_{0,T})  \, \circ dx^{\otimes k}, 
 \end{align}
 whenever the series converge, i.e. $R_{K}(X_T) \to 0$. Similarly if  $g(\cdot +Y_T)$ is infinitely Malliavin differentiable,  then the intrisic value series of $g$ around $Y_T$ writes 
  \begin{align} \label{eq:IVS2}
 g(X_T + Y_T) = g(Y_T) + \sum_{k= 1}^{\infty} \int_{\triangle_{k,T}}  D_{t_1 \cdots t_k}g(Y_T)  \, \circ dx^{\otimes k}. 
 \end{align}

 To the best of our knowledge, the IVE  and Intrinsic value series has not received attention in the current literature. 
 It  will be put into perspective with other expansions in \cref{sec:summary}  and \cref{sec:connectionIVE}. 
 
\begin{example}
 \label{ex:intrinsicfunctional} \textbf{(Intrinsic functional)} 
Let $g:\Lambda_T \to \R$ with intrinsic value series as in \eqref{eq:IVS}. As in \cref{ex:nullfunctional}, a natural embedding of $g$ into the space of functionals consists of restricting the iterated integrals appearing in the right side of \eqref{eq:IVS} to $[0,t]\subseteq  [0,T]$. In other words, define 
\begin{equation}\label{eq:IVEt}
    	f(X_t) = g(X_{0,T}) + \sum_{k= 1}^{\infty} \int_{\triangle_{k,t}}  D_{t_1 \cdots t_k}g(X_{0,T})  \, \circ dx^{\otimes k}, \quad t \in [0,T]. 
\end{equation}
Interestingly, $f$ is precisely the intrinsic functional of $g$. Indeed,   the IVE of $\iota_{0}g(X_t) = g(X_{t,T-t})$ on $[0,t]$ has constant term  $\iota_{0}g(X_{0,t})$ and kernels given by $D_{t_1 \cdots t_k}\iota_{0}g(X_{0,t})$,  $(t_1,\ldots,t_k) \in \triangle_{k,t}$.  Observing that $\iota_{0}g(X_{0,t}) = g((X_{0,t})_{t,T-t}) = g(X_{0,T})$, we conclude from $\eqref{eq:IVEt}$ that $f = \iota_{0}g$.
It comes as no surprise that  the intrinsic functional, illustrated in \cref{fig:WienerEmbedding}, is the natural embedding of $g$ induced by the intrinsic value expansion (and in turn, the Wiener series). 
\end{example}
\subsection{Wiener Chaos} \label{sec:chaos}
In this section, we restate the Wiener chaos expansion 
as well as useful tools from Malliavin calculus. We refer the interested reader to the excellent books \cite{DiNunno, Nualart,Malliavin} for further details. 
Let us consider a
stochastic basis  $(\Lambda,\calF,\F, \Q)$ with $\sigma-$algebra $\calF$, filtration $\F$ and probability measure $\Q$. 
With a slight abuse of notation, 
We  write $X$ for the coordinate process while we rather use the letter $Y$  when taking expectation under $\Q$.  
   For $p\in [1,\infty)$,  $t\in [0,T]$, 
let $L^p(\Lambda_t)$ be the  set of functionals on $\Lambda_t$   such that  
\begin{equation}\label{eq:L(Lambda_t)}
    \infty > \lVert f \rVert_{L^p(\Lambda_t)} := \left(\int_{\Lambda_t} |f(Y)|^p \Q(dY)\right)^{1/p}. 
\end{equation}
Notice that $L^p(\Lambda_t)$ implicitly depends on  $\Q$. 
 Furthermore,   we define 
\begin{equation}\label{eq:L(Lambda)}
L^p(\Lambda) = \left\{ f: \Lambda \to \R \ : \ 
\infty > \lVert f \rVert_{L^p(\Lambda)} :=   \left(\int_{\Lambda_T} \int_0^T |f(Y_t)|^p dt \, \Q(dY) \right)^{1/p}  \right \}. 
\end{equation}
Observe that $f \in L^p(\Lambda)$ if and only if the process $f(X) = (f(X_t))_{t\in [0,T]}$ belongs to $L^p(\Q \otimes dt)$. For the special case $p=2$, we see from $\eqref{eq:L(Lambda)}$ 
that $L^2(\Lambda) = \int_{[0,T]}^{\oplus} L^2(\Lambda_t) dt,$ i.e.  $L^2(\Lambda)$ is the direct integral of the Hilbert spaces $(L^2(\Lambda_t))_{t\in[0,T]}$ with inner product $(  f,f')_{L^{2}(\Lambda)} = \int_0^T (f,f')_{L^{2}(\Lambda_t)} dt $,  $\ (f,f')_{L^{2}(\Lambda_t)} = \int_{\Lambda_t} f(Y)f'(Y) \Q(dY)$.

\subsubsection{Wiener Chaos Expansion of   $T-$functionals}
Let $\Q$ be the Wiener measure so that the coordinate process $X$ is Brownian motion.\footnote{Although the trajectories of $X_T$ are $\Q-$a.s. continuous, 
$T-$functionals are still defined on $\Lambda_T$ to allow the presence of jumps when computing Malliavin or functional derivatives.} Moreover, let $\F$ be the filtration generated by $X$.      
One way to derive the chaos expansion is to iterate It\^o's representation theorem.    
For every $g\in L^2(\Lambda_T)$, we recall that the latter guarantees the existence of 
 a  
 functional $\varphi_1 \in L^2(\Lambda)$ such that
\begin{equation}\label{eq:MRT} 
    g(X_T) = \E^{\Q}[g(Y_T)] + \int_0^T \varphi_1(X_t) \ dx_t, \quad \Q-\text{a.s.} 
\end{equation}
Next, set $\phi_0 = \E^{\Q}[g(Y_T)]$ and  $\phi_1(t) = \E^{\Q}[\varphi_1(Y_t)]$. 
Since $\lVert \varphi_1 \rVert_{L^2(\Lambda_t)} < \infty$ for almost every  $t\in [0,T]$,  
we can apply \eqref{eq:MRT} 
to $\varphi_1(X_t)$ and obtain for some  functional $\varphi_2(t,\cdot)\in \int_{[0,t]}^{\oplus} L^2(\Lambda_{s}) ds$, 
$$   g(X_T) = \phi_0 + \int_0^T \phi_1(t) dx_t + \int_{\triangle_{2,T}}  \varphi_2(t_2,X_{t_1})  \ dx_{t_1} dx_{t_2}, \quad \Q-\text{a.s.}$$  
Pursuing the same logic further leads to the \textit{Wiener chaos expansion} of $g$, stated in \cref{thm:Wiener1}. 
Let us  introduce 
the integral  operators  
$J_k:L^2(\triangle_{k,T}) \to L^2(\Lambda) $, given by\footnote{From the definition of $\lVert \cdot \rVert_{L^2(\Lambda)}$ and  repeated use of It\^o isometry, note  that  
$\lVert J_{k}\phi_k \rVert_{L^2(\Lambda)}^2 \le T \ \lVert \phi_k  \rVert_{L^2(\triangle_{k,T})}^2. $  
Thus 
$J_{k}\phi_k \in L^2(\Lambda)$ as implicitly claimed.} 
 \begin{align} 
 J_{0}\phi_0(X_t) &\equiv  \phi_0 \in \R, \ (=: L^2(\triangle_{0,T})) \nonumber \\[0.5em]
  J_{k}\phi_k(X_t) &= \int_{\triangle_{k,t}} \phi_k\, dx^{\otimes k} = \int_{\triangle_{k,t}} \phi_k(t_1,...,t_k)\, dx_{t_1}\cdots dx_{t_k} , \quad  k\ge 1.\label{eq:chaosOperator}
\end{align}   
The image 
$\frakI_k: = \{J_k\phi_k \big |_{\Lambda_{T}} :  
\phi_{k} \in L^2(\triangle_{k,T}) \}$ 
is called the \textit{Wiener chaos of order $k$}. When $\phi_k \equiv 1$, $J_k$ generates the $k-$fold {It\^o iterated integral}, which we simply denote by $J_{k}(X_t)$.  We restate the Wiener chaos expansion  in terms of $T-$functionals; see, e.g. \cite[Theorem 1.10]{DiNunno} for a more classical formulation.

\begin{theorem} \textnormal{(\textbf{Wiener Chaos Expansion for $T-$functionals})}
\label{thm:Wiener1} The space of square integrable $T-$functionals can be expressed as a direct sum of chaos, namely
$L^2(\Lambda_T) = \bigoplus_{k=0}^\infty \frakI_k. $ 
That is, for all $g \in L^2(\Lambda_T)$, there exists  a unique sequence $\{\phi_k \in L^2(\triangle_{k,T})\}_{k \in \N}$  such that  
\begin{equation}
    \sum_{k= 0}^{K} J_k\phi_k(X_T)  \ 
    \underset{L^2(\Lambda_T)}{\overset{K\uparrow \infty}{\xrightarrow{\hspace{1.2cm}}}} \ 
    g(X_T).
\end{equation} 
Furthermore, 
$\lVert g \rVert^2_{L^2(\Lambda_T)} = \sum_{k=0}^{\infty} \lVert\phi_k \rVert^2_{L^2(\triangle_{k,T})}$. 
\end{theorem}

As shown by \citet{Stroock} (see also Theorem 1.32 in \cite{Malliavin}), the kernels can be made explicit when  $g$ is smooth in the sense of Mallavin, 
namely $g\in \D^{K,2}(\Lambda_T)$ for all $K\in \N$. We recall that $g\in \D^{K,2}(\Lambda_T)$ if for all $0\le k\le K$, the iterated Malliavin derivative $D_{t_1 \cdots t_k}  g = (D_{t_1} \cdots D_{t_k})g$ exists and satisfies 
\begin{equation}\label{eq:D_l_2}
    \E^{\Q}\left [\int_{\triangle_{k,T}} |D_{t_1 \cdots t_k}  g(Y_T)|^2 \ dt^{\otimes k }\right] < \infty. 
\end{equation}

If $k=0$, \eqref{eq:D_l_2} simply requires that $g\in L^2(\Lambda_T)$. 
The 
expression for $(\phi_k)_{k=1}^{\infty} $ can be derived recursively from the celebrated  
Clark-Ocone formula 
\cite{Clark,Ocone} refining It\^o's representation theorem when $g$ belongs to $\D^{1,2}(\Lambda_T)$.  
\begin{theorem}\label{thm:ClarkOcone} \textnormal{\bf (Clark-Ocone Formula)}
If $g\in \D^{1,2}(\Lambda_T)$, then  the integrand in \eqref{eq:MRT} is $\Q-$a.s. equal to $\varphi_1(X_t) = \E^{\Q}[D_t g(Y_T) | X_t]$, that is 
\begin{equation}\label{eq:ClarkOcone} 
    g(X_T) = \E^{\Q}[g(Y_T)] + \int_0^T \E^{\Q}[D_t g(Y_T) | X_t] \ dx_t, \quad \Q-\text{a.s.} 
\end{equation}
\end{theorem}
We therefore conclude from \cref{thm:ClarkOcone} that $\phi_1(t) = \E^{\Q}[D_t g(Y_T)]$.  When $g$ is smooth, iterating the argument together with the tower property leads to Stroock's representation of the Wiener kernels \cite{Stroock}:
 \begin{equation}\label{eq:stroock}
     \phi_k(t_1,...,t_k) = \E^{\Q}[D_{t_1 \cdots t_k}  g(Y_{T})],  \quad k\ge 0.
 \end{equation}

\begin{example}
 \label{ex:pricefunctional} \textbf{(Price functional)} Similarly to \cref{ex:nullfunctional,ex:intrinsicfunctional}, we truncate the simplexes in the chaos expansion to obtain a functional. 
Let $g\in L^2(\Lambda_T)$ with chaos expansion  $g(X_T) = \sum_{k= 0}^{\infty} J_k\phi_k(X_T)$ and define 
$f(X_t) = \sum_{k= 0}^{\infty} J_k\phi_k(X_t). $ Then $f\in L^2(\Lambda)$ and we claim that $f$ coincides with the \textit{price functional} $\iota_{\Q}g(X_t) =  \E^{\Q}[g(Y_T) \ | \ X_t]$. Indeed, the martingality of It\^o iterated  integrals  
gives
 $$\iota_{\Q}g(X_t)= \sum_{k =0}^{\infty} \iota_{\Q}J_k \phi_k( X_t) = \sum_{k =0}^{\infty} \E^{\Q}[J_k \phi_k  (Y_T)\ | \ X_t] = \sum_{k =0}^{\infty} J_k \phi_k  (X_t).$$
The price functional is therefore the natural embedding associated to the chaos expansion. 
\end{example}

\subsubsection{Expressing Wiener kernels with  functional Derivatives}\label{sec:MROperator}
The Clark-Ocone formula (\cref{thm:ClarkOcone}) allows to express a $T-$functional,  up to a constant,  as a stochastic integral.
The integrand is given by the  the optional projection 
of the (anticipative) Malliavin derivative of $g$. 
In contrast, the functional Itô calculus permits to rewrite this integrand 
directly in terms of a non-anticipative functional. This is  summarized in the next result, taken from \cite{Dupire}. 
\begin{theorem}\label{thm:FMRT}
 For every $g \in L^1(\Lambda_T)$
 such that  $f(X_{t})= \iota_{\Q}g(X_t) =  \E^{\Q}[g(Y_T) \,|\, X_{t}]$ is {spatially differentiable},
then 
\begin{equation}\label{eq:FMRT}
   g(X_T) = \E^{\Q}[g(Y_T)] + \int_0^T \Delta_x f(X_t) dx_t, \quad \Q-\text{a.s.}
\end{equation}
\end{theorem}
We shall see that the Wiener chaos expansion of $g$ can be recovered
by  successive applications of \cref{thm:FMRT}. 
To this end, consider the family of "martingale representation operators" 
\begin{equation} \label{eq:opA}
    \calM = (\calM_t)_{t\in [0,T]}, \quad   \calM_{t}\! f(X_s) = \Delta_x\E^{\Q}[f(Y_t) | X_s], \quad 0 \le s\le t, \quad f\in L^1(\Lambda_t),
\end{equation}
when the above expression is well-defined.  
It is worth noting that the domain of $\calM_T$ contains $\D^{1,2}(\Lambda_T)$, namely the set of $T-$functionals with Clark-Ocone representation \eqref{eq:ClarkOcone}. 
This follows from the general fact that integrating 
 before differentiating gives  more regularity than the other way around.  Besides, the uniqueness of the integrand in Itô's representation theorem implies that $\calM_Tg(X_t) = \E^{\Q}[D_t g(Y_T) \ | \ X_t ]$ $\Q-$a.s.  whenever $g \in \D^{1,2}(\Lambda_T)$. In terms of the price embedding $\iota_{\Q}$, this reads 
 $$ \Delta_x \iota_{\Q} = \iota_{\Q} D_t \;\; \text{ on } \D^{1,2}(\Lambda_T).  $$
 
Next, set $\calM_{t_1...t_k} = \calM_{t_1} \cdots \calM_{t_k}$, $(t_1,...,t_k) \in \triangle_{k,T}$, creating an entanglement of spatial derivatives and conditioned expectations. If $k=2$, this reads
$\calM_{t_2t_3}\!  f(X_{t_1}) = \Delta_x \E^{\Q}\left[ \Delta_x \E^{\Q}[ f(Z_{t_3}) \, | \, Y_{t_2}]  \, \big | \, X_{t_1} \right] \!.$  
To safely apply the operator over and over, we introduce the \text{domain of $\calM$}, 
\begin{equation}\label{eq:domain}
    \calD_{\! \calM} = \bigcap_{k=1}^\infty \ \{f:\Lambda \to \R \ : \ 
\, \calM_{t_1...t_k} \! f \text{ is well-defined } \forall \,  (t_1,...,t_k) \in \triangle_{k,T} \}.
\end{equation}
Hence $\calM_t(\calD_{\! \calM}) \subseteq \calD_{\calM}$.
For a $T-$functional $g$,  the  operator $\calM$ is only valid for $t=T$.  
The recursive use of $\calM$ is 
allowed as long as  $\calM_T \! g$, henceforth a running functional, belongs to  $\calD_{\! \calM}$.  
The next proposition highlights the role of the martingale representation operator to recover the  Wiener-Itô chaos expansion.

\begin{proposition} \label{prop:MRT_Chaos}
 Let $g$ be a $T-$functional such that $\calM_T g \in \calD_{\! \calM}$. Then the kernels in the Wiener chaos expansion of $g$ can be expressed as $\phi_0 = \calM_T \!  g(X_0)$ and 
\begin{equation}\label{eq:integrandFunc}
   \phi_k(t_1,...,t_k) = \E^{\Q}[\calM_{t_2...t_kT} \! g(Y_{t_1})]. 
\end{equation}
\end{proposition}

\begin{proof} 
The martingale representation of $g(X_T)$ can be rewritten as  
$ g(X_T) = \phi_0 +\int_0^T \calM_{T}\!  g(X_{t}) dx_{t}.$ 
Next, for fixed $t \le T$, the functional $\calM_{T}\!  g(X_{t})$ admits a martingale representation as well, namely
\begin{align*}
    \calM_{T}\!  g(X_{t}) &= \E^{\Q}[\calM_{T}\!  g(Y_{t})\, | \, X_0] + \int_0^{t} \Delta_x \E^{\Q}[\calM_{T}\!  g(Y_{t})\, | \, X_{s}]  dx_{s}
    = \phi_1(t) + \int_0^{t} \calM_{t T}\!  g(X_{s}) dx_{s}.
\end{align*}
 Hence, a simple integration gives
$  g(X_T) 
          = \phi_0 +  J_1\phi_1(X_T)  + \int_0^T \int_0^{t}  (\calM_{t T} g)(X_{s}) dx_{s}  dx_{t}.$
Iterating the argument 
with the adequate  
time variable renaming yields the claim.
\end{proof} 
The expression in  $\eqref{eq:integrandFunc}$ is in appearance similar to Stroock's representation $\eqref{eq:stroock}$ involving the Malliavin derivative $D$.  
 However, the operators $D$, $\calM$ are drastically different in nature:  $(D_t g(Y_{T}))_{t\in [0,T]}$ is a family of $\calF_T-$measurable random variables while $(\calM_T g(Y_t))_{t\in [0,T]}$ is an adapted process. 
\subsection{Summary}
\label{sec:summary}

We recapitulate 
the static expansions seen so far in 
\cref{tab:summary}. The second column recalls the natural embeddings discussed in \cref{ex:nullfunctional,ex:intrinsicfunctional,ex:pricefunctional} and  are also displayed in \cref{fig:embedding}.  
Note that the Wiener series and intrinsic value are very similar since the latter generalizes the former.  Also, there is an important difference between the constant terms: in the deterministic case, it is the image of  the flat path $X_{0,T}$ through $g$. On the other hand, in the chaos expansion, it is replaced by the expected value of $g$ over all Brownian paths. 
 
 \begin{table}[H]
    \centering
    \caption{Classification of static expansions, i.e. for $T-$functionals $g:\Lambda_T  \to \R$ } 
    \begin{tabular}{ccccc}
    \hline \hline 
    Expansion  & Constant term & Kernel   & Integral & Embedding \\ \hline \hline 
    & & & \\[-0.6em]
    Volterra series  & $g(X_{0,T})$ & $F_{t_1 \cdots t_k}g(X_{0,T})$  & Riemann &  $g(X_{t} \oplus \boldsymbol{0}_{T-t})$ \\[0.25em] \hline
  & & & \\[-0.8em]
     Wiener series  & $g(X_{0,T})$ & $D_{t_1 \cdots t_k}g(X_{0,T})$   & Stieltjes   & $g(X_{t,T-t})$\\[1em]
  IVE  &    $"$ & $"$     & Stratonovich & $"$ \\[0.25em] \hline
  & & & \\[-0.8em]
         Wiener chaos  & $\E^{\Q}[g(Y_{T})]$ & $\E^{\Q}[D_{t_1 \cdots t_k}  g(Y_{T})]$   & It\^o & $\E^{\Q}[g(Y_{T}) \ | \ X_t]$\\[0.2em] \hline  
    \end{tabular}
    \label{tab:summary}
\end{table}

Finally, let us compare the intrinsic value expansion and Wiener chaos expansion in further depth. For convenience,  the expansions are recalled: 
\begin{align}
 	g(X_T)\hspace{0.1cm}  &= \hspace{0.1cm} g(X_{0,T})  &&+ \hspace{0.6cm} \sum_{k= 1}^{\infty} \int_{\triangle_{k,T}}  D_{t_1 \cdots t_k}g(X_{0,T})  \, \circ\, dx^{\otimes k} , \hspace{1cm} \quad &&\text{(IVE)} \nonumber \\[0.5em] 
 	g(X_T) \hspace{0.1cm} &= \hspace{0.1cm} \E^{\Q}[g(Y_{T})]  &&+ \hspace{0.6cm} \sum_{k= 1}^{\infty} \int_{\triangle_{k,T}} \E^{\Q}[D_{t_1 \cdots t_k}  g(Y_{T})]  \; dx^{\otimes k}, \quad \Q-a.s. \quad &&\text{(Wiener chaos)} \quad \label{eq:WienerChaos2}
 \end{align}
  First, the nature of iterated integrals differ (Stratonovich versus It\^o sense). Second, although both expansions involve iterated Malliavin derivatives of $g$, the kernels are computed according to the corresponding natural embedding 
 (intrinsic  and price functional, respectively). 
But the main difference comes from the paths themselves: In the intrinsic value expansion, the paths can be quite arbitrary, as long as their roughness is controlled along some given sequence of partitions. In the Wiener chaos expansion, they must be  "typical" Brownian paths. 

To bring the expansions closer,  we compute the IVE of $X_T\to g(X_T + Y_T)$ 
as in \cref{thm:IVE}  and thereafter take expectation with respect to $Y_T$ under $\Q$. This yields 
\begin{equation}\label{eq:intrinsicVsChaos}
\E^{\Q}[g(X_T + Y_T)]  = \E^{\Q}[g(Y_T)] + \sum_{k= 1}^{\infty} \int_{\triangle_{k,T}}  \E^{\Q}[D_{t_1 \cdots t_k}g(Y_T)] \, \circ dx^{\otimes k}.
\end{equation}
The right side of \eqref{eq:intrinsicVsChaos} closely resembles the Wiener chaos expansion of $g(X_T)$ \eqref{eq:WienerChaos2}, except that the integrals are in the sense of Stratonovich instead of Itô. If $X_T$ is a typical Brownian  path and admissible for the IVE, subtracting \eqref{eq:WienerChaos2} to \eqref{eq:intrinsicVsChaos} yields the   bias representation 
\begin{align}\label{eq:Bias}
    \E^{\Q}[g(X_T + Y_T)] -g(X_T) = \sum_{k= 1}^{\infty} \int_{\triangle_{k,T}}  \E^{\Q}[D_{t_1 \cdots t_k}g(Y_T)] \, [\circ \ dx^{\otimes k} -  dx^{\otimes k}], \quad \Q-a.s.
\end{align}
We remark that the first  term on the right side of \eqref{eq:Bias} vanishes 
when $t \mapsto \E^{\Q}[D_{t}g(Y_T)]$ (equal to  $\phi_1$) is of finite variation.  Also, the second term  simply reads  $\frac{1}{2}\int_0^T \E^{\Q}[D_{t T}g(Y_T)] dt$ using Itô's isometry. 
 
\begin{figure}[H]
\caption{Natural embeddings of static expansions. We recall that (b) is the intrinsic functional, while  (c) is  the  price functional under the Wiener measure.}
\begin{subfigure}[b]{0.33\textwidth}
    \centering
    \caption{Volterra}
    \includegraphics[height=1.65in,width=1.95in]{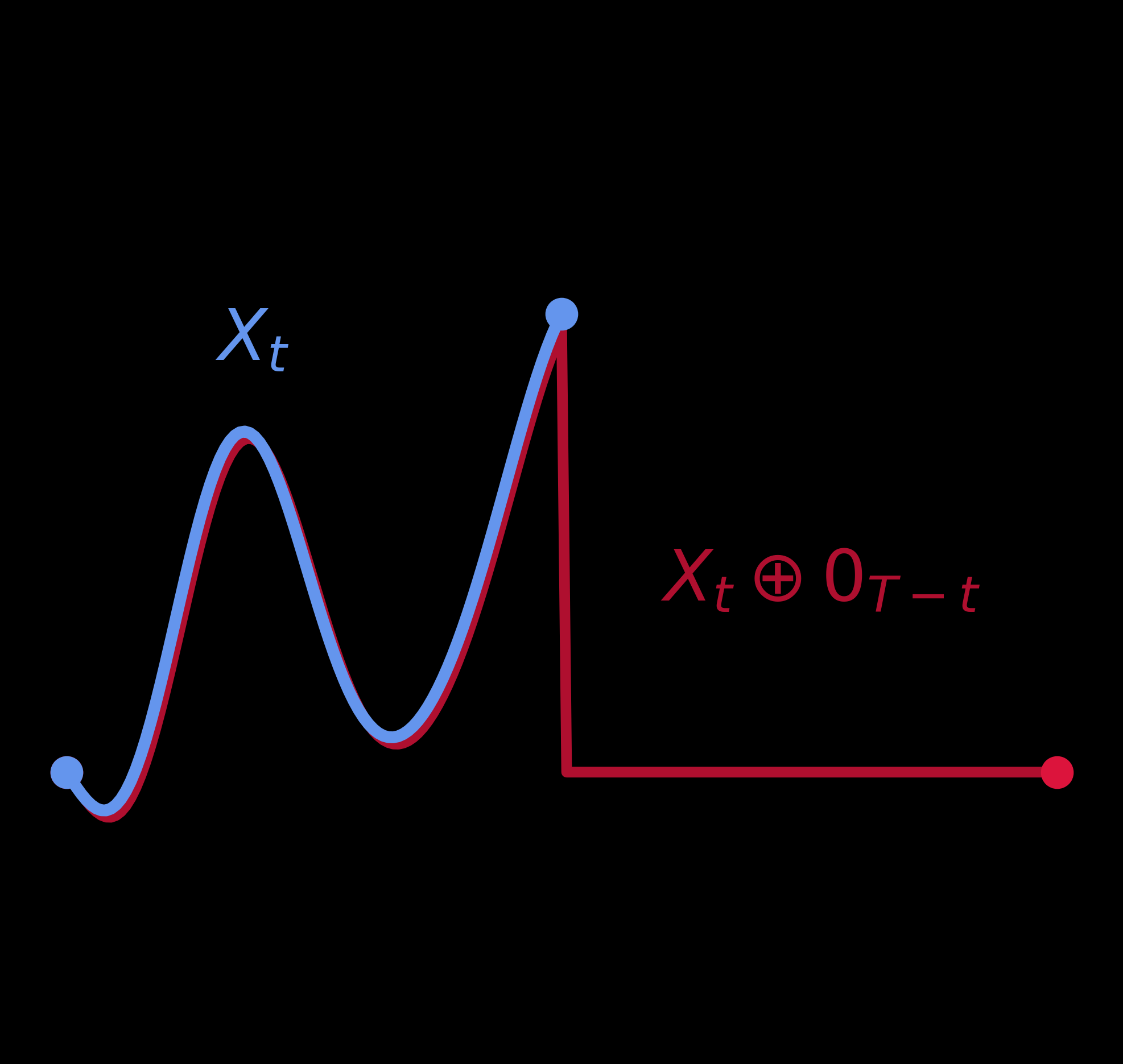}
    \label{fig:VolterraEmbedding}
\end{subfigure}
\begin{subfigure}[b]{0.33\textwidth}
    \centering
    \caption{Wiener series/IVE}
    \includegraphics[height=1.65in,width=1.95in]{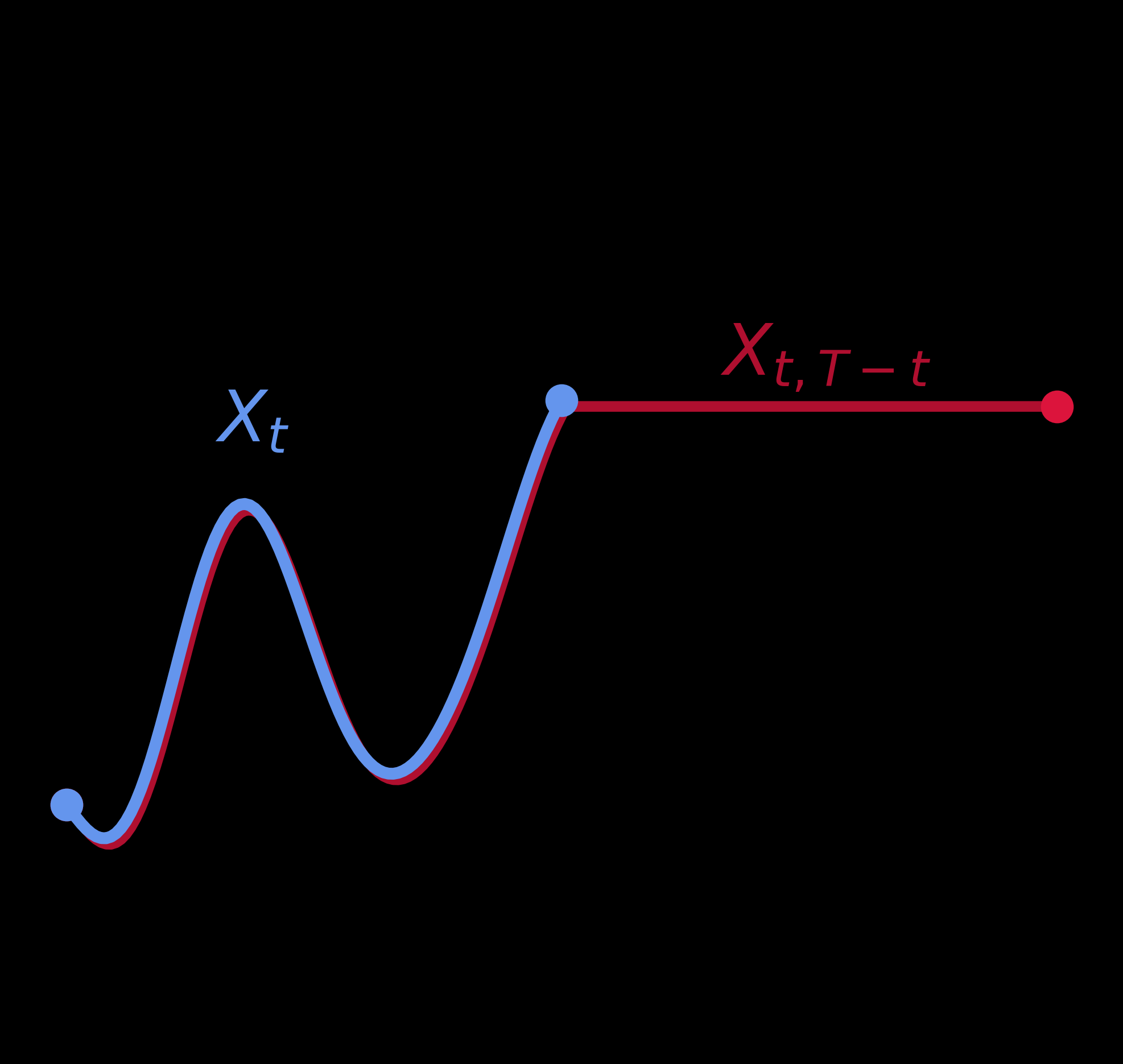}
    \label{fig:WienerEmbedding}
\end{subfigure}
\begin{subfigure}[b]{0.33\textwidth}
    \centering
    \caption{Wiener chaos}
\includegraphics[height=1.65in,width=1.95in]{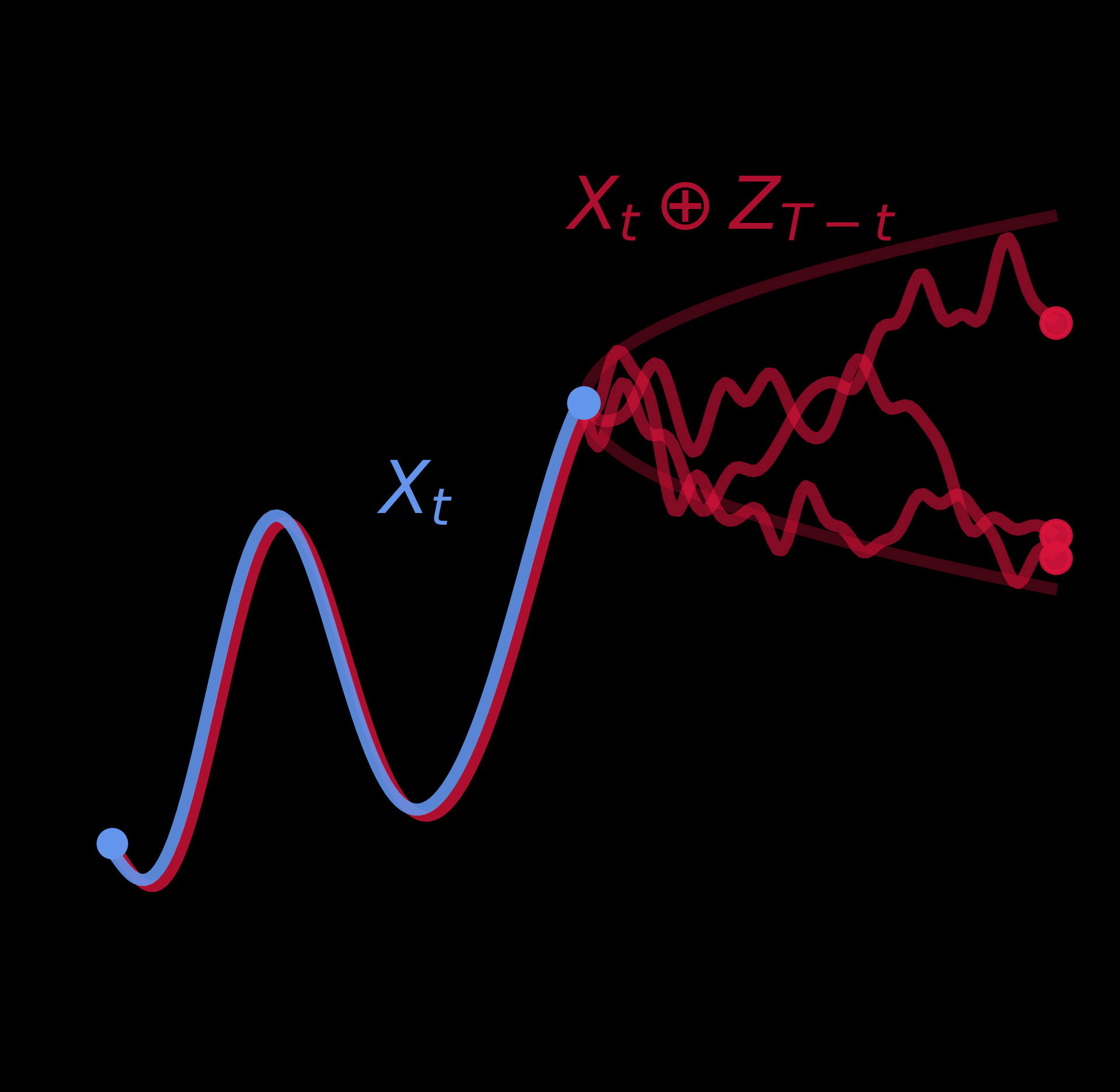}
    \label{fig:chaosEmbedding}
\end{subfigure}
\label{fig:embedding}
\end{figure}

\label{sec:dynamic}
\section{Dynamic Expansions}\label{sec:DynamicExpansions} 

\subsection{Functional Taylor Expansion}\label{sec:FTE}
The  Taylor expansion  
 lies at the core of calculus and finds applications  in countless fields. It provides an explicit 
 approximation of smooth functions by polynomials. A natural question is whether the Taylor expansion 
 can be extended to more general spaces  such as $\Lambda$. 
 In particular, one may wonder whether there are functionals echoing 
 the  monomials. 
A natural candidate is the family of \textit{signature functionals}. The path signature, first studied by \citet{Chen}, has gained much attention in recent years \cite{Arribas,Szpruch,CuchieroSig,LyonsNum}, due notably   
to its universal approximation property. 
Indeed, density results à la Stone-Weierstrass exist  \cite{LittererOberhauser,CuchieroSig} 
allowing, financially speaking, to both price and replicate exotic payoffs with signature elements; see \cite{LyonsNum,Szpruch}. However, the Stone-Weierstrass theorem is existential instead of constructive.  As such, a procedure that  explicitly generates the replication of a claim would be preferred. 

In response to this appeal,  
we propose a pathwise 
Taylor expansion for  path-dependent payoffs: 
the \textit{functional Taylor expansion} (FTE). 
 To the best of our knowledge, this generalization has been first proposed by \citet{Fliess81,Fliess83,Fliess86},  termed the \textit{causal  Taylor expansion},  for $T-$functionals and continuously differentiable paths. Interestingly, the \textit{causal derivative} in \cite{Fliess83} 
shares similarities with the spatial derivative from the functional Itô calculus \cite{Dupire}. The main difference is that the discontinuous "bump" in the latter is replaced by the vertical segment joining the terminal value of the original path and the bumped one. Moreover, the causal derivative were defined on the whole horizon, i.e. for $T-$functionals only. 
 More recently, using the functional derivatives in \cite{Dupire}, 
 \citet{LittererOberhauser}  establish a functional Taylor expansion$-$called \textit{Chen-Fliess approximation}$-$where the underyling path 
comes from the strong solution of a stochastic differential equation with smooth vector fields. 
 
Another 
approach is given in  
 \citet{BMZ} where the authors build a Taylor expansion for 
 stochastic processes 
suitable to characterize viscosity solutions of path-dependent partial differential equations (PPDEs). In particular, a pathwise estimate of the remainder is given in \cite{BMZ}, although still requiring  
a stochastic framework. 
It is worth noting that the spatial functional derivative in  \citet{BMZ}  
is slightly different  from the ones in \cite{Dupire} and the present article (in fact, weaker) and arises from the   martingale part in the Doob-Meyer decomposition. 
In a subsequent work, 
\citet{BKMZ} introduce the rough functional Taylor expansion  where the path   derivative is in the sense of  \citet{Gubinelli}. The signature is thus given by iterated rough integrals either defined pathwise in the one-dimensional case or using higher order "lifts" of the  path otherwise; see \citet{Lyons,FrizHairer}. 

In this work, we adopt the pathwise framework pioneered by \citet{Follmer}   
combined with  tools from the functional Itô calculus \cite{Dupire}.  
 Consequently, our Taylor expansion is made "$\omega$ by $\omega$" and  thus model-free.  
Choosing a similar avenue as \citet{ContBally} where pathwise functional It\^o formulae are derived, 
we establish a pathwise functional Stratonovich formula (\cref{thm:FSF}) for one-dimensional continuous paths of finite quadratic variation along a given sequence of partition.  This provides, in passing, a way to construct Stratonovich integrals as well as the signature in a pathwise manner. We stress that our construction of the signature relies on limits of simple Riemann sums and thus differs from the more 
abstract definition  in rough path theory \cite{Lyons,FrizHairer} given by imposed algebraic relations. 
In \cref{sec:results}, we eventually state  the functional Taylor expansion 
in various forms.  
In particular, the flexibility of the proposed framework allows, in \cref{thm:FTE}, to expand a functionals 
after a path which may not be regular at all.

The FTE presents several applications in control theory, engineering and finance. 
In \cref{sec:FTEApp}, we focus on the latter and outline the main use of the FTE when pricing or hedging derivatives. 
As also pointed out in \cite{LittererOberhauser}, the FTE is of great relevance 
in \textit{cubature methods} (see \citet{LyonsVictoir,Crisan} and the references therein). 
In short, cubature  is a  generalization of  Gaussian quadrature  to the path space; 
the exact integration of  polynomials 
turns into a perfect fit of expectation of signature elements.  
This goes, however, beyond the scope of this work.

\subsubsection{Functional Stratonovich Formula }\label{sec:FSf}

One of the main ingredients in the FTE is a variation of the Functional It\^o formula \cite{Dupire}.  
In \cref{thm:FSF}, we simultaneously  employ a pathwise setting similar to \citet{Follmer}   and Stratonovich integrals as in  \citet{LittererOberhauser}. 
See also   \citet[Chapter 5]{ContBally}  and the references therein for   pathwise functional It\^o  formulae, also along the lines of F\"ollmer's approach. %
First, we define a suitable class of integrators. 
\begin{definition} \label{def:QV}
 Let $\Pi = (\Pi^N)$ be  a sequence  of  partitions of  $[0,T]$ with vanishing mesh size, that is  $|\Pi^N| := \max_{t_n \in \Pi^N} |t_{n}-t_{n-1}| \overset{N \uparrow \infty}{\longrightarrow} 0 $. For fixed  $t\in [0,T]$, 
we say that $X \in \Lambda_t$ is a \textit{$\Pi$-integrator}  if $X$ is continuous and such that the quadratic variation function along $\Pi$, i.e.,  
\begin{equation}\label{eq:QV}
   [0,t] \ni s \mapsto  \langle X \rangle^{\Pi}_s := \lim_{| \Pi_s^N | \downarrow 0}\sum_{t_n\in \Pi^N_s} (x_{t_n}-x_{t_{n-1}})^2, \qquad \Pi^N_s := \{t_n \wedge s \ | \ t_n \in \Pi^N\},
\end{equation} 
exists and is continuous.
We write $\Omega^{\Pi}_t$ for the set of $\Pi$-integrators in $\Lambda_t$  and   $\Omega^{\Pi} := \bigcup_{t\in [0,T]} \Omega^{\Pi}_t$. 

\end{definition}

Notice that if $X\in \Omega^{\Pi}_t$, then $X_s \in \Omega^{\Pi}_s$ for all $s\le t$ using  the sequence of "stopped"  
partitions $\Pi_s := (\Pi_s^N)$ defined in $\eqref{eq:QV}$. The choice of $\Pi$ may considerably impact the richness of $\Omega^{\Pi}$. Good candidates are refining partitions, that is $\Pi^{1} \subseteq \Pi^{2} \subseteq \ldots $ such as the dyadic partition $\Pi^{N} = \{T n 2^{-N} \ : \ n=0,\ldots , 2^N\}$. In effect, having a nested structure in the sequence of partitions  guarantees the almost sure convergence of the quadratic variation for some 
stochastic processes, e.g. Brownian motion \cite[Proposition 2.12.]{RevuzYor}.  

 To compute the spatial derivatives of $f$ at a continuous path $X \in \Omega^{\Pi}_t$, we need to make sure that $f$ is well-defined when adding vertical bumps to $X$. To this end, we consider the enlargement   
 $\bar{\Omega}^{\Pi}_t = \Omega^{\Pi}_t + \calJ_t$,  where $\calJ_t$ contains all piecewise constant  paths in $\Lambda_t$ of finite variation. 
 Defining as usual  $\bar{\Omega}^{\Pi} = \bigcup_{t\in [0,T]} \bar{\Omega}^{\Pi}_t$ gives the sequence of strict inclusions $\Omega^{\Pi} \subset \bar{\Omega}^{\Pi} \subset \Lambda$.

As is customary,  regularity conditions are required on the function to expand. More specifically, consider the classes of  functionals defined on  some subset $\Lambda' \subseteq \Lambda$,
$$\C^{K,L}(\Lambda') = \{f:\Lambda' \to \R \, : \, \Delta_{\alpha}f \text{ exists and is $\Lambda-$continuous } \, \forall \, 
\alpha \text{ s.t. }\, |\alpha|_0 \le K, \,  |\alpha|_1 \le L\}, \quad K,L \ge 0.$$ 
Taking $K=L=\infty$ gives the class $\C^{\infty,\infty}(\Lambda')$ of \textit{smooth functionals} in $\Lambda'$. 

\begin{theorem}\textnormal{(\textbf{Functional Stratonovich formula})}\label{thm:FSF}  
Let  $X\in \Omega^{\Pi}_t$ for some $t \in [0,T]$  and $f\in \C^{1,2}(\bar{\Omega}^{\Pi})$. 
Then, 
\begin{align}\label{eq:FSF}
    f(X_t) &= f(X_0) + \int_0^t \Delta_t f(X_s)  ds + \int_0^t\Delta_x f(X_s) \circ dx_s,
\end{align}
where the last integral implicitly depends on the sequence of partitions $\Pi$ associated to $X$. 

\end{theorem}

\begin{proof}
	See \cref{app:FSF}.  
\end{proof}

\begin{remark}
Although it is necessary to work with c\`adl\`ag trajectories to define the spatial derivative, \cref{thm:FSF} and in turn, the functional Taylor expansion, is established for continuous paths. 
Indeed, doing so prevents the presence of jump terms which would make the FTE more general but   much more intricate. 
\end{remark}

\begin{example}
Let $X\in \Omega^{\Pi}_t$ and $f(X_s) = h(x_s)$ for some real function $h\in \calC^2(\R)$. Then clearly $f\in \C^{1,2}(\bar{\Omega}^{\Pi}_t)$ since $\Delta_t f = 0 $, $\Delta_{x}f = h'$, and  $\Delta_{xx}f = h''$. We can thus apply the pathwise Stratonovich formula to obtain
\begin{equation*}
    h(x_t) = h(x_0) + \int_0^t h'(x_s) \circ dx_s. 
\end{equation*}
We here recover the well-known fact that Stratonovich integrals yield a first order calculus. 
\end{example}

\begin{example}\label{ex:IntDev}
Let us verify the pathwise Stratonovich formula the time integral $f(X_t) = \int_0^t \varphi(X_s) ds$ where the pair $(\varphi,X)$ satisfies $\int_0^t |\varphi(X_s)| ds < \infty$. 
Since $\{ s\in [0,t] \, : \, X_t(s) \ne X^{h}_t(s)\} = \{ t \}$ has Lebesgue measure zero for all $h \neq 0$, we conclude that $\Delta_xf \equiv 0$. For the time derivative, we  obtain  
    $  \Delta_t f(X_t) = \lim_{\delta t\downarrow 0} \,  \frac{1}{\delta t}\int_0^{\delta t} \varphi(X_{t,u})  du =  \varphi(X_{t}). 
    $
As expected, this gives
$$f(X_t) = \underbrace{f(X_0)}_{= \ 0} + \int_0^t \underbrace{\Delta_tf(X_s)}_{\varphi(X_s)} ds  + \int_0^t \underbrace{\Delta_xf(X_s)}_{= \ 0} \circ \, dx_s = \int_0^t \varphi(X_s) ds. $$
We also remark that  $\Delta_t \Delta_x f = 0$ and   $\Delta_x \Delta_t f = \Delta_x \varphi$ provided that  $\Delta_x \varphi$ exists. Hence $\Delta_t, \Delta_x$ do not commute when $\Delta_x \varphi \ne 0$. 
\end{example}

\subsubsection{Spatial Anti-derivative and Pathwise Integration} 
\label{sec:pathStrat}
In the same spirit as Föllmer's pathwise It\^o calculus 
\cite{Follmer}, \cref{thm:FSF} allows to incidentally define pathwise Stratonovich integrals $\int y \circ dx$ where the integrator may be of infinite variation.  However, there are constraints on the integrand: $Y$ must be the spatial derivative of some functional $f\in \C^{1,2}$ evaluated at $X$. In other words, the admissible integrands correspond to all delta hedging strategies where $f$  is interpreted as the price functional of an option. 

In the path-independent case 
as in \citet{Follmer}, the integrands 
are simply the continuously differentiable 
 functions (both in time and space) of the path, i.e. $y_t = h(t,x_t), \ h \in \calC^{1,1}(\R)$. Indeed, it suffices to apply the Stratonovich formula to $H(t,x_t) := \int^{x_t}_{x_0} h(t,y) dy$. 
 Back to our framework, we construct a spatial anti-derivative  as follows. 
 Given $\varphi\in \C^{1,1}(\bar{\Omega}^{\Pi})$, write $X_t^{(\varepsilon)} := X_t^{\varepsilon-x_t}$  
 ($\Longrightarrow  X_t^{(\varepsilon)}(t) = \varepsilon$)  and define 
 \begin{equation}\label{eq:anti-der}
     \Phi(X_t) = \int_{x_0}^{x_t} \varphi(X_t^{(\varepsilon)}) \ d \varepsilon. 
 \end{equation} 
 Then,  we easily see that 
 $\frac{\Phi(X^{\delta x}_t) - \Phi(X_t)}{\delta x} = \frac{1}{\delta x}\int_{0}^{\delta x}\varphi(X_t^{(\varepsilon)}) \ d \varepsilon \,  \overset{\delta x \to 0}{\xrightarrow{\hspace{0.6cm}}} \,  \varphi(X_t)  $ 
and $\Phi$ is a spatial anti-derivative  as desired.  
Moreover,  $\Delta_t \Phi(X_t)  = \int_{x_0}^{x_t} \Delta_t \varphi(X_t^{(\varepsilon)}) d \varepsilon$. 
Since $\Phi(X_0) = 0$, we can invoke and rearrange the functional Stratonovich formula to obtain the pathwise  integral 
 \begin{equation}\label{eq:pathStrat}
     \int_0^t \varphi(X_s) \circ dx_s = \int_{x_0}^{x_t} \varphi(X_t^{(\varepsilon)}) \ d \varepsilon - \int_0^t \int_{x_0}^{x_s} \Delta_t \varphi(X_s^{(\varepsilon)}) d \varepsilon \ ds. 
 \end{equation}
 In particular, when $\Delta_t \varphi = 0$, we simply have $   \int_0^t \varphi(X_s) \circ dx_s = \int_{x_0}^{x_t} \varphi(X_t^{(\varepsilon)}) \ d \varepsilon. $
 Thus, 
 the Stratonovich integral 
 can be made pathwise for \textit{all} functionals in $\C^{1,1}(\bar{\Omega}^{\Pi})$.  Incidentally, this allows us to properly define signature functionals in $\bar{\Omega}^{\Pi}$; 
 see \cref{prop:Sig}. It is important to point out that this simple construction  works only for one-dimensional paths. 
 \begin{remark}
 From the pathwise functional It\^o formula \cite{ContAnanova,ContBally}, note that 
 \begin{align}\label{eq:pathwiseIto}
      \int_0^t \varphi(X_{s-}) dx_s &= \int_{x_0}^{x_t} \varphi(X_t^{(\varepsilon)}) \ d \varepsilon - \int_0^t \left(\Delta_t + \frac{1}{2}\Delta_{xx}\right)\Phi(X_s) \ ds. 
  \end{align}
  The right side of \eqref{eq:pathwiseIto} thus provides a  pathwise definition of It\^o integrals.  
 \end{remark}

 \begin{example}\label{ex:StratDev}Similar to \cref{ex:IntDev}, we  compute the functional derivatives of  the  Stratonovich integral $f(X_t) = \int_0^t \varphi(X_s) \circ dx_s$ defined on $\bar{\Omega}^{\Pi}$. From the above discussion, we obtain as expected that 
 $$\Delta_t f(X_t)  = 0, \quad \Delta_x f(X_t)  = \Delta_x \int_{x_0}^{x_t} \varphi(X_t^{(\varepsilon)}) \ d \varepsilon = \varphi(X_t). $$ 
\end{example}

\subsubsection{Signature Functionals and Properties} \label{sec:sig}

The FTE entails a collection of objects that characterizes the underlying path in its entirety. This is given by the \textit{signature} \cite{Chen,Lyons} which can be seen as the infinite skeleton of a path. 

\textbf{Notations.} A \textit{word} is a sequence  $\alpha = \alpha_1 \ldots \alpha_k$ of letters 
from the alphabet $\{0,1\}$.  The number of $0$'s and $1$'s in $\alpha$ is denoted by $|\alpha|_0$, $|\alpha|_1$, respectively. The \textit{length} of $\alpha$ is therefore $|\alpha| := |\alpha|_0 + |\alpha|_1$.  
We will often use the special words $\boldsymbol{0}_k := 0\ldots 0$ such that  $|\boldsymbol{0}_k| = |\boldsymbol{0}_k|_0=k$ and $\mathds{1}_k := 1\ldots 1$ such that  $|\mathds{1}_k| = |\mathds{1}_k|_1=k$.
Also, for convenience, we write  
$\alpha - j := \alpha_1 \ldots \alpha_{k-j}$ where $|\alpha|=k$ and $0\le j \le k-1$. Finally,   $\A$ denotes the set containing all the words. 

Next, we enlarge a path $X \in \Lambda$ with the time itself
and henceforth set 
$x^0_{t} = t$, $x^1_{t} = x_t$. This is a common procedure used to enrich the signature, see  \cite{Szpruch,CuchieroSig,LyonsNum} and \cref{rem:timeResconstruction} below. 
The letters $0,1$ are therefore identified with the time $t$ and path $x$, respectively. The definition of the signature is a byproduct of the following  proposition. 

 \begin{proposition}\label{prop:Sig}
 Let $\Pi$ be a sequence of  partitions  as in \cref{def:QV}. Recall also the simplexes $(\triangle_{k,t})$ defined in $\eqref{eq:simplex}$.  
 Then for all  $X_t \in \bar{\Omega}^{\Pi}$ and $\alpha \in \A$,  
 the iterated  Stratonovich  integral
     \begin{align}\label{eq:sigDefStrat}
       \calS_{\alpha}(X_t) := \int_{\triangle_{k,t}} \circ \, dx^{\alpha} =\int_{0}^{t} \int_{0}^{t_k} \cdots \int_{0}^{t_2} \circ \, dx^{\alpha_1}_{t_1} \cdots \circ dx^{\alpha_k}_{t_k}, \qquad (k=|\alpha|)   
     \end{align}
     is well-defined pathwise.   
     Moreover,   $\calS_{\alpha} \in \C^{\infty,\infty}(\bar{\Omega}^{\Pi})$ for all $\alpha \in \A$, 
     and satisfies the recursion, 
     \begin{equation}\label{eq:derivSig}
              \Delta_{t} \calS_{\alpha} = \begin{cases}\calS_{\alpha-1}, & \alpha_k=0\\
              0, & \alpha_k=1\\\end{cases}, \qquad  \Delta_{x} \calS_{\alpha} =\begin{cases}\calS_{\alpha-1}, & \alpha_k=1\\
              0, & \alpha_k=0\\\end{cases}. \qquad  (k=|\alpha|)  
     \end{equation}
 \end{proposition}
 
 \begin{proof}
   See \cref{app:propSig}
 \end{proof}
   Together with the constant functional $\calS_{\varnothing} \equiv 1$, we call the collection $\calS = \{\calS_{\alpha}\in \C^{\infty,\infty}(\bar{\Omega}^{\Pi}) \ : \  \alpha \in \A \}$ the \textit{signature} in $\Omega^{\Pi}$. 
   When we are given a specific path $X\in \Omega^{\Pi}_t$, then the \textit{signature of $X$} is understood as the image $\calS(X) =  \{\calS_{\alpha}(X_u) \ : \  u\in [0,t], \ \alpha \in \A \}$.  
Notice that the integrals in the above definition are in the Riemann-Stieltjes  (respectively It\^o) sense when  the integrator (respectively integrand) is of bounded variation. In such situations, the symbol $\circ$ may be removed. 
Assuming $x_0=0$ and writing $\calS_{\alpha}$ instead of $\calS_{\alpha}(X_t)$ for brevity, the first signature functionals are given by 
\begin{align}\label{eq:tree}
  \begin{pmatrix} 
& &  \calS_{\varnothing}& &  \\
&  \calS_{0} & &   \calS_{1} &\\
\calS_{00} &    \calS_{01} & &  \calS_{10}  & \calS_{11}\\
\vdots & \vdots & &\vdots & \vdots
 \end{pmatrix} \;=\; \begin{pmatrix} 
& &  1 & &  \\
&  t & &  x_t &\\
\frac{t^2}{2} &  
 \int_0^t s \, dx_s & & \int_0^t  x_s ds & \frac{x_t^2}{2}
 \\
\vdots & \vdots & &\vdots & \vdots
 \end{pmatrix}.
\end{align} 
As can be seen, 
each entry in the above infinite pyramid 
 generates two descendants by either integrating the former with respect to $t \sim 0$ or $x \sim 1$. Inversely, we gather from \eqref{eq:derivSig} that the functional derivatives permit to retrieve the antecedent by removing the rightmost letter of a given word. 
 \begin{remark}\label{rem:timeResconstruction}
It is essential to enlarge one-dimensional paths with the time itself. Otherwise, the signature would only consists of the rightmost diagonal of the signature tree $\eqref{eq:tree}$, namely 
$\calS_{\mathds{1}_k}(X_t) = \frac{\delta x_t^{k}}{k!}$, $k\ge 0$. Hence, the only information of the path contained in the  signature would be its endpoints. 
 Moreover, adding the increasing path $t\mapsto t$ ensures that the signature map $X \mapsto \calS(X)$ is invertible; see  \cite{Hambly}. In fact, knowing the words $\alpha^{(k)} :=1\boldsymbol{0}_{k+1}\,$,  $k \ge 0$, is enough to show the injectivity of the signature. Indeed, observe that 
 $\calS_{\alpha^{(k)}}(X_t) = \int_0^{t} x_s \frac{(t -s)^{k}}{k!}ds,$ 
 which are scaled $L^2([0,t])$ product of $X_t$ with the time-reversed monomials $s \mapsto (t -s)^k$. As the latter form a complete basis of $L^2([0,t])$, the path can thus be uniquely retrieved from $(\calS_{\alpha^{(k)}}(X_t))_{k\ge 0}$. See Section 2.3 in  \cite{Tissot}
 for further details. 
\end{remark}
 
 \begin{remark}
Contrary to rough path theory \cite{Lyons,FrizHairer} where iterated integrals arise from an imposed algebraic property$-$namely, Chen's identity \cite{Chen}$-$our framework allows to express the signature as  limits of trapezoidal Riemann sums. This exhibits a clear practical advantage as the integrals are directly interpretable. Admittedly, a fair comparison is  difficult  as we here assume one space dimension
and different conditions on $X$  (finite quadratic variation along a sequence versus $p-$variation metric in rough path theory).  
\end{remark}
 
We now collect several properties of signature functionals. 
 
 \begin{proposition} \label{prop:sigProperties}
 Let $\calS$ be the signature in $\Omega^{\Pi}$ and fix two words $\alpha, \gamma$. 
 Then the following properties holds: 
 \begin{itemize}
     \item[(i)]  (\textnormal{\bf Higher derivatives}) $\Delta_{\gamma} \calS_{\alpha} \ne 0$ if and only if  $\alpha = \beta \gamma$ for some $\beta \in \A$. If so, then $\Delta_{\gamma} \calS_{\alpha} = \calS_{\beta}$. 
     \item[(ii)] (\textnormal{\bf Linear independence of signature functionals}) The signature functionals are linearly independent in the sense that   
     $\forall \ \A' \subseteq \A$, $|\A'| < \infty$,  and coefficients $(c_{\alpha})_{\alpha\in \A'}$ such that $ \sum_{\alpha \in \A'} c_{\alpha} \calS_{\alpha} \equiv 0$,  
     then $c_{\alpha}= 0$ $\forall \ \alpha \in \A'$. 
     \item[(iii)] (\textnormal{\bf Linear dependence of signature $T-$functionals}) The signature functionals restricted to $\Omega_T^{\Pi}$ are linearly dependent. 
         \item[(iv)] 
         (\textnormal{\bf Linear independence of signature $T-$functionals}) 
         The family 
     $\{ \calS_{\beta 1}|_{\Omega_T^{\Pi}} \ : \ \beta \in \A \}$ forms a linearly independent system.  
 \end{itemize}
 \end{proposition}

 \begin{proof}
 See \cref{app:sigProperties}. 
 \end{proof}

An alternative formulation of assertion \textit{(ii)} is given in \cite{Fliess81}.  Assertion \textit{(iv)} states that essentially "half" of the words are retained in order to obtain an independent system. In fact, more can be said when considering words up to some fixed length, as stated in the next Theorem. 
\begin{theorem} \label{thm:Basis}
The family $\{\calS_{\beta 1}|_{\Omega_T^{\Pi}} :   |\beta| < K\}$ 
is a basis for the space  spanned by $\{ \calS_{\alpha}|_{\Omega_T^{\Pi}} \ : \ |\alpha| \le K \}$. 
\end{theorem}

\begin{proof}
 See \cref{app:thmBasis}.
\end{proof}
\subsubsection{Main Results} \label{sec:results}  
We first state the functional Taylor expansion (FTE) in  general form (\cref{thm:FTE}) and later discuss some important consequences 
(\cref{cor:FTEX,cor:FME}). 
Multiple derivatives are denoted by  $\Delta_{\alpha} = \Delta_{\alpha_1} \cdots \Delta_{\alpha_k}$. For instance, 
$\Delta_{001}f = \Delta_{t} (\Delta_t(\Delta_x f))$, where we stress again that  $0,1$ is identified with $t,x$, respectively. Besides if $\alpha = \varnothing$, then $\Delta_{\alpha}$ is the identity operator, i.e. $\Delta_{\varnothing}f=f$. 

\begin{theorem}\textnormal{(\textbf{Functional  Taylor expansion (FTE)})}
\label{thm:FTE}
 Let $f\in \C^{K,K+1}(\Lambda)$ and fix two paths $X_s \in \Lambda_s$, $Y_u \in \Omega^{\Pi}_u$ with $s,u \in [0,T]$. 
 Then,  we have, 
\begin{align}
    f(X_{s} \oplus Y_u) &= \sum_{|\alpha|< K}  \Delta_{\alpha}f(X_s) \calS_{\alpha}(Y_u) + R_{K}(X_s,Y_u), \label{eq:FTE1}
\end{align}
with the remainder functional, 
\begin{align}
        R_{K}(X_s,Y_u) &= \sum_{|\alpha| = K} \int_{\triangle_{K,u}} \Delta_{\alpha}f(X_{s}\oplus Y_{t_1}) \circ \, dy^{\alpha}.   \label{eq:FTE2}
\end{align}
\begin{proof}
See \cref{app:FTE}.
\end{proof}

\end{theorem}
We stress that $X_s$ may be completely different from the added path $Y_u$ $-$ see \cref{fig:FTEFlower}. In particular, the signature of $X_s$ may well be undefined. 
When $X_s\in \Omega^{\Pi}$, we can expand $f$ around a portion of $X_s$ itself. This is the essence of the next corollary. 
If  $X\in \Lambda$ is of length at least $t>0$, write $X |_{[s,t]} \in \Lambda_{t-s}$, $ s\le t $, for  the restriction of $X$  to $[s,t]$, i.e.  $X |_{[s,t]}(u) = x_{s+u}$, $u \in [0,t-s]$. 
Also, when $X$ is continuous, we  can define $X |_{[t,s]}$ as 
the time-reversed version of $X |_{[s,t]}$, that is 
$$X |_{[t,s]}(u) = \overleftarrow{X |_{[s,t]}}(u)  :=  x_{t-u} , \quad  u \in [0,t-s]. $$
Note that the continuity of  $X$ is essential as the time-reversed path would be otherwise strictly càglàd, hence no longer in $\Lambda$. 
We adopt the convention that $X_t \oplus X|_{[t,s]} = X_s$ so concatenating a time-reversed path reduces the length.  

\begin{corollary}\label{cor:FTEX}
Let $X\in \Omega^{\Pi}_t$ and $f\in \C^{K,K+1}(\bar{\Omega}^{\Pi})$. Then for all $ s<t$, 
\begin{align}
    f(X_{t}) &= \sum_{|\alpha|< K}  \Delta_{\alpha}f(X_s) \calS_{\alpha}(X |_{[s,t]}) + \sum_{|\alpha| = K} \int_{\triangle_{K,t-s}} \Delta_{\alpha}f(X_{s+t_1}) \circ \, dx_{s+\cdot}^{\alpha}\; . \label{eq:CorFTE1}
\end{align}
Also, swapping the roles of $s$ and $t$ gives the anticipative expansion,
\begin{align}
    f(X_{s}) &= \sum_{|\alpha|< K}  \Delta_{\alpha}f(X_t) \calS_{\alpha}(X |_{[t,s]}) + \sum_{|\alpha| = K} \int_{\triangle_{K,t-s}} \Delta_{\alpha}f(X_{t-t_1}) \circ \, dx_{t-\cdot}^{\alpha}\; . \label{eq:CorFTE2}
\end{align}
\end{corollary}
\begin{proof} See \Cref{app:FTEX}.
\end{proof}

\begin{figure}[t]
\caption{Illustration of the pertubation paths ($Y$) added to the source path ($X$) appearing in the functional Taylor expansion (\cref{thm:FTE}).}
\vspace{-2mm}
    \centering    \includegraphics[height=1.8in,width=2.5in]{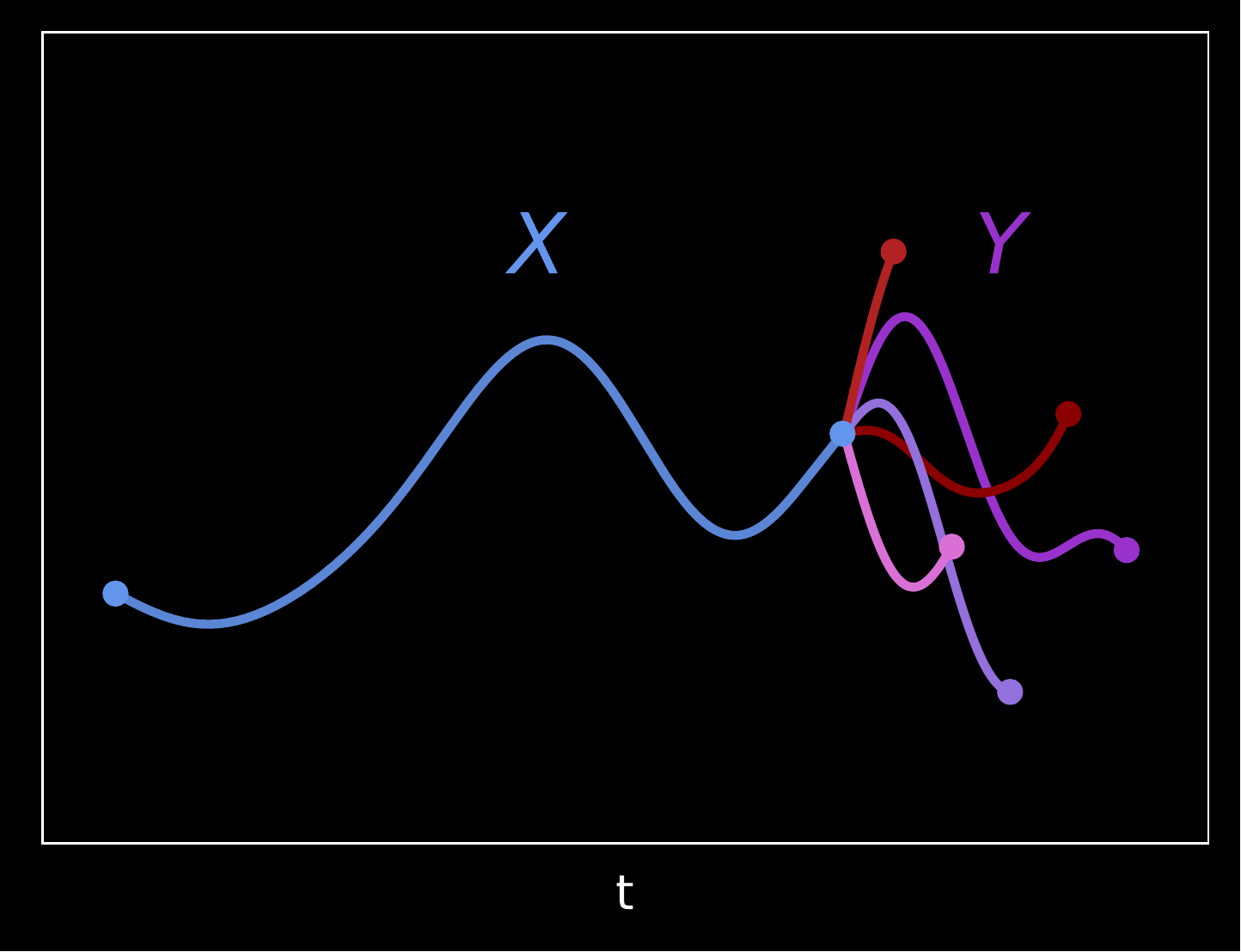}
    \label{fig:FTEFlower3D}
\label{fig:FTEFlower}
\end{figure}

The expansion $\eqref{eq:CorFTE1}$ is in appearance similar to \citet[Theorem 14]{LittererOberhauser}. The main difference is that our expansion is made pathwise.  
Also, the word enumeration differs: \citet{LittererOberhauser} use the weighted length $\lVert \alpha \rVert = 2|\alpha|_0 + |\alpha|_1$ to truncate the expansion. This is to reflect the "$\sqrt{t}$ scale" of It\^o diffusions. This partition of the words will be used in \cref{sec:dynamicChaos} when comparing the chaos expansion with the FTE.

We finally expand a functional around the initial value of the path. Borrowing the terminology from differential calculus, we call such a series a  \textit{functional  Maclaurin expansion}.  

\begin{corollary}\textnormal{(\textbf{Functional  Maclaurin expansion})} \label{cor:FME}
Let $f$ be a functional in $\C^{K,K+1}(\bar{\Omega}^{\Pi})$. Then for all $t\in [0,T]$, the following expansion holds,
\begin{align}
    f(X_t) &= \sum_{|\alpha|< K}  \Delta_{\alpha}f(X_0) \calS_{\alpha}(X_t) + R_{K}(X_t), \quad X_t \in \Omega_t^{\Pi}, \label{eq:FME} \\
    R_{K}(X_t) &= \sum_{|\alpha| = K} \int_{\triangle_{K,t}} \Delta_{\alpha}f(X_{t_1}) \circ \, dx^{\alpha}. \label{eq:FMER}
\end{align}
\end{corollary}
\begin{remark}
    The functional Maclaurin expansion (or more broadly, the FTE) presents obvious benefits to approximate functionals. 
Indeed, an estimate of $f \in \C^{K,K+1}(\bar{\Omega}^{\Pi})$ is given by the truncation,  
\begin{equation*}
     f^{K}(X_t) := \sum_{|\alpha|\, \le \, K} \Delta_{\alpha}f (X_0) \calS_{\alpha}(X_t).
\end{equation*}\\[-0.5em]
    In other words, $f^{K}(X)$ approximates the transformed path $ f(X)$ by a linear combination of signature elements. \cref{fig:FTE} provides an illustration. 
    The latter is referred to  as a \textit{polynomial functional}  \cite{LittererOberhauser} and in finance, as a  \textit{signature payoff}  \cite{LyonsNum}   or simply \textit{sig-payoff} \cite{CuchieroSig}. In calculus and differential geometry, $f^K$ may also be regarded  as the $K-$\textit{jet} of $f$. 
By applying the Stone–Weierstrass theorem, it can be shown that signature payoffs are dense in the space of continuous functionals  restricted to a compact subset of $\Lambda$. 
Among others,  the result is proved in \citet{LittererOberhauser} for continuous paths  
of bounded variation and in \citet{CuchieroSig} for paths emanating from a continuous semimartingale.  
However, the Stone-Weierstrass theorem only guarantees the existence of a signature payoff arbitrarily close to a functional. The functional Taylor series, on the other hand, makes the approximation explicit.
\end{remark}

We now discuss several examples.
\begin{example}
  Let $f(X_t) = \sqrt{2}\int_0^t \sin(x_s)ds$ and assume that $x_0 = \pi/4$ for reasons that will become clear below. Following from \Cref{ex:IntDev}, we obtain $\Delta_t f =\sqrt{2}\sin(x_t)$ and $\Delta_x f =0$. The only non-zero,  second order derivative   is therefore  $\Delta_{xt} f = \sqrt{2}\cos(x_t)$. Then \Cref{cor:FME} gives the approximations
  \begin{align*}
      f^1(X_t) &= \Delta_t f(X_0) \calS_0(X_t) = \sqrt{2}\sin(\pi/4) t =  t, \\
      f^2(X_t) &= f^1(X_t) + \Delta_{tx}
      f(X_0) \calS_{10}(X_t) = f^1(X_t) +\sqrt{2}\cos(\pi/4)\calS_{10}(X_t) = \int_0^t (1+x_s) ds.
  \end{align*}
\cref{fig:FTE} provides an illustration in the $(t,X,f(X))$ space for two  trajectories of different lengths.  
\end{example}

\begin{figure}[H]
    \centering
        \caption{Taylor approximations 
        of $f(X_t) = \sqrt{2}\int_0^t \sin(x_s)ds$   of order $1$ and $2$, giving $f^1$, $f^2$, respectively. Two realisations are shown in the $(t,X,f(X))$ space.  }
    \includegraphics[height=2.0in,width=2.8in]{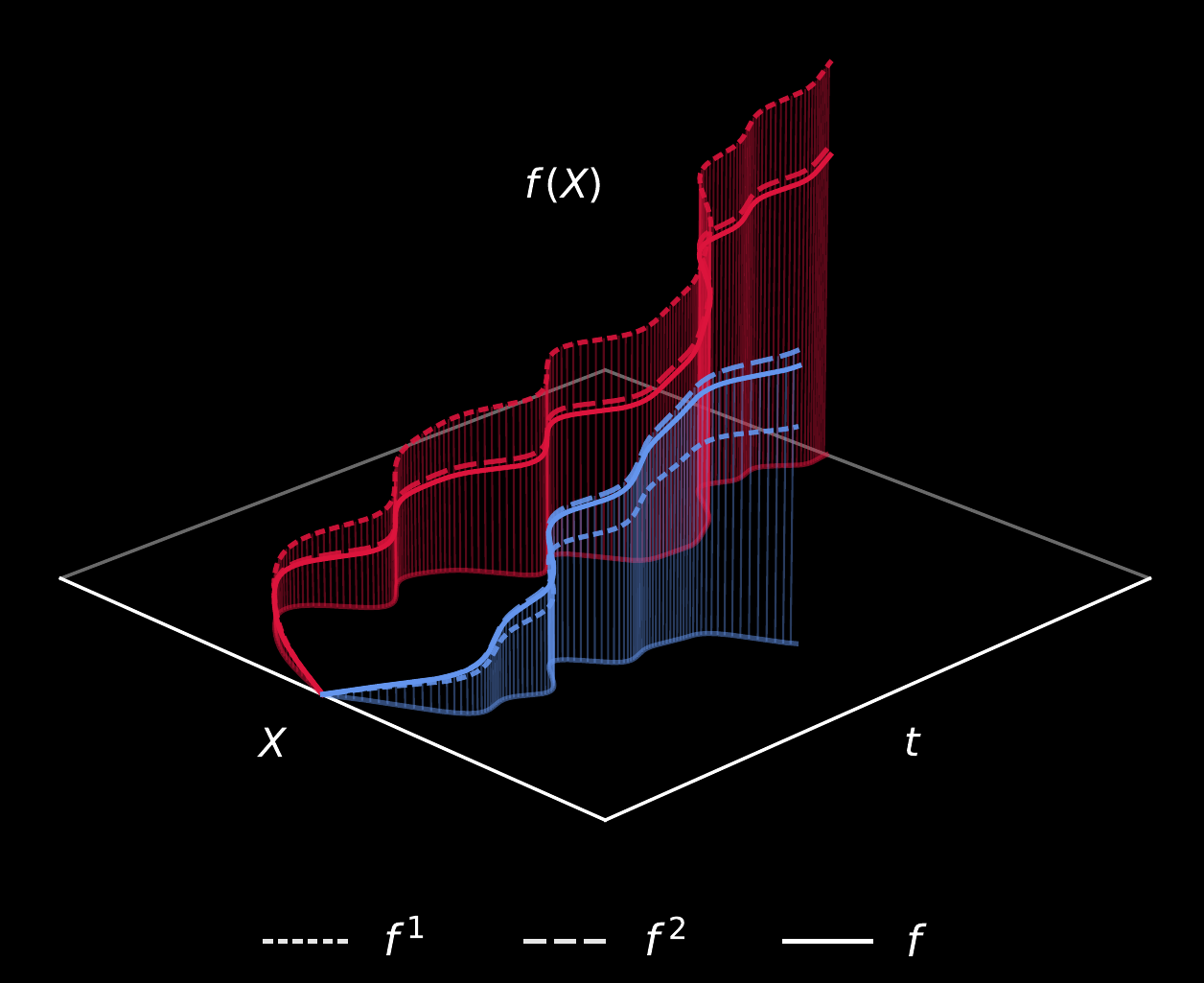}
    \label{fig:FTE}
\end{figure}

\begin{example} \label{ex:classicTaylor}\textbf{(Taylor expansion)}
Let $h:\R\to \R$ be an analytic European payoff.
Going back to the fundamentals, Taylor's theorem gives 
\begin{equation}\label{eq:stdTaylor}
    h(x_t)= \sum_{k=0}^{K-1} h^{(k)}(x_0) \frac{(x_t-x_0)^k}{k!} + R_K(x_t), \quad R_K(x_t)= o(|x_t-x_0|^{K-1}).
\end{equation}
 It is often overlooked that the scaled monomials appearing in the expansion are in fact iterated integrals. Indeed, observe that 
$\frac{(x_t-x_0)^k}{k!} = \int_0^t \cdots \int_0^{t_2} \circ \, dx_{t_1}\ldots \circ dx_{t_k} = \calS_{\mathds{1}_k}(X_t). $ 
Moreover, the integral form of the remainder allows us to write
\begin{align*}
    R_K(x_t) = \int_{x_0}^{x_t} h^{(K)}(z) \frac{(z-x_0)^{K-1}}{(K-1)!} dz 
    = \int_{\triangle_{K,t}} h^{(K)}(x_{t_1}) \circ \, dx^{\mathds{1}_k}.
\end{align*}
On the other hand, we can embed $h$ into the space of functionals by defining
$f(X_t) = h(x_t)$. Of course, $\Delta_x f = \partial_x h$ and $\Delta_t f = \partial_t h = 0$. Hence $f\in \C^{K,K+1}$ and  \Cref{cor:FME} yields
\begin{align*}
    f(X_t)  
    &= \sum_{k=0}^{K-1} \Delta_{\mathds{1}_k} f(X_0)\calS_{\mathds{1}_k}(X_t) +  \int_{\triangle_{K,t}} \Delta_{\mathds{1}_k} f(X_{t_1}) \circ \, dx^{\mathds{1}_k},
\end{align*}
which is consistent with $\eqref{eq:stdTaylor}$. We conclude that the elements $\{\calS_{\varnothing}, \calS_{1},\calS_{11},\calS_{111}, \ldots\}$ 
suffices to replicate European payoffs. Equivalently, they allow to  recover the volatility surface. 
\end{example} 

\begin{example} \label{ex:stochTaylor} \textbf{(Stochastic Taylor expansion)}
We slightly digress from our deterministic setting and verify that \Cref{thm:FTE} is a generalization of 
the stochastic Stratonovich-Taylor expansion  \cite{BenArous,KP}. 
For simplicity, let us stick to the one-dimensional case and consider the nonautonomous Stratonovich stochastic differential equation, 
\begin{equation}\label{eq:SDE}
    dz_t = a(t,z_t)dt + b(t,z_t) \circ dx_t, \quad z_0 \in \R,
\end{equation}
for smooth (hence Lipschitz) measurable functions $a,b: [0,T]\times \R \to \R$, and a  standard Brownian motion $X$. If we further assume that $a,b$ have 
 at most linear growth, there exists a unique strong solution of $\eqref{eq:SDE}$ \cite[Theorem~4.5.3.]{KP}.
 Moreover, the adaptedness of $Z$ 
 ensures the existence of a functional $\calI:\Lambda\to \R-$the \textit{Itô map}$-$such that 
\begin{equation}\label{eq:SDEFunc}
    z_t = \calI(X_t)  = z_0 + \int_0^t a(s,\calI(X_s))ds +  \int_0^t b(s,\calI(X_s)) \circ dx_s.
\end{equation}
Given $h \in \calC^{K,K+1}([0,T],\R)$ and $s\le t$, the Stratonovich-Taylor expansion  reads
\begin{align}\label{eq:STE}
    h(t,z_t) &= \sum_{|\alpha|< K}  \calV_{\alpha}h(s,z_s) \calS_{\alpha}(X_t) + \sum_{|\alpha| = K} \int_{\triangle_{K,t-s}} \calV_{\alpha}h(s+t_1,z_{s+t_1})\circ \, dx^{\alpha},
\end{align}
with the vector fields 
$\calV_0 = \partial_t + a \,\partial_z$, $\calV_1=b\, \partial_z$, and $\calV_{\alpha} =\calV_{\alpha_1}\cdots \calV_{\alpha_k}$. Next, define the functional $f(X_t) = h(t,\calI(X_t))$. Comparing $\eqref{eq:STE}$ with  \Cref{cor:FTEX}, we gather that  $\calV_{\alpha}h = \Delta_\alpha f$ must hold. Even though the regularity of $f$ is a priori not guaranteed, we will see that the smoothness of $h$, $a$, $b$ carries over to $f$. 
It suffices to show that $\calV_\alpha$ acts on $h$ exactly as $\Delta_\alpha$ acts on $f$ for $\alpha=0,1$. Indeed, the same would be true for longer words due to the similar recursive structure of the differential operators. If $\alpha=0$, the chain rule, equation $\eqref{eq:SDEFunc}$ and Example \ref{ex:IntDev} give respectively 
\begin{align*}
    \Delta_{t}f = \partial_t h + \partial_z h \, \Delta_t\! \calI = (\partial_t  + a \, \partial_z)h = \calV_0 h.
\end{align*}
As $b$ is smooth and  $X$ continuous $\Q-$a.s., we have $\Delta_x\calI=b$.  
Hence   $\Delta_{x}f = \partial_z h \, \Delta_x \!\calI = b \, \partial_z h = \calV_1 h.$ 
In conclusion, the two expansions coincide as claimed.
\end{example}

\subsubsection{Remainder Estimates and Real Analytic Functionals}\label{sec:realAnalytic}
First, we give in \cref{prop:remainder} 
an  upper bound for the remainder of the FTE.  
In the diffusion case, one typically study the convergence of (functional) Taylor expansions in the $L^2(\Q)$ sense; see  \citet[Section 5.10.]{KP} or \citet[Section 5.10.]{LittererOberhauser}.  
Indeed, It\^o's isometry greatly  facilitates manipulation 
of expected iterated integrals.  
However, this powerful tool is no longer available when one is interested in the pathwise convergence of the Taylor expansion. 
Another avenue pursued in  \citet{BMZ} consists of applying the Burkholder–Davis–Gundy inequality several times to obtain  an "almost sure" upper bound of the remainder. 
This procedure, although powerful,  
still requires  a stochastic framework. 
We here propose a purely pathwise remainder estimate. 
For simplicity, we establish the result for the functional Maclaurin expansion as in \cref{cor:FME}. 
To this end, we associate to each path $X_t \in \Lambda$ 
the  seminorm\footnote{Note that $[x_0, x_s]$ is understood as $[x_0\wedge x_s, x_0\vee x_s]$ throughout. We also recall that $X_s^{(\varepsilon)}= X_s^{\varepsilon-x_s}$ (so that $ X_s^{(\varepsilon)}(s) = \varepsilon$) as seen in \cref{sec:pathStrat}.} 
\begin{equation}\label{eq:funcNorm}
    \lVert \varphi \rVert_{X_t} = \sup_{s\le t} \sup_{\varepsilon \in [x_0, x_s]}|\varphi(X_s^{(\varepsilon)})|, \quad  \varphi:\Lambda \to \R. 
\end{equation}
As $\{X_s^{(\varepsilon)} \ : \  \varepsilon \in [x_0,x_s], \ s\le t\}$ is a compact subset of $\Lambda$, then $\lVert \varphi \rVert_{X_t} < \infty$ whenever $\varphi$ is $\Lambda-$continuous. Lastly, let $\lVert X_t \rVert_{\infty} = \sup_{s\le t} |x_s|$ and assume without loss of generality that $x_0 =0$.  

\begin{proposition}\label{prop:remainder} \textnormal{\bf(Remainder Estimate)} Let $X_t\in \Omega_t^{\Pi}$ and $f\in \C^{\infty,\infty}(\bar{\Omega}^{\Pi})$. For $K\ge 1$, consider the functional Maclaurin expansion 
 $$f(X_t) = \sum_{|\alpha| < K} \Delta_{\alpha}f(\boldsymbol{0}) \calS_{\alpha}(X_t) + R_K(X_t). $$ Moreover, write 
$r_{\alpha}(X_t) = \int_{\triangle_{K,t}}\Delta_{\alpha}f(X_{t_1}) \circ dx^{\alpha},$
so that 
$R_K(X_t) = \sum_{|\alpha| = K} r_{\alpha}(X_t).$ Let $|\alpha|_{01}$ be the number of occurrences of the word $01$ in $\alpha$. 
Then for each $|\alpha| = K$,  we have 
\begin{equation}\label{eq:r_alpha}
    |r_{\alpha}(X_t)| \le \frac{2^{|\alpha|_{01}}}{|\alpha|_0!}  c_{\alpha} \  t^{|\alpha|_0}\lVert X_t \rVert_{\infty}^{|\alpha|_1},
\end{equation}
with $c_{\alpha} = \lVert \Delta_{\alpha}
      f\rVert_{X_t} + (t \vee 1) \ \lVert \Delta_{0\alpha} f\rVert_{X_t}$.  
     Letting $C_{K} = \max_{|\alpha|=K} c_{\alpha}$ and $\rho(X_t) = 2(t \vee  \lVert X_t \rVert_{\infty})$, then the remainder is bounded by
     \begin{equation}\label{eq:remainder}
         |R_{K}(X_t)|\le C_{K} \rho(X_t)^K.
     \end{equation} 
\end{proposition}

\begin{proof} See \cref{app:remainder}
\end{proof}

In light of $\eqref{eq:remainder}$ and the definition of $\rho(\cdot)$, we conclude that the functional Taylor expansion converges to its functional Taylor series whenever both $t$ and $\lVert X_t \rVert_{\infty}$ are strictly smaller than $1/2$. As $X_t$ is assumed to be uniformly continuous  in \cref{prop:remainder}, we can therefore  find $t$ sufficiently small such that the FTE converges. 

Note that \cref{prop:remainder} gives an error bound of the FTE at a single path. 
Shifting gears, we now investigate functionals such that its FTE converges for \textit{every} path inside a neighborhood of some fixed $X \in \Lambda$. 
 In line with the terminology from calculus, we shall qualify these functionals as \textit{real analytic}. 

What remains is to choose an adequate topology on an appropriate subspace of $\Lambda$. 
Let $\Omega^{\text{Lip}}_t 
\subseteq \Omega^{\Pi}_t$ be the subspace of Lipschitz   paths, that is $Y \in \Omega^{\text{Lip}}_t$ if and only if $[Y]_{\text{Lip}} := \sup_{ u\neq s} \frac{|y_s - y_u|}{|s-u|}<\infty$. As usual, we write $\Omega^{\text{Lip}} = \bigcup_{t\in [0,T]}\Omega^{\text{Lip}}_t$. As Lipschitz paths 
have 
finite variation, the signature is therefore defined without difficulty as Riemann-Stieljes integrals. Moreover, one has the upper bound, 
\begin{equation}\label{eq:sigLipschitz}
    |\calS_{\alpha}(Y_t)| \le [Y_t]_{\text{Lip}}^{|\alpha|_1}\frac{t^{K} }{K!}, \quad |\alpha| = K. 
\end{equation}
For $s \le t$ and Lipschitz paths $Y_t,Y_s' \in \Omega^{\text{Lip}}$, define the metric
$$d_1(Y_t,Y_s')  = t-s + |y_0-y_0'|+ [Y_t - Y_{s,t-s}']_{\text{Lip}},$$ 
and let $B_{\delta}^{\text{Lip}}(Y)$ be the Lipschitz ball centered at $Y\in \Omega^{\text{Lip}}$ of radius $\delta $, i.e. the set of paths $Y' \in \Omega^{\text{Lip}}$ such that $  d_{1}(Y,Y') < \delta$. 
When $Y = \boldsymbol{0} \in \Lambda_0$, we  write $B_{\delta}^{\text{Lip}}$ instead of $B_{\delta}^{\text{Lip}}(\boldsymbol{0})$. 

\begin{definition} Let $\Lambda' \subseteq \Lambda$. 
A functional $f \in \C^{\infty,\infty}(\Lambda')$ is \textit{real analytic at} $X\in \Lambda'$,  if its functional Taylor expansion converges in a neighborhood of $X$, i.e. there exists $\delta >0$ such that 
\begin{equation}\label{eq:analytic}
    f(X\oplus Y) = \sum_{\alpha \in \A}\Delta_{\alpha}f(X) \calS_{\alpha}(Y)  := \lim_{K\to \infty} \sum_{|\alpha|\le K} \Delta_{\alpha}f(X) \calS_{\alpha}(Y) \quad \forall \ Y\in B_{\delta}^{\text{Lip}}.  
\end{equation}
The largest $\delta \ge 0 $ such that $\eqref{eq:analytic}$ holds is called the \textit{radius of convergence} of the functional series of $f$ centered at $X$ and is denoted by $\varrho(f,X)$. 
Finally, we say that  $f$ is \textit{real analytic} (notation: $f\in \C^{\omega}(\Lambda')$) if $f$ is real analytic at every $X \in \Lambda'$.
\end{definition}

We emphasize that the center path $X$ in the above definition need not be Lipschitz; what really matters is the regularity of the concatenated path $Y$. Nevertheless, the functional has to be defined at $X \oplus Y$ which may restrict the subspace $\Lambda'\subseteq \Lambda$ as we shall see in the next example. 
\begin{example}\label{def:realanalytic}
It is easily seen that every signature functional belongs to $\C^{\omega}(\Lambda')$ with $\Lambda' = \bar{\Omega}^{\Pi}$. Indeed for any word  $\alpha \ne \varnothing$ and $X\in \Omega^{\Pi}$, Chen's identity \cite{Chen} reads 
\begin{equation}\label{eq:ChenIdentity}
    \calS_{\alpha}(X\oplus Y) = \sum_{\beta   \gamma = \alpha}  \calS_{\beta}(X)\calS_{\gamma}(Y) \quad \forall \ Y \ \in \Omega^{\text{Lip}}. 
\end{equation}
Since $ \Delta_{\gamma} \calS_{\beta  \gamma}(X) = \calS_{\beta}(X) $ (see \cref{prop:sigProperties} (i)) we conclude that the right side of $\eqref{eq:ChenIdentity}$ is the Taylor series of $\calS_{\alpha}(X\oplus Y)$ and in turn that $\calS_{\alpha} \in \C^{\omega}(\bar{\Omega}^{\Pi})$. Moreover,  
the radius of convergence at $X$ is infinite, i.e. $\varrho(\calS_{\alpha},X) = \infty$. 
\end{example}

Finally, we provide an explicit upper bound for the remainder term of  smooth functionals with Lipschitz pertubation paths and a characterization for the corresponding radius of convergence. 
\begin{proposition} \label{prop:remainderLip} Let $\Lambda' \subseteq \Lambda$, $X\in \Lambda'$,  and  $f\in \C^{K,K+1}(\Lambda')$. Suppose that there exist constants $C_1, C_2 \ge 0$ such that $ |\Delta_{\alpha}f(X\oplus Y')| \le C_1^{K}(K!)^{C_2}$ for all $|\alpha| = K$ and $Y' \in \Omega^{\textnormal{Lip}}$. Then, the remainder $\eqref{eq:FTE2}$ of the FTE of $f$ satisfies   
\begin{align}\label{eq:LipRemainder}
    |R_K(X,Y_u)| &\le (K!)^{C_2-1}  \rho(Y_u)^K, \quad Y_u\in \Omega_u^{\textnormal{Lip}}, \quad u\ge 0, 
\end{align}
with $\rho(Y_u) = 2 \ C_1  ([Y_u]_{\textnormal{Lip}} \vee 1) u$. If $f$ is also smooth, then it is real analytic if and only if $C_2 \le 1$, with 
 radius of convergence 
 $$\varrho(f,X)\ge  \frac{1}{2C_1}\wedge 1\; \  \text{ if $C_2 < 1$}, \quad \varrho(f,X) = \infty \; \ \text{ if $C_2 = 1$}.$$
\end{proposition}

\begin{proof} For all $|\alpha| = K$, note that $|\int_{\triangle_{K,u}} \Delta_{\alpha}f(X\oplus Y_{t_1}) dy^{\alpha}| \le C_1^{K}(K!)^{C_2} [Y]_1^{|\alpha|_{\textnormal{Lip}}} \frac{u^K}{K!}$  using $\eqref{eq:sigLipschitz}$ and the assumption on $\Delta_{\alpha}f$. Hence $\eqref{eq:LipRemainder}$ follows from  $[Y]_1^{|\alpha|_1} \le ([Y]_{\textnormal{Lip}} \vee 1)^K$ and the triangular inequality. Now if $C_2 = 1$ and $\delta = \frac{1}{2C_1}\wedge 1$, then  for all $Y_u \in B_{\delta}^{\text{Lip}}$,  $[Y_u]_{\textnormal{Lip}} < \delta \le 1 $ and similarly $u<1$. This implies that $\rho(Y_u) = 2C_1 u < 1$ as desired. The cases $C_2 < 1$ and $C_2 > 1$ are obvious. 
\end{proof}

 \begin{remark}
Alternatively, one can  choose  $\Omega^{\gamma}_t = \calC^{0,\gamma}([0,t],\R)$, the space of $\gamma-$H\"older paths. When $\gamma>1/2$, the signature elements are well-defined using Young integration \cite{Young}. More generally, one can consider paths of finite $p-$variation with $p\in (1,2)$ and use estimates of the iterated iterated as in \citet{Lyons}. In this work, we nevertheless 
stick to Lipschitz paths for ease of presentation. 
\end{remark}

\subsubsection{Connections with Static Expansions}\label{sec:connectStatic}
{\textbf{Intrinsic Value Expansion} (\cref{sec:IVE}).} \label{sec:connectionIVE}
Let $g$ be a  smooth $T-$functional and consider its  IVE,  
\begin{equation}\label{eq:IVEComp}
    g(X_T) = g(X_{0,T}) + \sum_{k= 1}^{\infty} I_k\varphi_k (X_T),\qquad \ X_T \in \Omega^{\Pi}_T,
\end{equation}
where we define $ \varphi_k(t_1,\ldots,t_k) =  D_{t_1 \cdots t_k}g(X_{0,T})$ and    $I_k\varphi_k (X_T) =  \int_{\triangle_{k,T}} \varphi_k \, \circ dx^{\otimes k }$. 
 On the other hand, assuming that the associated  intrinsic functional $\iota_{0}g(X_t) = g(X_{t,T-t})$ is real analytic, one can compute its functional Maclaurin series. In particular, for $X_T \in \Omega^\Pi_T$, 
  \begin{align}\label{eq:FTEComp}
      g( X_T)  = \iota_{0}g( X_T) = \underbrace{\iota_{0}g(\boldsymbol{0})}_{= \  g(X_{0,T})} +  \sum_{\alpha \in \A} \Delta_{\alpha} \iota_{0}g (\boldsymbol{0}) \calS_{\alpha}(X_T). 
  \end{align}
  The aim of this section is to show how $\eqref{eq:IVEComp}$ relates to $\eqref{eq:FTEComp}$. One way is to derive the intrinsic value expansion of each signature $T-$functional $\calS_{\alpha}|_{\Lambda_T}$. In other words, we seek $(\varphi_k^\alpha)_{k\ge 1}$ such that 
  \begin{equation*}
      \calS_{\alpha}(X_T) = \calS_{\alpha}(X_{0,T}) + \sum_{k\ge 1} I_k \varphi_k^{\alpha}(X_T). 
  \end{equation*}
  First, we conclude from a homogeneity argument that $\varphi^{\alpha}_k=0$ unless $k= |\alpha|_1$. Indeed, defining the scaled path $\lambda X_T= (\lambda x_t)_{t\in [0,T]}$, $\lambda >0$, we must have that 
 $$  \lambda^{|\alpha|_1}\calS_{\alpha}(X_T)= \calS_{\alpha}(\lambda X_T) = \calS_{\alpha}(X_{0,T}) + \sum_{k\ge 1} I_k \varphi_k^{\alpha}(X_T)\lambda^k, \qquad  \forall \ \lambda >0.  $$
The proposition gives an explicit formula for  $(\varphi^{\alpha}_k)$; see the first example in \cite{Fliess81} for a similar expression in the case of absolutely continuous paths.  
  \begin{proposition}\label{prop:IVEPara}
  Let $\alpha \in \A$ and $k= |\alpha|_1$. Moreover, express  $\alpha$ as 
  \begin{equation*}
      \alpha = \boldsymbol{0}_{\gamma_0}1 \boldsymbol{0}_{\gamma_1}1 \cdots   1 \boldsymbol{0}_{\gamma_k},
  \end{equation*}
  where  we allow $\gamma_l=0$ ($\Longrightarrow \ \boldsymbol{0}_{\gamma_l} =\varnothing $) for $l=0,\ldots,k$.  
  Then  $\calS_{\alpha}(X_T) = I_k \varphi_k^{\alpha}(X_T)$ with the kernel  
 \begin{equation}\label{eq:kernelSig}
    \varphi_k^{\alpha}(t_1,\ldots,t_k) = \prod_{l=0}^k \frac{(t_{l+1}-t_{l})^{\gamma_l}}{\gamma_l!}, \quad (t_{0}=0, \ t_{k+1}=T). 
 \end{equation}
 
 \end{proposition}
\begin{proof}
The case $k=0$ is trivial. For $k\ge 1$, the result can be established by computing the $k-$th order Malliavin derivatives of $\calS_{\alpha}$. However, this would be tedious.  Alternatively, we can  use Fubini's theorem iteratively to rearrange the order of integration. Writing $\alpha = \beta 1 \boldsymbol{0}_{\gamma_k}$ with $\beta = \boldsymbol{0}_{\gamma_0}1 \boldsymbol{0}_{\gamma_1}1 \cdots   1 \boldsymbol{0}_{\gamma_{k-1}}$, observe that 
\begin{align*}
    \calS_{\alpha}(X_T) = \int_{\triangle_{\gamma_k+1,T}} \calS_{\beta}(X_{t_0})\circ dx_{t_0}  dt^{\otimes \gamma_k}
    = \int_{0}^T \calS_{\beta}(X_{t_0})\frac{(T-t_0)^{\gamma_k}}{\gamma_k!} \circ dx_{t_0}.  
\end{align*}
The result thus follows from a  simple induction on $k$. 
\end{proof}

It is worth noting that the kernel of signature functionals does not depend on the precise location of the time points. Indeed, we see from $\eqref{eq:kernelSig}$ that $\varphi_k^{\alpha}$ only entails time increments. 
Finally, we  connect the kernels $\varphi_k$ of $g$ with the signature kernels $\varphi_k^{\alpha}$.  

\begin{corollary}\label{cor:IVEPara}
Let $g$ be a smooth $T-$function with analytic intrinsic functional $\iota_{0}g$. Then the kernel of the intrinsic value expansion can be expressed as, 
\begin{equation*}
    D_{t_1 \cdots t_k}g(X_{0,T}) = \sum_{|\alpha|_1 = k } \Delta_{\alpha}\iota_{0}g(X_0)\prod_{l=0}^k \frac{(t_{l+1}-t_{l})^{\gamma_l}}{\gamma_l!}.
\end{equation*}
\end{corollary}

\begin{proof} First, we recall from $\eqref{eq:IVEComp}$ and $\eqref{eq:FTEComp}$ that 
\begin{equation}\label{eq:IVEFTE}
    \sum_{k\ge 1}  I_k  \varphi_k (X_T) = g(X_T) - g(X_{0,T})  =    \sum_{\alpha \in \A \setminus \{\varnothing\}} \Delta_{\alpha} \iota_{0}g (\boldsymbol{0}) \calS_{\alpha}(X_T) \\
    =  
     \sum_{k\ge 1} \sum_{|\alpha|_1 = k } \Delta_{\alpha} \iota_{0}g (\boldsymbol{0}) I_k \varphi_k^\alpha (X_T).
\end{equation}
Using the linearity of $I_k$ and the fact that $\eqref{eq:IVEFTE}$ must hold for all $X_T\in \Lambda_T$, then  $\varphi_k=\sum_{|\alpha|_1 = k } \Delta_{\alpha} \iota_{0}g (\boldsymbol{0})  \varphi_k^\alpha$.  
We finally invoke the expression of $\varphi_k^\alpha$ given in  \cref{prop:IVEPara} to conclude. 
\end{proof} 
{\textbf{Wiener Chaos  }(\cref{sec:chaos}).}  
We show how the chaos expansion can be transformed into the FTE  for  path-independent $T-$functionals. Let    $\Q$ be the Wiener measure as usual and $g(X_T)=h(x_T)$ for some real analytic function $h:\R\to \R$. 
The following result provides a characterization of path independence for $T-$functionals. In passing, it brings us one step closer to the functional Taylor expansion. We refer the reader to \cref{sec:MROperator} for the definition of the operator $\calM$ and its domain $\calD_{\! \calM}$ (given in \eqref{eq:domain}).  
\begin{proposition}\label{prop:Martingale}
Let $g$ be a $T-$functional such that $\calM_T g \in \calD_{\! \calM}$  
and consider its price functional $f = \iota_{\Q}g$ introduced in \cref{ex:pricefunctional}.  Then the following are equivalent: 
\begin{itemize}
\item[(i)] $g$ is path-independent. 
\item[(ii)]  $\Delta_{\mathds{1}_k} f(X)$ is an $\Q-$martingale for all $k\ge 0$. 
\item[(iii)]  The Wiener chaos expansion writes $g(X_T) = \sum_{k=0}^{\infty}\Delta_{\mathds{1}_k}f(\boldsymbol{0}) J_k(X_T)$.
\end{itemize}
\end{proposition}

\begin{proof}
See \cref{app:Martingale}.
\end{proof}
\begin{remark}
  
  Although $f = \iota_{\Q} g$ in  \cref{prop:Martingale} is termed price functional, the spatial derivatives $(\Delta_{\mathds{1}_k} f(X))_{k\ge 0}$ may not
  correspond to Greeks in derivatives pricing. 
Indeed, 
$X$ has to be interpreted as the source of risk, which usually differs from the underlying. 
For instance, a call option in the Black-Scholes with zero dividend and interest rate would read $g(X_T) = (\calI(X_T) -K)^+$, where $X$ is Brownian motion and $\calI$  the  It\^o map $\calI(X_T) = x_0 e^{\sigma x_T - \frac{1}{2}\sigma^2T}$.  
Exceptions include the Bachelier model, where the (scaled) coordinate process $X$ is the asset itself. In this case,  $(\Delta_{\mathds{1}_k} f(X))_{k\ge 0}$   indeed correspond to  the option Delta, Gamma, and so on. 
\end{remark}
Next, we express $k-$fold It\^o iterated integrals as  linear combinations of finitely many signature functionals. 
First, we recall that
 \begin{equation}\label{eq:Hermite}
    J_{k}(X_t) = h_k(t,x_t), \quad  h_k(t,x) =  \frac{t^{k/2}}{k!}H_k\left(\frac{x}{\sqrt{t}}\right),  
 \end{equation}
where $H_k$ is the $k-$th probabilist's Hermite polynomial; see  \cite{DiNunno}. 
Although $X_t\mapsto J_k(X_t)$ is indistinguishable from the time-space Hermite polynomial $h_k(t,x_t)$, their functional derivatives may differ as nicely 
explained in \citet{Oberhauser}. 
However, using the functional It\^o formula and the uniqueness of the  Doob-Meyer decomposition, the operator $\calL := \Delta_t +  \frac{1}{2}\Delta_{xx},$ must act in the same way,  regardless of the representation chosen; see again 
\cite{Oberhauser}. 
For obvious reasons, $\calL$ is often called the   \textit{causal}  \textit{heat operator}. 
The next simple lemma relates the It\^o iterated integrals to the signature. 

\begin{lemma}\label{lem:hermite}
The It\^o iterated integrals are smooth functionals and admit the Maclaurin expansions 
  \begin{align*}
      J_{k}(X_t)
      = \sum_{\alpha \in \A_k} (-2)^{-|\alpha|_0}  \calS_{\alpha}(X_t),  \quad k \ge 0,
 \end{align*}
with  the disjoint subsets 
 $\A_k = \left\{\alpha \in \A \ : \ \lVert \alpha\rVert = k \right\}$, $\ \lVert \alpha\rVert := 2|\alpha|_0 + |\alpha|_1$. 
\end{lemma}

\begin{proof} See \cref{app:hermite}.
\end{proof}
We observe from \cref{lem:hermite} that  It\^o iterated integrals induce a partition of the words in $\A$ 
based on the "weighted length" $\lVert \cdot \rVert$. Note that the  latter arises naturally when establishing convergence results about iterated It\^o or Stratonovich integrals  \cite{KP,BenArous,LyonsVictoir}. 
Finally, we bridge the gap between the FTE and Wiener-It\^o chaos expansion in this simple case. 
\begin{proposition}\label{prop:connectionIndep}
Let $g$ be an analytic path-independent $T-$functional with price functional $f = \iota_{\Q} g$.  Under the assumptions of  \Cref{prop:Martingale}, then $\Delta_{\mathds{1}_k} f(\boldsymbol{0}) (-2)^{-|\alpha|_0}=\Delta_{\alpha}f(\boldsymbol{0}) $. Consequently, 
   \begin{equation*}
       g(X_T) = \sum_{k=0}^{\infty} J_k \phi_k(X_T)  = \sum_{k=0}^{\infty} \sum_{\alpha \in \A_k} \Delta_{\mathds{1}_k} f(\boldsymbol{0}) (-2)^{-|\alpha|_0}   \calS_{\alpha}(X_T) = \sum_{\alpha }
   \Delta_{\alpha}f(\boldsymbol{0})    \calS_{\alpha}(X_T).
   \end{equation*}
\end{proposition} 

\begin{proof} 
As $g$ is path-independent, 
\cref{prop:Martingale} gives
$
g(X_T) = 
\sum_{k=0}^{\infty}\Delta_{\mathds{1}_k}f(\boldsymbol{0}) J_k(X_T)
$ and  $\Delta_{\mathds{1}_k} f(X)$ is an $\Q-$ martingale for all $k\ge 0$. As $X$ is Brownian motion, the functional It\^o formula \cite{Dupire} implies that each $\Delta_{\mathds{1}_k} f$ solves  the path-dependent PDE  
$\calL \varphi =  \left(\Delta_t + \frac{1}{2}\Delta_{xx} \right)\varphi=0.$ 
Put differently, any second spatial derivative can be converted into a temporal derivative by multiplying by the factor $-2$. Thus,  $\Delta_{\mathds{1}_k}f =  (-2)^{|\alpha|_0} \ \Delta_{\alpha} f$ for every word $\alpha \in \A_k$. 
The result follows from \cref{lem:hermite}. 
\end{proof}
\begin{remark}
In the general case where the kernels are time-dependent,  
it is delicate to compare the Wiener chaos expansion with the FTE. Indeed, we would need to expand the terms $J_k \phi_k$ which does not have the required regularity if expressed pathwise;  
see the discussion in \cite{Oberhauser}. Inversely, one can compute the Wiener chaos expansion of each signature functional, namely $\calS_{\alpha}(X_T) = \sum_{k=0}^{\infty} J_k \phi_k(X_T)$ for some  kernels  $(\phi_k)$ to be determined. In light of Stroock's formula $\eqref{eq:stroock}$, one has to compute the expected higher-order Malliavin derivatives of the signature to obtain $(\phi_k)$. 
This was conducted recently by \citet{Cass} for a wide class of Gaussian processes encompassing fractional Brownian motion 
with Hurst parameter $H\in(1/4,1)$. 
\end{remark}

\subsection{Other Dynamic Expansions}
\label{sec:otherExp}

\subsubsection{Hilbert Projections} \label{sec:Hilbert}

Another procedure to expand functionals 
consists of projecting 
its image seen as a collection of transformed paths. More precisely, let $f:\Lambda \to \R$, $X \in \Lambda_T$ and consider
$Z=f(X)$ given by $z_t = f(X_t)$, $t \in [0,T]$. If $Z\in \calH$ where $\calH$ is a  separable Hilbert space, then we can write $Z = \sum_{k = 1}^{\infty} (Z,F_k)_{\calH} F_k$ for an orthonormal basis (ONB) $\frakF = (F_k)_{k=1}^{\infty}$ of $\calH$. 
Put differently, 
\begin{equation}
    \label{eq:Hilbert}
    f(X_t) = \sum_{k = 1}^{\infty} \  (f(X),F_k)_{\calH} \ F_k(t). 
\end{equation}
Truncating $\eqref{eq:Hilbert}$ at some level  $K\ge 1$ gives the projection
$f^{K,\frakF}(X) = \sum_{k \le K} (f(X),F_k)_{\calH}  \ F_k.$ An illustration is given \cref{fig:KL}. 
One may wonder which ONB yields the best approximation$-$in some 
suitable sense$-$for a fixed truncation level. 
For instance, given  a measure $\Q$ on $\Lambda$ and  $\calH = L^2([0,T],dt)$, one can minimize the projection error, that is 
$$\min_{\frakF, \, \text{ONB}} \; \lVert f -f^{K,\frakF} \rVert_{L^2(\Lambda)} = \min_{\frakF, \, \text{ONB}} \; \lVert f(Y) -f^{K,\frakF}(Y) \rVert_{L^2(\Q\otimes dt)}, \quad K\ge 1. $$
The unique optimal ONB turns out to be, the same for all truncation level $K$ and  given by the eigenfunctions of the covariance kernel $\kappa(s,t) = \C^{\Q}(f(X_s),f(X_t))$. 
This is a classical result established by \citet{Karhunen} and \citet{Loeve} and employing the optimal basis in \eqref{eq:Hilbert} leads to the \textit{Karhunen-Loève (KL) expansion of $f$}.  

\begin{figure}[H]
	\centering
	\caption{Hilbert projection of functionals in the $(t,X,f(X))$ space. }
	\vspace{-2mm}
\includegraphics[height=2.2in,width=3in]{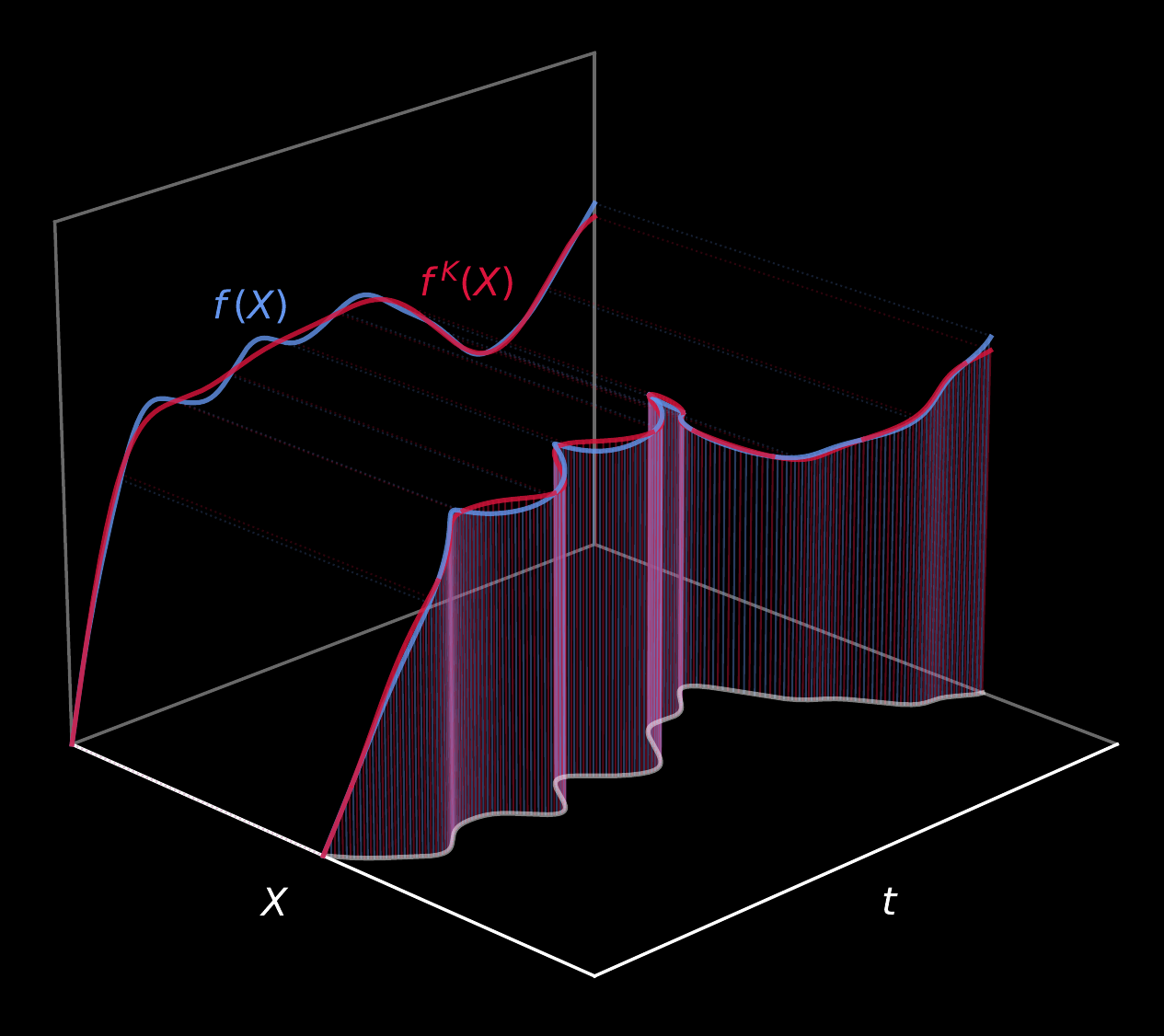}
\label{fig:KL}
\end{figure}

Among other financial applications, the KL expansion allows to
directly simulate the functional(s)  appearing in the payoff of exotic options \cite{Tissot}. 
Also, a market generator can be devised 
when $f$ represents 
the map linking the  random source (e.g. Brownian motion) to the stock price process. 
Indeed, the Fourier coefficient in \eqref{eq:Hilbert} can be "learned" from historical data permitting the simulation of new scenarios.  

\begin{remark}
If we are only given a $T-$functional $g \in L^2(\Lambda_T)$, one can consider an embedding $\iota: L^2(\Lambda_T) \hookrightarrow L^2(\Lambda)$ (such as $\iota_{\Q}$ defined in \cref{sec:TFuncEmbedd}) and apply the above methodology to $\iota g$. 
We thereafter
choose $t=T$ in \eqref{eq:Hilbert} to obtain the representation
$$g^{K,\frakF}(X_T) = \sum_{k\le K} (\iota g (X),F_k)_{L^2(dt)} F_k(T) .$$
\end{remark}

\subsubsection{Dynamic Wiener Chaos Expansion}\label{sec:dynamicChaos}
We refer the reader to the beginning of \cref{sec:chaos} for some notations 
pertaining to the Wiener chaos expansion.  
In particular, we consider a stochastic basis $(\Lambda,\calF,\F, \Q)$ and take $\Q$ to be  the Wiener measure throughout. 

It is well-known that the Wiener chaos expansion can be extended from random variables to stochastic processes \cite[Chapter 2]{DiNunno}. Indeed, if $Y $ 
is a measurable square integrable process 
(possibly non-adapted to the Brownian filtration), then one can apply  \Cref{thm:Wiener1} at each intermediate time.  
The chaos expansion for processes is of great importance in Malliavin calculus as it allows to define, among other things, the Skorokhod integral \cite{Skorokhod}. 
As our focus is on non-anticipative functionals, we henceforth suppose that $Y$ is adapted. 
An application of Doob's functional representation \cite[Lemma 1.13.]{kallenberg} shows that every  right-continuous process $Y$ adapted to the Brownian filtration  
is paired with a functional $f:\Lambda \to \R$ such that $y_t=f(X_t)\,$ $\forall \ t \in [0,T],$ $\Q-$a.s. We therefore expect that a "dynamic" Wiener chaos expansion is available, namely a chaos expansion for  functionals in $ L^2(\Lambda)$. This is confirmed in the next theorem.  

\begin{theorem}\label{thm:dynamicChaos}
\textnormal{(\textbf{Wiener Chaos Expansion for Functionals})}  \label{thm:chaosfunctional}
Every functional $f \in L^2(\Lambda)$ is associated with a unique sequence $\{\phi_k \in L^2(\triangle_{k,T})\}_{k \ge 1}$  such that for all $t\in [0,T]$
\begin{equation}
    \sum_{k= 0}^{K} J_k\phi_{k+1}(\cdot,t)(X_t) =   \sum_{k= 0}^{K} \int_{\triangle_{k,t}} \phi_{k+1}(t_1,\ldots,t_k, t)\, dx^{\otimes k}  \ \underset{L^2(\Lambda_t)}{\overset{K\uparrow \infty}{\xrightarrow{\hspace{1.2cm}}}} \ f(X_t).
\end{equation}  
\end{theorem}
\begin{proof}
See \cite[Chapter 2]{DiNunno}.
\end{proof}

\begin{example}
    Let $g$ be a square integrable $T-$functional with Wiener chaos expansion $g = \sum_{k\ge 0} J_k\phi_k$ and consider the price functional $f(X_t) = \iota_{\Q}g(X_t)= \E^{\Q}[g(Y_T) \ | \ X_t]$. Applying \cref{thm:chaosfunctional} to $f$, we obtain that 
    \begin{equation}
    f(X_t) =   \sum_{k= 0}^{\infty} \int_{\triangle_{k,t}} \tilde{\phi}_{k+1}(t_1,\ldots,t_k, t)\, dx^{\otimes k},
\end{equation} 
for some kernels $\{\tilde{\phi}_k \in L^2(\triangle_{k,T})\}_{k \ge 1}$. On the other hand, we have seen in \cref{ex:pricefunctional} that $f(X_t)$ coincides with truncated Wiener chaos of $g$  on $[0,t]$, namely  $f(X_t) = \sum_{k\ge 0} J_k\phi_k(X_t)$. We conclude from the uniqueness (in $L^2(\Q)$) of the Wiener kernels that $\tilde{\phi}_{k+1}(\cdot,t) = \phi_k |_{\triangle_{k,t}}$. 
\end{example}

\section{Applications}\label{sec:FTEApp}

\subsection{Pricing of Path-dependent Claims}
\label{app:pricing}
The functional Taylor expansion has obvious applications in the pricing of exotic options. 
In what follows, fix a risk-neutral measure  $\Q$ and assume zero interest rate. Moreover, the paths represent the evolution of the underlying stock price. 
If  $g\in L^1(\Lambda_T)$ represents the payoff of a path-dependent claim of European type, then its value (or price) with respect to $\Q$  is given by
$v_0 = \E^{\Q}[g(Y_T)]$. If $g$ can be nearly replicated 
by a  signature payoff, i.e.  $g \approx \sum_{\alpha \in \A'} c_{\alpha} \calS_{\alpha}$  where $\A'$ is a finite subset of $\A$, then
\begin{equation}\label{eq:sigPrice}
    v_0 \approx \sum_{\alpha \in \A'} c_{\alpha}  \E^{\Q}[\calS_{\alpha}(Y_T)]. 
\end{equation}
Thus, a pricing method can be devised  provided that both the coefficients $(c_{\alpha})$ and risk-neutral prices of the "primitive" securities $(\calS_{\alpha})$ are computed efficiently. 

This was carried out by \citet{Szpruch} and \citet{LyonsNum}, where  $(c_{\alpha})$ are  calculated by regressing realizations of the payoff against the signature functionals. 
The expected signature is either known explicitly \cite{Szpruch} or implied by the market \cite{LyonsNum}. 

If a functional Taylor (or Maclaurin) expansion was available for $g$, we would set $c_{\alpha} = \Delta_{\alpha}g(X_0)$ in view of  \cref{eq:sigPrice}. However, $g$ is a priori defined only for paths of length $T$ of its temporal derivative may not exist.   A remedy is to consider an embedding  $\iota: \frakF_T \hookrightarrow \frakF$ (see \cref{def:embedding}) with $\frakF_T,\frakF$ to be specified. In the context of pricing, we see that a natural choice  is $\frakF_K= L^1(\Lambda_T)$ and $\frakF =\bar{L}^1(\Lambda)$, the subspace of functionals $f:\Lambda \to \R$ such that $f(Y)$ is bounded in $L^1(\Q)$. 
Although other embeddings can be used, we  here choose the family of "Bachelier embeddings" $(\iota_{\Q_\sigma})_{\sigma \ge 0 }$ seen in \cref{sec:TFuncEmbedd}, namely $$\iota_{\Q_\sigma}   g(X_t)  =  \E^{\Q_\sigma}[g(Y_T) \,|\, X_t ],$$ 
where the scaled canonical process $Y/\sigma$ is $\Q_\sigma-$Brownian motion.  For all $\sigma\ge 0$, then   $\iota_{\Q_\sigma}   g \in \bar{L}^1(\Lambda)$  from Jensen's inequality and the tower property of expectations. Hence
$\iota_{\Q_\sigma}$ embeds $L^1(\Lambda_T)$  into $\bar{L}^1(\Lambda)$. 
If   $\iota_{{\Q_\sigma}}g\in \C^{K,K+1}$ for some $K\ge 1$, we can  choose $\A' = \{\alpha \in \A \ : \ |\alpha| \le K\}$ and apply the FTE to obtain 
\begin{align}
 v_0 = \E^{\Q}[g(Y_T)] = \E^{\Q}[\iota_{\Q_\sigma}g(Y_T)]  \approx \sum_{|\alpha| \le K}\underbrace{\Delta_{\alpha}\iota_{\Q_\sigma}g(X_0)}_{=:\ c_{\alpha,\sigma}} \E^{\Q}[\calS_{\alpha}(Y_T)]. 
 \label{eq:expTFuncPrice}
\end{align}
Provided that the signature claims are priced correctly (namely the "model" $\Q$ is calibrated to the  market), then \eqref{eq:expTFuncPrice} gives an explicit approximation for the value of $g$. We emphasize that the Bachelier model is solely used to compute the coefficients.  
    Note that each value of $\sigma$ will generate a distinct sequence of coefficients $(c_{\alpha,\sigma})$ and in turn a distinct approximation of the price. In view of  \cref{thm:Basis}, the coefficients will nevertheless be unique once  expressed in the basis $\{\calS_{\beta1}|_{\Omega_T^{\Pi}} \ : \ |\beta| < K \}$.  Because expectations with respect to an atomless measure 
tend to smooth out the integrand, we  
favor  $\sigma > 0$ to compute $(c_{\alpha,\sigma})$  over the degenerate case $\sigma =0$ associated to the intrinsic functional.
 
\begin{remark} \label{rmk:lkbk}
If $g$ is a path-independent payoff, i.e. $g(X_T) = h(x_T)$ for some $h:\R\to \R$, then $\iota_{\Q_{\sigma}} g(X_t) = \E^{\Q_{1}}[h(x_t + \sigma Y_{T-t})]$, $t<T$, is a smooth function of $x_t$. However,  the regularity of $\iota_{\Q_{\sigma}} g$ is no longer guaranteed for general path-dependent payoffs. As an illustration, consider a fixed strike,  at-the-money  lookback call option, i.e.  $g(X_T)= \max_{0 \le t \le T}x_t - x_0$. Assuming no interest rates, 
a simple calculation shows that $\iota_{\Q_\sigma}g(X_t) = u(t,x_t,\max_{0 \le s \le t}x_s)$, where 
\begin{align*}
  u(t,x,m) 
&= x - x_0 + (m -x)\left[2\Phi\left(\frac{m-x}{\sigma \sqrt{T-t}}\right)-1\right] + 2 \sigma \sqrt{T-t} \ \phi\left(\frac{m-x}{\sigma \sqrt{T-t}}\right),
  \end{align*}
where $\phi,\ \Phi$ is the PDF and CDF of the standard normal distribution, respectively. Although $u$ is clearly a smooth function in $m$, $\iota_{\Q_{\sigma}}g$ does not even belong to $\C^{0,1}(\Lambda)$ as the functional spatial derivative of $\max_{0 \le s \le t}x_s$ is not $\Lambda-$continuous.  
One way to sort out this issue is to regularize the running maximum as in  \cite[Example 2]{Dupire} and  compute  thereafter the needed functional derivatives in \eqref{eq:expTFuncPrice}.
\end{remark}

 \subsection{Static Hedging}
  The FTE possesses also  immediate applications 
 to static hedging. 
 Consider a claim $g$ with embedding  $f= \iota_{\Q_\sigma}g\in \C^{K,K+1}$, $\sigma>0$. Suppose we are given  a replicating portfolio $\varphi = \sum_{|\alpha|\le K}c_{\alpha}\calS_{\alpha}$ of $g$ such that $\Delta_{\alpha}f(X_0) = c_{\alpha}$ for all $|\alpha|\le K$. We shall see that the FTE allows to quantify the hedging error $|g-\varphi|$. First, applying \cref{cor:FME} to the difference $f-\varphi$ gives 
 \begin{equation}\label{eq:}
     f(X_t)-\varphi(X_t) =  R_{K}(X_t) =  \sum_{|\alpha| = K} \int_{\triangle_{K,t}} \Delta_{\alpha}f(X_{t_1}) \circ \, dx^{\alpha},  \end{equation}
 for any price path $X \in \Omega^{\Pi}$. If $g$ belongs to the space spanned by the signature functionals up to order $K$, we would obtain a perfect hedge, i.e. $R_{K} \equiv 0$. Otherwise, the hedging error can be bounded as we now explain.  
In line with financial markets, we assume that $X$ is the 
  piecewise interpolation of 
  tick data and is in particular Lipschitz continuous. 
  If also  $ \sup_{|\alpha| = K}|\Delta_{\alpha}f(X)| \le C^{K}$ for all $X\in \Omega^{\text{Lip}}$, $C< \infty $, then 
\cref{prop:remainderLip} gives
  $R_{K}(X_t) = R_{K}(X_0,X_t) = \calO(\frac{t^K}{K!})$.  
Choosing $t=T$, the  hedging error for the claim $g$ is therefore bounded by $$|g(X_T)-\varphi(X_T)| \le |R_{K}(X_T)| = \calO\left({\textstyle \frac{T^K}{K!}}\right), \quad \forall \ X_T\in \Omega^{\text{Lip}}_T. $$
We emphasize that the above upper bound holds pathwise, matching the  needs of exotic option traders to be protected against any future scenario. 

\section{Conclusion} \label{sec:conclusion}

This paper 
gathers and intends to 
elucidate 
expansions in the path space. 
In particular, we draw a distinction between static expansions 
(Volterra, Wiener, IVE) and dynamic ones$-$in particular the functional Taylor expansion (FTE)  lying at the heart of this work.  
We can summarize the latter as a pathwise tool combining the functional It\^o calculus and the signature to unravel path dependence. 
We then establish parallels 
between the FTE and static expansions such as the  intrinsic value expansion and the celebrated Wiener chaos. 
As seen in the applications, the FTE proves useful for the pricing and hedging of exotic claims as it effectively separates the payoff functional from the price path. For instance, this separation can be  incorporated into numerical methods such as cubature schemes to guarantee pricing accuracy and speed up computation. 

We deem the study of path functionals a necessity in modern quantitative finance, given the omnipresence of path dependence. As such, it seems that many works  can be undertaken in this direction. 
In particular, one can investigate financial applications  of  the Volterra and Wiener series, used  extensively 
in nonlinear systems but barely in finance.  
Another avenue would be to tackle non-Markovian problems with the FTE, such as the pricing of American options involving  path-dependent features in the payoff, the dynamics, or both. 


 
\subsection*{Acknowledgments}
We would like to thank Josef Teichmann, Harald Oberhauser, Harvey Stein  as well as the participants of the Research in Options conference (RiO 2021 and RiO 2022) and the 2022 CFMAR Workshop at UCSB for precious comments. We are also grateful to our colleagues  at Bloomberg (Bryan Liang,   Guixin Liu)  for fruitful discussions. 
\addtocontents{toc}{\protect\setcounter{tocdepth}{1}}
\appendix

\section{Proofs}\label{app:Proofs}

\subsection{\cref{thm:FSF}}\label{app:FSF}
\begin{proof} 
	 Recall that $X\in \Omega^{\Pi}_t$ implies $X_s \in \Omega^{\Pi}_s$ for all $s\le t$. We can therefore  show the result for $s = t$ without loss of generality.  
	For fixed $N$, write   $\delta x_{t_n} = x_{t_n}-x_{t_{n-1}}$, $\delta t_n = t_n- t_{n-1}$ with  $1\le n\le N$. Note that we omit the dependence of $\delta x_{t_n}, \ \delta t_n$ on $N$ for ease of presentation. 
	Let $X^N \in \bar{\Omega}^{\Pi}_{t}$ be the càdlàg piecewise constant approximation of $X$ along $\Pi^N$, i.e. $X^N = \sum_{t_n \in \Pi_t^N}  x_{t_{n-1}} \mathds{1}_{[t_{n-1},t_n)} + x_t \mathds{1}_{\{t\}}.$ Since $X^N_0 = X_0$, we can write 
	\begin{align*}
		f(X_t) - f(X_0) =   f(X_t) - f(X_t^{N}) +  \sum_{t_n\in \Pi_t^N} (f(X_{t_n}^{N})-f(X_{t_{n-1}}^{N})).
	\end{align*} 
	Clearly, $X^N \overset{N \uparrow \infty}{\longrightarrow} X$ uniformly, thus $f(X_t^N) \overset{N \uparrow \infty}{\longrightarrow} f(X_t)$ as well since $f$ is $\Lambda-$continuous. 
	Next, we decompose the summands in the above telescopic sum as 
	$$f(X_{t_n}^{N})-f(X_{t_{n-1}}^{N}) =\underbrace{f(X_{t_{n}-}^{N})-f(X_{t_{n-1}}^{N})}_{\text{flat extension}} + \underbrace{f(X_{t_n}^{N})-f(X_{t_{n}-}^{N})}_{\text{vertical bump}}. $$
	We thus retrieve from $X^N$ an alternation of flat extensions and bumps.  
As $f\in \C^{1,2}(\bar{\Omega}^{\Pi}_t)$,   the mean value theorem applied to $u\mapsto f(X^N_{t_{n-1},u})$, $h\mapsto f((X^N_{t_{n}-})^h)$  gives the existence of $(u_n,h^{+}_n) \in (0,\delta t_{n})\times (0,\delta{x}_{t_n})$ such that  
\begin{align}
	f(X_{t_{n}-}^{N})-f(X_{t_{n-1}}^{N})  &= \Delta_t f(X^N_{t_{n-1},u_n}) \delta t_n, \nonumber \\ f(X_{t_n}^{N})-f(X_{t_{n}-}^{N}) &= \Delta_x f(X_{t_{n}-}^{N}) \delta x_{t_n}  + \frac{1}{2}\Delta_{xx} f((X_{t_{n}-}^{N})^{h^{+}_n}) \delta x^2_{t_n}.\label{eq:MVTX1}
\end{align} 
Similarly, using the function  $h\mapsto f((X^N_{t_{n}})^{-h})$, there exists $h_n \in (0,\delta{x}_{t_n})$ such that  
\begin{equation}
	f(X_{t_n-}^{N})-f(X_{t_n}^{N}) = -\Delta_x f(X_{t_{n}}^{N}) \delta x_{t_n}  + \frac{1}{2}\Delta_{xx} f((X_{t_{n}}^{N})^{-h_n}) \delta x^2_{t_n}.\label{eq:MVTX2}
\end{equation} 
Hence, subtracting \eqref{eq:MVTX2} to \eqref{eq:MVTX1} yields
\begin{equation*}
	2(f(X_{t_n}^{N})-f(X_{t_{n}-}^{N})) = (\Delta_x f(X_{t_n-}^{N}) + \Delta_x  f(X_{t_n}^{N}))\delta x_{t_n} + \frac{\Delta_{xx} f((X_{t_n-}^{N})^{h_{n}^{+}})  - \Delta_{xx} f((X_{t_{n}}^{N})^{-h_n}) }{2}\delta x^2_{t_n}. 
\end{equation*}
Noticing that $(X_{t_{n}}^{N})^{-h_n} = (X_{t_n-}^{N})^{h_{n}^{-}}$ with $h_{n}^{-} := \delta x_{t_n} - h_n$, we obtain
\begin{align}
	\sum_{t_n\in \Pi_t^N} (f(X_{t_n}^{N})-f(X_{t_{n-1}}^{N})) 
	&= 
	\sum_{t_n\in \Pi_t^N} \Delta_t f(X^N_{t_{n-1},u_n}) \ \delta t_n \label{eq:riemann} \\
	&+ \sum_{t_n\in \Pi_t^N} \frac{\Delta_x f(X_{t_n-}^{N}) + \Delta_x  f(X_{t_n}^{N})}{2} \  \delta x_{t_n} \label{eq:strat} \\ 
	&+ \sum_{t_n\in \Pi_t^N} \frac{\Delta_{xx} f((X_{t_n-}^{N})^{h_{n}^{+}})  - \Delta_{xx} f((X_{t_n-}^{N})^{h_{n}^{-}}) }{4} \ \delta x^2_{t_n} \label{eq:remain}
\end{align} 
As $t\to t$ is a smooth integrator and $\Delta_t f$ is $\Lambda-$continuous,  
the right side of $\eqref{eq:riemann}$ converges to the Riemann integral
$$\lim_{| \Pi^N_t | \downarrow 0} \sum_{t_n\in \Pi_t^N} \Delta_t f(X^N_{t_{n-1},u_n}) \ \delta t_n = \lim_{| \Pi^N_t | \downarrow 0} \sum_{t_n\in \Pi_t^N} \Delta_t f(X^N_{t_{n-1}}) \ \delta t_n = \int_0^t \Delta_t f(X_{s}) ds.$$
Now consider the compact set $\calK_N = \{(X_{s}^{N})^{h} \ | \ (s,h) \in [0,t] \times [-\delta_N, \delta_N]\} \subset \Lambda$ with $\delta_N = \max_{n=1,...,N} |\delta x_{t_n}|$. 
Then $\Delta_{xx}f$ is uniformly $\Lambda-$continuous in $\calK_N$. In turn, the function $(s,h) \mapsto \Delta_{xx}f((X_{s}^{N})^{h})$, $(s,h) \in [0,t] \times [-\delta_N, \delta_N]$,  admits a modulus of continuity $\omega_N:\R_{+}\to \R_{+}.$ Thus, $$|\Delta_{xx} f((X_{t_n-}^{N})^{h_{n}^{+}})  - \Delta_{xx} f((X_{t_n-}^{N})^{h_{n}^{-}}) | \le \omega_N(|h_{n}^{+}-h_{n}^{-}|) \le \omega_N(\delta_N ) \quad \forall \ n = 1,\ldots,N. $$ 
As $\sum_{t_n\in \Pi_t^N} \delta x_{t_n}^2$ converges to $\langle X \rangle^{\Pi}_t < \infty$,  there exists $N_0 \in \N$ such that $\sum_{t_n\in \Pi_t^N} \delta x_{t_n}^2 \le 2\langle X \rangle^{\Pi}_t$ $\ \forall \ N\ge N_0$. 
Then $\eqref{eq:remain}$ is dominated by $\frac{1}{2}\omega_N(\delta_N)\ \langle X \rangle^{\Pi}_t$  for all $N\ge N_0$, 
and thus vanishes as $N \uparrow \infty$  since $X$ is uniformly continuous. 
Altogether, we have shown that 
\begin{align*}
	f(X_t) - f(X_0) &=
	\int_0^t \Delta_t f(X_{s}) ds +  \lim_{| \Pi^N_t | \downarrow 0} \sum_{t_n\in \Pi_t^N}  \frac{\Delta_x f(X^N_{t_{n}-}) + \Delta_x  f(X^N_{t_n})}{2} \  \delta x_{t_n}. 
\end{align*}
In fine, the  limit of $\eqref{eq:strat}$  exists and coincides with the Stratonovich integral $\int_0^t \Delta_x f(X_{s}) \circ dx_s$. 
\end{proof}
\subsection{\cref{prop:Sig}}\label{app:propSig}
 \begin{proof}
   We fix a path $X\in \Omega^{\Pi}_t$ and show the result by induction on $k = |\alpha| \ge 1$. If $k=1$, then clearly $\calS_0(X_t) = t$, $\calS_1(X_t) = x_t-x_0$ are well-defined smooth functionals. Now suppose that $\calS_{\beta}$ exists  for all $|\beta| < k$, $k\ge 2$. 
Pick any word $\alpha$ such that $k = |\alpha| > |\alpha|_1$. 
If  $\alpha_k=0$, then $\calS_{\alpha} = \int_0^t \calS_{\alpha -1}(X_s)ds$ is a well-defined Riemann integral. Moreover, $\Delta_t \calS_{\alpha} = \calS_{\alpha -1 }$, $\Delta_x \calS_{\alpha} \equiv 0 $.  The smoothness $\calS_{\alpha}$ follows by induction.
   
   Finally, we treat the case when  $\alpha_{k}=1$.  We would like to make sense of 
   $\calS_{\alpha}(X_t) := \int_0^t \calS_{\alpha_{-1}}(X_s)\circ dx_s$ by showing that the integrand is the spatial derivative of some smooth functional to apply \cref{thm:FSF}. In light of the discussion in \cref{sec:pathStrat} and the smoothness of $\calS_{\alpha-1}$, let us  introduce
   $\Phi(X_t) = \int_{x_0}^{x_t} \calS_{\alpha-1}(X_t^{(\varepsilon)}) \ d \varepsilon.$ 
   First, $\Delta_x \Phi = \calS_{\alpha-1} \in \C^{\infty,\infty}(\bar{\Omega}^{\Pi})$ by construction. If $\alpha_{k-1} = 1$, then $\Delta_t \Phi(X_t) = 0$ and in turn 
   $\Phi(X_t) = \int_0^t \calS_{\alpha-1}(X_s) \circ dx_s = \calS_{\alpha}(X_t)$ using the functional Stratonovich formula.  
   When $\alpha_{k-1} = 1$, then $\Delta_t \Phi(X_t) =  \int_{x_0}^{x_t} \calS_{\alpha-2}(X_t^{(\varepsilon)}) \ d \varepsilon$, so $\eqref{eq:sigDefStrat}$ reads
      $\calS_{\alpha}(X_t) = \Phi(X_t) - \int_{0}^{t}\int_{x_0}^{x_s} \calS_{\alpha-2}(X_s^{(\varepsilon)}) \ d \varepsilon ds. $
   Hence $\calS_{\alpha}$ is well-defined in both cases  and taking the spatial (respectively temporal) derivative 
   gives  $\Delta_{x}\calS_{\alpha} = \calS_{\alpha-1}$ (resp. $\Delta_{t}\calS_{\alpha} = 0$), which is $\eqref{eq:derivSig}$. 
   \end{proof}
\subsection{\cref{prop:sigProperties}}
\label{app:sigProperties}
\begin{proof} 
(i) This  follows from an iterative application of  $\eqref{eq:derivSig}$. 
     (ii) Suppose that there exists a 
finite subset 
$\A' \subseteq \A$ and $(c_{\alpha})_{\alpha \in \A'}$ such that 
$f := \sum_{\alpha \in \A'} c_{\alpha} \calS_{\alpha} \equiv 0. $ 
As $f$ is a combination of finitely many signature functionals, then $f$ is smooth using \cref{prop:Sig}. Moreover, assertion (i) implies that  $\Delta_{\alpha} \calS_{\beta}(X_0) = \mathds{1}_{\{\alpha = \beta\}}$ for all $\alpha,\beta \in \A$. We conclude that 
$c_{\alpha} = \Delta_{\alpha}f(X_0) = 0$ for all  $ \alpha \in \A'$  as desired. 
     (iii) There are many examples. A simple one follows from the product rule: define $g(X_T) = T (x_T-x_0)$, which is equal to $\calS_{01}(X_T)+\calS_{10}(X_T)$.  With $T$ being fixed, 
 it can be treated as a coefficient. 
 The linear dependence follows from the fact that $X_T \mapsto x_T-x_0$ is a signature functional, namely $\calS_1(X_T)$. 
     (iv) We start with a Lemma.
     

     

\begin{lemma} \label{lem:incremental}
Let $f \in \C(\Omega^{\Pi})$ such that $\calI \! f(X_T) : = \int_0^T f(X_t) \circ dx_t = 0$ on $\Omega^{\Pi}_T$. Then $f$ vanishes on $\Omega^{\Pi}$.
\end{lemma}

\begin{proof}
Take $X_t \in \Omega_t^{\Pi}$. Consider  the  family $(Y_T^{(h,\delta t)}) \subseteq \Omega_T^{\Pi}$ defined as 
$$Y_T^{(h,\delta t)}(s) =  x_{t\wedge s} + h \left(\frac{(t-s)^{+}}{\delta t } \wedge 1\right), \quad s \in [0,T], \quad (h,\delta t) \in \R_+ \times (0,\infty). $$
Note in particular that $Y_T^{(0,\delta t)} = X_{t,T-t}$. 
By assumption $\calI \!  f (X_{t,T-t}) = \calI \! f(Y_T^{(h,\delta t)}) = 0$. Hence, 
\begin{align*}
    0 =   \frac{1}{h}\left( \calI f(Y_T^{(h,\delta t)}) -  \calI f(X_{t,T-t}) \right)
   =  \frac{1}{\delta t}\int_t^{t+\delta t} f(Y_s^{(h,\delta t)}) \ ds \; \,\overset{h \ \downarrow \  0}{\xrightarrow{\hspace{0.8cm}}} \; \, \frac{1}{\delta t}\int_0^{\delta t} f(X_{t,u}) \ du.
\end{align*}
Letting $\delta t$ go to zero yields the result. 
\end{proof}

To prove (iv), suppose that there exists  
$\A' \subseteq \{ \calS_{\beta 1}|_{\Omega_T^{\Pi}} \ : \ \beta \in \A \}$, $|\A'| < \infty$, and $(c_{\alpha})_{\alpha \in \A'}$ such that $\sum_{\alpha\in \A'} c_{\alpha} \calS_{\alpha}(X_T) =0$ for all $X_T\in \Lambda_T$. As every word in $\A'$ ends with a "$1$", we can 
write 
$$\sum_{\alpha\in \A'} c_{\alpha} \calS_{\alpha}(X_T) = \sum_{\beta1 \in \A'} c_{\beta1} \calS_{\beta 1}(X_T) = \int_0^T \sum_{\beta1\in \A'} c_{\beta 1} \calS_{\beta}(X_t) \circ dx_t. $$
Applying \cref{lem:incremental} to $f = \sum_{\beta1\in \A'} c_{\beta 1} \calS_{\beta}$ and  conclude that $f \equiv 0$. 
Using (ii), this gives $c_{\alpha}=0$ for all $\alpha \in \A'$ as desired.
 \end{proof} 
 
\subsection{\cref{thm:Basis}}
\label{app:thmBasis}
 \begin{proof}
The independence is already shown in \cref{prop:sigProperties} (iv). We prove the spanning property by induction. The result is clearly true for $K=1$ as $\calS_0(X_T) = T$ can be regarded as a constant. 
Next, write $\B_K = \{\beta1 \ : \ |\beta| <K\}$ and take any word $\alpha$ of length at most $K\ge 2$ and outside of $\B_K$. If $|\alpha|<K$, then $\calS_{\alpha}(X_T)$ can be expressed as a linear combination of $(\calS_{\gamma}(X_T))_{\gamma\in \B_{K-1}}$ by induction. We can therefore  assume that 
$|\alpha| = K$ and $\alpha \notin \B_K$, that is  $\alpha = \beta 0$ with $|\beta| = K-1$. Invoking the induction hypothesis, there exist coefficients $(c_{\gamma})$ such that $\calS_{\beta}(X_T) = \sum_{\gamma1 \in \B_{K-1}} c_{\gamma}\calS_{\gamma 1}(X_T)$. We thus obtain using Fubini's theorem and the shuffle product \cite{Ree} that 
\begin{align*}
    \calS_\alpha (X_T) 
    &= \int_0^T \calS_{\beta}(X_t)dt \\
    &= \sum_{\gamma1 \in \B_{K-1}} c_{\gamma}\int_0^T \int_0^t\calS_{\gamma}(X_s) dx_s dt  \\ 
    &= \sum_{\gamma1 \in \B_{K-1}} c_{\gamma} \left( T \calS_{\gamma 1}(X_T) -  \int_0^T \calS_{1}(X_s) \calS_{\gamma}(X_s) dx_s  \right)\\
    &=  \sum_{\gamma1 \in \B_{K-1}} c_{\gamma} \left( T \calS_{\gamma 1}(X_T) - \sum_{j=1}^{K-1}  \calS_{\gamma_1\cdots \gamma_j 1 \gamma_{j+1} \cdots \gamma_{K-2} 1}(X_T)  \right).
\end{align*}
The result follows since  each signature term in the last expression belongs to $ \{\calS_{\beta 1}|_{\Omega_T^{\Pi}} :   |\beta| < K\}$.  
\end{proof}

\subsection{\cref{thm:FTE}}
\label{app:FTE}
\begin{proof} We carry out an induction on $K \ge 1$. We start with  $K=1$ (hence $f\in \C^{1,2}$) and define the functional $\tilde{f}(Y_r) = f(X_s \oplus Y_r)$, $r\le u$.  
As $Y_u \in \Omega^{\Pi}$, we can apply the functional Stratonovich formula (\Cref{thm:FSF}) to $\tilde{f}$, which gives  
\begin{align*}
    f(X_s \oplus Y_u) &= \tilde{f}(Y_u) \\
     &= \tilde{f}(Y_0) + \int_0^u \Delta_t \tilde{f}(X_r)  dr + \int_0^t\Delta_x \tilde{f}(Y_r) \circ dy_r\\
    &= f(X_s) + \int_0^u \Delta_t f(X_s \oplus Y_{r})  dr + \int_0^u\Delta_x f(X_s\oplus Y_{r}) \circ dy_r\\
        &= \Delta_{\emptyset}f(X_s) \calS_{\emptyset}(Y_u)+ \underbrace{\int_{\triangle_{1,u}} \Delta_t f(X_s\oplus Y_{t_1})  \circ  dy^{0} +\int_{\triangle_{1,u}} \Delta_x f(X_s\oplus Y_{t_1}) \circ  dy^{1}}_{= \, r_1(X_s,Y_u)}.
\end{align*}
Now take  $K\ge 2$ and suppose the result true for all $k\le K$. 
 As $f\in \C^{K,K+1}$ by assumption,  then trivially $f\in  \C^{K-1,K}$ and the induction hypothesis yields
\begin{align*}
    f(X_{s} \oplus Y_u) &= \sum_{|\alpha|< K-1}  \Delta_{\alpha}f(X_s) \calS_{\alpha}(Y_u) + r_{K-1}(X_s,Y_u).
\end{align*}
As  $\Delta_{\alpha}f$ is at least $\C^{1,2}$ for all $|\alpha| = K-1$, we can apply the pathwise Stratonovich formula to the integrands constituting the remainder functional. That is,
\begin{align*}
    r_{K-1}(X_s,Y_u) &= \sum_{|\alpha| = K-1} \int_{\triangle_{K-1,u}} \Delta_{\alpha}f(X_s \oplus Y_{t_1}) \circ \, dy^{\alpha}\\
    &= \sum_{|\alpha| = K-1}\int_{\triangle_{K,u}} \left[ \Delta_{\alpha}f(X_s)   + \int_0^{t_1} 
    \Delta_{0\alpha}f(X_s \oplus Y_{t_0})\, dt_0 + \int_0^{t_1} 
    \Delta_{1\alpha}f(X_s \oplus Y_{t_0})\circ \, dy_{t_0} \right] \circ  dy^{\alpha} \\
    &= \sum_{|\alpha| = K-1} \Delta_{\alpha}f(X_{s})\calS_{\alpha}(Y_u) \\
    &+ \underbrace{\sum_{\substack{|\alpha| = K \\ \alpha_1 =\,0}} \int_{\triangle_{K,u}} \Delta_{\alpha}f(X_s \oplus Y_{t_0}) \circ \, dy^{\alpha} 
    + \sum_{\substack{|\alpha| = K \\ \alpha_1 = \,1}}
    \int_{\triangle_{K,u}} \Delta_{\alpha}f(X_s \oplus Y_{t_0}) \circ \, dy^{\alpha}}_{= \, r_K(X_s,Y_u)}.
\end{align*}
Bundling the terms together completes the proof.  
\end{proof}

\subsection{\cref{cor:FTEX}}
\label{app:FTEX}
\begin{proof} Equation $\eqref{eq:CorFTE1}$ simply follows from   \Cref{thm:FTE} by setting  $u = t-s$ and  $Y_u = X |_{[s,t]} \in \Omega^{\Pi}$. Although $\eqref{eq:CorFTE2}$ is a  consequence of \Cref{thm:FTE} as well, we outline its proof as we deem it instructive. 
Indeed, some care is needed when expanding backward. 
For simplicity, write $W_u =X|_{[s,t]}$ so that  $\overleftarrow{W}=X|_{[t,s]} =: Y_u$. In particular, $X_r = X_s \oplus W_{r-s}$ for all $r\in [s,t]$. 
We now rearrange and  iterate  the pathwise Stratonovich formula  as follows: 
\begin{align*}
    f(X_s) 
    &= f(X_t) - \sum_{|\alpha|=1}\int_s^t \Delta_{\alpha} f(X_{t_1}) \circ  dx^{\alpha}\\ 
    &= f(X_t) - \sum_{|\alpha|=1}\int_{0}^{u} \Delta_{\alpha} f(X_{s} \oplus W_{\! t_1}) \circ  dw_{t_1}^{\alpha}\\ 
    &= f(X_t) - \sum_{|\alpha|=1}\Delta_{\alpha} f(X_t) \int_{\overleftarrow{\triangle}_{1,u} } \circ  dw^{\alpha} + \sum_{|\alpha|=2}\int_{\overleftarrow{\triangle}_{2,u} }  \Delta_{\alpha} f(X_s\oplus W_{t_1})  \circ  dw^{\alpha},
\end{align*}
with the time-reversed simplexes $\overleftarrow{\triangle}_{k,t}=  \{(t_1,\ldots,t_k)\in [0,t]^k\,|\, t_1 \ge \ldots \ge t_k\}$. 
Proceeding until $K$ yields 
\begin{align}\label{eq:intermediate}
    f(X_s) =  \sum_{|\alpha|<K} \Delta_{\alpha} f(X_t) \ (-1)^{|\alpha|} \ \int_{\overleftarrow{\triangle}_{|\alpha|,u} }   \circ \ dw^{\alpha} + (-1)^K \sum_{|\alpha|=K}  \int_{\overleftarrow{\triangle}_{K,u} }  \Delta_{\alpha} f(X_s \oplus W_{t_1})  \circ  dw^{\alpha}.
\end{align}
We now verify that 
\begin{equation}\label{eq:claim}
    \int_{\overleftarrow{\triangle}_{|\alpha|,u} } \varphi(X_s \oplus W_{t_1})  \circ  dw^{\alpha} = (-1)^{|\alpha|}  \int_{\triangle_{|\alpha|,u} } \varphi(X_t \oplus Y_{t_1})  \circ  dy^{\alpha},
\end{equation}
for every word $\alpha$ and functional $\varphi:\Lambda \to \R$  such that the above integrals are well-defined. 
Write  $k = |\alpha|$ and  consider the affine involution $T_u:[0,u]^k \to [0,u]^k$ given by $T_u(t_1,\ldots, t_k) = (u-t_1,\ldots,u-t_k)$. In particular,  $T_u(\overleftarrow{\triangle}_{k,u}) = \triangle_{k,u}$,  $w^{\alpha}= y^{\alpha} \circ T_u $ and $|\text{det}(\nabla T_u)| \equiv 1$.  Noticing also that $X_s \oplus W_{t_1} = X_t \oplus Y_{u-t_1}$, changing variables yields 
\begin{align*}
        \int_{\overleftarrow{\triangle}_{k,u} } \varphi(X_s \oplus W_{t_1})  \circ  dw^{\alpha} &=
          \int_{0}^u  \int_{t_k}^u \cdots \int_{t_2}^u   \varphi(X_t \oplus Y_{u-t_1}) \circ d \left(y^{\alpha} \circ T_u\right) \\ 
        &= \int_{u}^0  \int_{t_k}^{0} \cdots \int_{t_2}^0   \varphi(X_t \oplus Y_{t_1})  \circ  dy_{t_1}^{\alpha_1} \cdots \circ  dy_{t_k}^{\alpha_k}\\
          &=
        (-1)^{k}  \int_{\triangle_{k,u} } \varphi(X_t \oplus Y_{t_1})  \circ  dy^{\alpha}. 
\end{align*}

This proves the claim. In light of $\eqref{eq:intermediate}$, we choose  $\varphi \equiv 1$ and $\varphi=\Delta_{\alpha}f$ in $\eqref{eq:claim}$,  giving respectively the signature elements $\calS_{\alpha}(Y_u)$ and  terms in the remainder.   
\end{proof}

\subsection{\cref{prop:remainder}}
\label{app:remainder}
\begin{proof}
Fix $\alpha \in \A$ such that $|\alpha|=K$. Write $\alpha^{k} = \alpha_1\cdots\alpha_k$, $k\le K$, and 
$r_{\alpha,k}(X_t) = \int_{\triangle_{k,t}}\Delta_{\alpha}f(X_{t_1}) \circ dx^{\alpha^{k}}$. In particular, $r_{\alpha}  = r_{\alpha,K}.$ 
We prove by induction on $k$ that 
\begin{equation}\label{eq:inductionRemainder}
    |r_{\alpha,k}(X_s)| \le \frac{2^{|\alpha^k|_{01}}}{|\alpha^k|_0!}  c_{\alpha} \  t^{|\alpha^k|_0}\lVert X_t \rVert_{\infty}^{|\alpha^k|_1} \quad \forall \ s\le t, 
\end{equation}

so that $\eqref{eq:r_alpha}$ corresponds to $k=K$.  Without loss of generality, we can assume that $s=t$. 
\begin{itemize}
     \item \underline{$k=1$:} First note that $|\alpha^k|_{01}=0$. If $\alpha_1 = 0$, then using the norm $\lVert \cdot \rVert_{X_t}$ defined in $\eqref{eq:funcNorm}$, we obtain $$|r_{\alpha,1}(X_t)| \le \int_0^t |\Delta_{\alpha} f(X_s)|ds \le \lVert \Delta_{\alpha} f\rVert_{X_t}\  t \le  c_{\alpha}  \ t.$$ 
   If $\alpha_1 = 1$, using the expression $\eqref{eq:pathStrat}$ of the Stratonovich integral and  the notation $\int_{(0,x_t)} := \int_{x_t\wedge 0}^{x_t\vee 0}$, this gives 
   \begin{align*}
      |r_{\alpha,1}(X_t)| &\le \int_{(0,x_t)} |\Delta_{\alpha} f(X_t^{(\varepsilon)})| d\varepsilon + \int_{0}^{t} \int_{(0,x_s)}|\Delta_{t\alpha} f(X_s^{(\varepsilon)})| d\varepsilon ds \\
      &\le  \lVert \Delta_{\alpha}
      f\rVert_{X_t} \ |x_t|  + \lVert \Delta_{t\alpha} f\rVert_{X_t}\lVert X_t \rVert_{\infty} \ \ t  \ . \\
      &\le c_{\alpha} \lVert X_t \rVert_{\infty}.
   \end{align*}
  Hence $\eqref{eq:inductionRemainder}$ is satisfied for $k=1$. 
   
\item \underline{$k\ge 2$:} There are several cases. 
    \begin{itemize}
        \item \underline{$\alpha_k = 0$:} As $|\alpha^{k-1}|_0 = |\alpha^k|_0-1$ and $|\alpha^{k-1}|_{01}=|\alpha^k|_{01}$, the induction hypothesis gives
       \begin{align*}
      |r_{\alpha,k}(X_t)| &\le \int_{0}^{t} |r_{\alpha,k-1}(X_s)| ds \\
      &\le \frac{2^{|\alpha^k|_{01}}}{(|\alpha^k|_0-1)!}  c_{\alpha}\int_{0}^{t} s^{|\alpha^k|_0-1}\lVert X_s \rVert_{\infty}^{|\alpha^k|_1} ds\\
      &\le\frac{2^{|\alpha^k|_{01}}}{|\alpha^k|_0!}   c_{\alpha}  \ t^{{|\alpha^k|_0}}\lVert X_t \rVert_{\infty}^{|\alpha^k|_1}.
   \end{align*}
        \item \underline{$\alpha_k = 1, \ \alpha_{k-1}=1$:} Since $|\alpha^k-1|_0 = |\alpha^k|_0$ and $|\alpha^{k-1}|_{01}=|\alpha^k|_{01}$, we directly obtain 
         \begin{align*}
      |r_{\alpha,k}(X_t)| &\le \int_{(0,x_t)} |r_{\alpha,k-1}(X^{(\varepsilon)}_s)| d\varepsilon \\
     &\le \frac{2^{|\alpha^k|_{01}}}{|\alpha^k|_0!} c_{\alpha} \  \int_{(0,x_t)}  t^{|\alpha^k|_0} \lVert X_t^{(\varepsilon)} \rVert_{\infty}^{|\alpha^k|_1 -1} d\varepsilon \\
  &\le \frac{2^{|\alpha^k|_{01}}}{|\alpha^k|_0!} c_{\alpha} \  \ t^{{|\alpha^k|_0}}\lVert X_t \rVert_{\infty}^{|\alpha^k|_1}.
   \end{align*}
   In the last inequality, we used the fact that $\lVert X_t^{(\varepsilon)} \rVert_{\infty} \le \lVert X_t \rVert_{\infty}$ for all $\varepsilon \in [0,x_t]$. 
    \item \underline{$\alpha_k = 1, \ \alpha_{k-1}=0$:} This is the interesting case. First, observe that 
    $|\alpha^{k-2}|_0 = |\alpha^k|_0-1$ and  $|\alpha^{k-2}|_0 = |\alpha^k|_{01}-1$. Thus, 
 \begin{align*}
      |r_{\alpha,k}(X_t)| &\le \int_{(0,x_t)} |r_{\alpha,k-1}(X^{(\varepsilon)}_t)| d\varepsilon + \int_{0}^{t} \int_{(0,x_s)}|r_{\alpha,k - 2}(X_s^{(\varepsilon)})| d\varepsilon ds \\
      &\le 2  \int_{(0,x_t)} \int_{0}^{t}|r_{\alpha,\alpha^k - 2}(X_s^{(\varepsilon)})|  ds d\varepsilon \\
      &\le \frac{2^{|\alpha^k|_{01}}}{(|\alpha^k|_0-1)!}  c_{\alpha} \int_{0}^{t}s^{|\alpha^k|_0-1}  ds  \ \lVert X_t \rVert_{\infty}^{|\alpha^k|_1},
   \end{align*}
   which concludes the proof of $\eqref{eq:inductionRemainder}$. 
    \end{itemize}
\end{itemize}
\qquad \quad Finally, $\eqref{eq:remainder}$ follows by observing that $\frac{2^{|\alpha|_{01}}}{|\alpha|_0!} \le \frac{2^{|\alpha|_{0}}}{|\alpha|_0!}\le 2$. 
\end{proof}




\subsection{\cref{prop:Martingale}}
\label{app:Martingale}
\begin{proof} First,  $(iii) \Longrightarrow (i)$ is immediate as  It\^o iterated integrals only depends on the final value of the path; see  $\eqref{eq:Hermite}$. 
It remains to show that $(i) \Longrightarrow (ii)$ and $(ii) \Longrightarrow (iii)$.

\begin{enumerate}

    \item $(i) \Longrightarrow (ii)$ 
    
    The case $k=0$ is trivial so we proceed with $k=1$. Since $f(X_T)=g(X_T)=h(x_T)$ for some $h\in \calC^{\infty}(\R)$ and $X$ is Brownian motion, then 
\begin{equation}\label{eq:martingale}
    \Delta_xf(X_t) = \frac{d}{dx}\E^{\Q}[h(y_{T})\ | \ x_t] = \E^{\Q}[h'(x + y_{T-t})]\big |_{x=x_t} = \E^{\Q}[\Delta_{x}f(Y_T) \, | \, X_t]. 
\end{equation}
 For $k =2$, notice that $\tilde{g} :=\Delta_{x}f\big |_{\Lambda_T}$ is itself a smooth path-independent payoff. We can thus apply a similar argument to $\tilde{g}$ and $\tilde{f}(X_t)=\E^{\Q}[\tilde{g}(Y_T)\,|\,X_t]$, which is exactly $\Delta_{x}f(X_t)$ thanks to $\eqref{eq:martingale}$. Hence $\Delta_{xx}f(X_t) = \E^{\Q}[\Delta_{xx}f(Y_T) \, | \, X_t]$ and the same holds for higher 
derivatives.

\item $(ii) \Longrightarrow (iii)$

Recall from \Cref{prop:MRT_Chaos} that $g(X_T) = \sum_{k=0}^{\infty} J_k \phi_k (X_T)$ with
$$\phi_0 = \calM_T g(X_0) , \quad   \phi_k(t_1,...,t_k) = \E^{\Q}[\calM_{t_2...t_kT}g(Y_{t_1})].$$
If we show that  $\phi_k \equiv \Delta_{\mathds{1}_k}f(X_0)$   for each $k\ge 0$, then 
$J_k \phi_k (X_T) = \Delta_{\mathds{1}_k}f(X_0) J_k(X_T)$ as desired. 
We  prove by induction the slightly stronger claim, 
$$\E^{\Q}[(\calM_{t_2...t_kT}g)(Y_{t_1})\,|\, X_{t_0}] \equiv \Delta_{\mathds{1}_k}f(X_{t_0}), \quad (t_0,t_1,\ldots,t_k)\in \triangle_{k+1,T}.$$
If $k=0$, then obviously $\calM_T g(X_0) = f(X_0) = \Delta_{\emptyset} f(X_0) $. For $k\ge 1$,  we have
\begin{align*}
\E^{\Q}[(\calM_{t_2...t_kT}g)(Y_{t_1})\,|\, X_{t_0}]
    &= \E^{\Q}[\Delta_x \E^{\Q}[(\calM_{t_3...t_kT}g)(Y_{t_2}')\,|\, Y_{t_1}]\,|\, X_{t_0} ]\\
    &= \E^{\Q}[ \Delta_{\mathds{1}_k}f(Y_{t_1})\,|\, X_{t_0} ] \qquad && \text{(induction)}\\
    &=  \Delta_{\mathds{1}_k}f(X_{t_0}), &&\text{(martingality)}
\end{align*}
Taking $t_0=0$ yields the result.
\end{enumerate}
\end{proof}

\subsection{\cref{lem:hermite}}
\label{app:hermite}

\begin{proof}

Fix $k\ge 0$. The Hermite polynomials can be  written as  \begin{equation}\label{eq:exHermite}
H_k(x) = \sum_{l=0}^{\lfloor k/2 \rfloor} c_{k,l} \, x^{k-2l}, \quad c_{k,l} = \frac{k!}{l! (k-2l)!} (-2)^{-l}, 
\end{equation}
with coefficients $(c_{k,l})_{l=0}^k$ retrieved  for instance from   the relation 
$ H_{k+1}(x) = xH_{k}(x) - kH_{k-1}(x),\,   H_1(x)=x, \, H_0\equiv 1
$; see \cite[Section 5.4]{Szego}. Thanks to $\eqref{eq:exHermite}$, we have 
 \begin{equation}\label{eq:poly}
     J_k(X_t) = \frac{t^{k/2}}{k!} \sum_{l=0}^{\lfloor k/2 \rfloor} c_{k,l} \, \left(\frac{x_t}{\sqrt{t}}\right)^{k-2l}  = \sum_{2l_0+l_1=k} (-2)^{-l_0}  \, \frac{t^{l_0} }{l_{0}!}\frac{ x_t^{l_1}}{ l_{1}!}. 
 \end{equation}
 Thus,  $J_k$ is a polynomial in $t$ and $x_t$. 
To conclude, we invoke the following identity,
 \begin{equation*}
     \sum_{|\alpha|_0=l_0, \, |\alpha|_1=l_1 } \calS_{\alpha}(X_t) = \calS_{\boldsymbol{0}_{l_0}}(X_t)\, \calS_{\mathds{1}_{l_1}}(X_t) = \frac{t^{l_0} }{l_{0}!}\frac{ x_t^{l_1}}{ l_{1}!},
 \end{equation*}
 which is  either seen as a consequence  of \cite[Proposition 5.2.10]{KP}  or a particular case of the shuffle product \cite{Ree}. Since $\bigcup\limits_{2l_0+l_1=k} \{\alpha\, |\,\, |\alpha|_0 =l_0, \, |\alpha|_1 =l_1\} =  \A_k$ with $\A_k$ in the statement, the result follows. 
\end{proof}

\bibliographystyle{abbrvnat}
\bibliography{main.bib}
\end{document}